\documentclass[12pt,reqno]{article}

\usepackage[latin9]{inputenc}
\usepackage[a4paper]{geometry}
\usepackage{color}
\usepackage{amsmath}
\usepackage{amsthm}
\usepackage{amssymb}
\usepackage{setspace}
\usepackage[authoryear]{natbib}
\usepackage[unicode=true,
 bookmarks=true,bookmarksnumbered=true,bookmarksopen=true,bookmarksopenlevel=3,
 breaklinks=false,pdfborder={0 0 1},backref=page,colorlinks=true]
 {hyperref}
\hypersetup{
 hypertexnames=false,pdfborderstyle=,pdfborderstyle={},pdfborderstyle={},pdfborderstyle={},pdfborderstyle={},pdfborderstyle={},urlcolor=chicago-maroon,linkcolor=chicago-maroon,citecolor=northwestern-purple}

\makeatletter

\providecommand{\tabularnewline}{\\}

\theoremstyle{plain}
\newtheorem{thm}{\protect\theoremname}
\theoremstyle{definition}
\newtheorem{defn}{\protect\definitionname}
\theoremstyle{plain}
\newtheorem{cor}{\protect\corollaryname}
\theoremstyle{plain}
\newtheorem{prop}{\protect\propositionname}
\theoremstyle{plain}
\newtheorem{lem}{\protect\lemmaname}
\theoremstyle{remark}

\newtheorem{claim}{\protect\claimname}

\definecolor{clemson-orange}{RGB}{234,106,32}\definecolor{chicago-maroon}{RGB}{128,0,0}\definecolor{northwestern-purple}{RGB}{82,0,99}\definecolor{cornell-red}{RGB}{179,27,27}\definecolor{sauder-green}{RGB}{171,180,0}\definecolor{gray}{RGB}{192,192,192}\definecolor{lawngreen}{RGB}{0,250,154}\usepackage{graphicx}
\usepackage{pgfpages}\usepackage{epsfig}\usepackage{multirow}\usepackage{lscape}\usepackage{epic}
\usepackage{tikz}
\usetikzlibrary{arrows,arrows.meta,decorations,decorations.pathreplacing,calc,matrix,patterns}
\definecolor{Red}{rgb}{1,0,0}
\definecolor{Blue}{rgb}{0,0,1}
\definecolor{Green}{rgb}{0,1,0}
\definecolor{magenta}{rgb}{1,0,.6}
\definecolor{lightblue}{rgb}{0,.5,1}
\definecolor{lightpurple}{rgb}{.6,.4,1}
\definecolor{gold}{rgb}{.6,.5,0}
\definecolor{orange}{rgb}{1,0.4,0}
\definecolor{hotpink}{rgb}{1,0,0.5}
\definecolor{newcolor2}{rgb}{.5,.3,.5}
\definecolor{newcolor}{rgb}{0,.3,1}
\definecolor{newcolor3}{rgb}{1,0,.35}
\definecolor{darkgreen1}{rgb}{0, .35, 0}
\definecolor{darkgreen}{rgb}{0, .6, 0}
\definecolor{darkred}{rgb}{.75,0,0}
\definecolor{lightgrey}{rgb}{.7,.7,.7}
\usepackage{paralist}

\usepackage[normalem]{ulem}

\allowdisplaybreaks[4]

\renewenvironment{itemize}[1]{\begin{compactitem}#1}{\end{compactitem}}
\renewenvironment{enumerate}[1]{\begin{compactenum}#1}{\end{compactenum}}

\providecommand{\corollaryname}{Corollary}
\providecommand{\definitionname}{Definition}
\providecommand{\lemmaname}{Lemma}
\providecommand{\propositionname}{Proposition}
\providecommand{\remarkname}{Remark}
\providecommand{\theoremname}{Theorem}

\providecommand{\claimname}{Claim}

\usepackage[title]{appendix}

\makeatother

\providecommand{\claimname}{Claim}
\providecommand{\corollaryname}{Corollary}
\providecommand{\definitionname}{Definition}
\providecommand{\lemmaname}{Lemma}
\providecommand{\propositionname}{Proposition}
\providecommand{\theoremname}{Theorem}

\begin{document}
\global\long\def\pl{\underline{p}}%
\global\long\def\ph{\overline{p}}%
\global\long\def\pil{\underline{\pi}}%
\global\long\def\pih{\overline{\pi}}%
\global\long\def\Vh{\overline{V}}%
\global\long\def\Uh{\overline{U}}%

\title{\textbf{Keeping the Listener Engaged: a Dynamic Model of Bayesian
Persuasion}\thanks{We thank  {Emir Kamenica and four anonymous referees for many insightful and constructive suggestions. We are also grateful to }Martin Cripps, Jeff Ely, Faruk Gul,  Stephan
Lauermann, George Mailath, Meg Meyer, Sven Rady, Nikita Roketskiy,
Hamid Sabourian, Larry Samuelson, Sara Shahanaghi, and audiences in
various seminars and conferences for helpful comments and discussions.
Yeon-Koo Che is supported by National Science Foundation (SES-1851821);
he and Kyungmin Kim are supported by the Ministry of Education of
the Republic of Korea and the National Research Foundation of Korea
(NRF-2020S1A5A2A03043516). }}
\author{Yeon-Koo Che\qquad{}Kyungmin Kim\qquad{}Konrad Mierendorff\thanks{Che: Department of Economics, Columbia University (email: yeonkooche@gmail.com);
Kim: Department of Economics, Emory University (email: kyungmin.kim@emory.edu);
Mierendorff: Department of Economics, University College London.
Our dear friend and coauthor, Konrad Mierendorff, passed away in August 2021. All the ideas and results in this paper are collaborative work by three authors, but Che and Kim are responsible for any remaining errors.}}
\maketitle
\begin{abstract}
\noindent We consider a dynamic model of Bayesian persuasion in which
information takes time and is costly for the sender to generate and
for the receiver to process, and neither player can commit to their
future actions. Persuasion may totally collapse in a Markov perfect
equilibrium (MPE) of this game. However, for persuasion costs sufficiently small, a version of a folk theorem holds: outcomes that approximate \citet{kamenica/gentzkow:11}'s sender-optimal persuasion as well
as full revelation and everything in between are obtained in MPE,
as the cost vanishes.

\noindent \smallskip{}

\end{abstract}

\section{Introduction}

Persuasion is a quintessential form of communication in which one
individual (the sender) pitches an idea, a product, a political candidate,
a point of view, or a course of action, to another individual (the
receiver). Whether the receiver ultimately accepts that pitch---or
is ``persuaded''---depends on the underlying truth (the state of
the world) but importantly, also on the information the sender manages
to communicate. In remarkable elegance and generality, \citet[henceforth KG]{kamenica/gentzkow:11} show how the sender should communicate information in such a setting,
when she can perform \textit{any} (Blackwell) experiment \textit{instantaneously},
\textit{without any cost incurred} by her or by the receiver. This frictionlessness gives full commitment power to the sender, as she can publicly choose any experiment and reveal its outcome, all before the receiver can act.

In practice, however, persuasion is rarely frictionless. Imagine a
salesperson pitching a product to a potential buyer. The buyer may
have an interest in buying the product but requires some evidence
that it matches his needs. To convince the buyer, the salesperson
might demonstrate certain features of the product, or marshal customer
testimonies and sales records, any of which takes real time and effort.
Likewise, to process information, the buyer must pay attention, which
is costly.  Clearly, these features are present in other persuasion contexts, such as a prosecutor seeking to convince juries or a politician trying to persuade voters.

In this paper, we study the implications of these realistic frictions. Importantly, if information takes time to generate {but the receiver can act at any time}, the sender no longer \emph{automatically} enjoys full commitment power. Specifically, she cannot promise to the receiver what experiments she will perform in the future, effectively reducing her commitment
power to a current ``flow'' experiment. Given the lack of commitment
by the sender, the receiver may stop listening and take an action if he does not believe that the sender's future experiments are worth waiting for. The buyer in the example above may walk away at any time when he becomes sufficiently pessimistic about the product or about the prospect of the salesperson eventually persuading him.
We will examine \textit{to what extent} and \textit{in what manner}
the sender can persuade the receiver in this environment with limited
commitment. As we will demonstrate, the key challenge facing the sender
is to instill the belief that she is worth listening to, namely, to
\textit{keep the receiver engaged}.

We develop a dynamic version of the canonical persuasion model: the
state is binary, $L$ or $R$, and the receiver can take a binary
action, $\ell$ or $r$. The receiver prefers to match the state by
taking action $\ell$ in state $L$ and $r$ in state $R$, while
the sender prefers action $r$ regardless of the state. Time is continuous
and the horizon is infinite. At each point in time, unless the game
has ended, the sender may perform some ``flow'' experiment. In response,
the receiver either takes an action and ends the game, or simply waits
and continues the game. Both the sender's choice of experiment and
its outcome are publicly observable. Therefore, the two players always
share a common belief about the state.

The sender has a rich class of Poisson experiments at her disposal.
Specifically, we assume that at each instant the sender can generate
a collection of Poisson signals. The possible signals are flexible
in their \emph{directionalities}: a signal can be either good-news
(inducing a posterior above the current belief), or bad-news (inducing
a posterior below the current belief). In addition, the news can be of arbitrary
\emph{accuracy}: the sender can choose any target posterior, although
more accurate signals (with targets closer to $0$ or $1$) arrive
at a lower rate. Our model generalizes the existing Poisson models
in the literature which considered either a good-news or bad-news
Poisson experiment of given accuracy \citep[e.g.,][]{keller:05,keller2015breakdowns,CM2019}.

Any experiment, regardless of its accuracy, requires a flow cost $c>0$
(per unit of time) for the sender to perform and for the receiver
to process.  That the cost is the same for both players is a convenient normalization, with no material consequence (see Footnote \ref{fn:asymmetric_costs}). Our model of information allows for the flexibility and richness of \citet{kamenica/gentzkow:11}, but adds the friction that
information takes time to generate. This serves to isolate the effects
of the friction.

We may interpret the model in the canonical communication context, such as a salesperson pitching a product to a buyer.  The former is trying to persuade the latter that the product fits his  needs, an event denoted by $R$. Once inside the store, the buyer is deciding whether to listen to the pitch (wait), leave the store (action $\ell$), or  purchase the product (action $r$). We interpret the series of pitches made by the salesperson as experiments. A salesperson's pitches may include the types of product features demonstrated as well as her manner, tones, and body languages with which her messages are delivered. Hence, the pitches can reveal a lot about what she is ``intending'' to say, not just what she is saying, consistent with public observability of experiments assumed in our model. Meanwhile, whether the pitches succeed or not depends on the buyer's idiosyncratic needs, and is uncertain from the salesperson's perspective. It is also reasonable that an experienced salesperson could get feedback on her pitch directly or indirectly from the buyer's reactions, which would make the outcome of the experiment public.  As in our model, the key issue is whether the buyer believes the salesperson's pitches to be worth listening to. Our analysis will focus on this issue.

We study \emph{Markov perfect equilibria (MPE)} of this game, that
is, subgame perfect equilibrium strategy profiles that prescribe the
sender's flow experiment and the receiver's action ($\ell,r,$ or
``wait'') at each belief $p$---the probability that the state
is $R$. We are particularly interested in the equilibrium outcomes
when the frictions are sufficiently small (i.e., in the limit as the
flow cost $c$ converges to zero). In addition, we investigate the
\textit{persuasion dynamics} or the ``type of pitch'' the sender
uses to persuade the receiver in equilibria of this game.

\paragraph*{Is persuasion possible? If so, to what extent?}

Whether the sender can persuade the receiver depends on
whether the receiver finds her worth listening to, or more precisely,
on his belief that the sender will provide enough information to justify
his listening costs. This belief depends on the sender's future experimentation
strategy, which in turn rests on what the receiver will do if the
sender betrays her trust and reneges on her information provision.
The multitude of ways in which the players can coordinate on these
choices yields a folk-theorem-like result. There is an MPE in which no persuasion occurs. When the cost $c$ becomes arbitrarily small,
however, we also obtain a set of ``persuasion'' equilibria that
ranges from ones that approximate \citet{kamenica/gentzkow:11}'s
sender-optimal persuasion to ones that approximate full revelation; we show that any sender (receiver) payoff between these two extremes is attainable in the limit as $c$ tends to $0$.

In the ``persuasion failure'' equilibrium, the receiver is pessimistic
about the sender generating sufficient information, so he simply takes
an action without waiting for information. Facing this pessimism,
the sender becomes desperate and maximizes her chance of once-and-for-all
persuasion involving minimal information, which turns out to be the
sort of strategy that the receiver would not find worth waiting for,
justifying his pessimism.

In a persuasion equilibrium, by contrast, the receiver expects the
sender to deliver sufficient information to compensate his listening
costs. This optimism in turn motivates the sender to deliver on her
``promise'' of informative experimentation; if she reneges on her
experimentation, the ever optimistic receiver would simply wait for
experimentation to resume an instant later, instead of taking the
action that the sender would like him to take. In short, the receiver's
optimism fosters the sender's generosity in information provision,
which in turn justifies this optimism. As we will show, equilibria
with this ``virtuous cycle'' of beliefs can support  a wide range of outcomes from KG's optimal persuasion to full revelation, as the flow cost $c$ tends to 0.\footnote{The mechanism using a virtuous cycle of beliefs to support cooperative
behavior in a dynamic environment has been utilized in other economic
contexts. Among others, \citet{che2004dynamic} show how this mechanism
can be used to overcome the hold-up problem. In fact, the main tension
in our dynamic persuasion problem can be interpreted as a hold-up
problem: the receiver wants to avoid incurring listening costs if
the sender will behave opportunistically and not provide sufficient
information. However, the current paper differs in other crucial aspects;
in particular, the rich choice of information structures is unique
here and has no analogue in \citet{che2004dynamic}.}

\paragraph*{Persuasion dynamics.}

Our model informs us of what kind of pitch the sender should make at
each point in time, how long it takes for the sender to persuade the
receiver, if ever, and how long the receiver listens to the sender
before taking an action. The dynamics of the persuasion strategy adopted
in equilibrium unpacks rich behavioral implications that are absent
in the static persuasion model.

In our MPEs, the sender optimally makes use of the following three
strategies: (i) \textit{confidence-building}, (ii) \textit{confidence-spending},
and (iii) \textit{confidence-preserving}. The \textit{confidence-building}
strategy involves a bad-news Poisson experiment that induces the receiver's
belief (that the state is $R$) to either drift upward or jump to
zero.   Under this strategy, the belief moves upward for sure when the state is $R$ and quite likely even when the state is $L$; in fact, this strategy  minimizes
the probability of bad news, by insisting that the news be conclusive. The sender finds it optimal to use this strategy when the receiver's belief is already close to the persuasion target (i.e., the belief that will trigger him to choose $r$).

The \textit{confidence-spending} strategy involves a good-news Poisson
experiment that generates an upward jump to some target belief, either
one inducing the receiver to choose $r$, or at least one inducing
him to listen to the sender. Such a jump arises rarely, however, and
absent this jump, the receiver's belief drifts downward. In this sense,
this strategy is a risky one that ``spends'' the receiver's confidence
over time. This strategy is used when the receiver is already quite pessimistic about $R$, so that either the confidence-building strategy would take too long, or the receiver would simply not listen. In particular, it is used as a ``last ditch'' effort, when the sender is close
to giving up on persuasion or when the receiver is about to choose $\ell$.

The \textit{confidence-preserving} strategy combines the above two
strategies---namely, a good-news Poisson experiment inducing the
belief to jump to a persuasion target, and a bad-news Poisson experiment
inducing the belief to jump to zero. This strategy is effective if
the receiver is sufficiently skeptical relative to the persuasion
target so that the confidence-building strategy will take too long.
Confidence spending could expedite persuasion for a range of beliefs but would run down the receiver's confidence in the process. Hence, at some point
the sender finds it optimal to switch to the confidence-preserving
strategy, which prevents the receiver's belief from deteriorating
further. The belief where the sender switches to this strategy constitutes
an absorbing point of the belief dynamics; from then on, the belief
does not move, unless either a sudden persuasion breakthrough or breakdown
occurs.

The equilibrium strategy of the sender combines these three strategies
in different ways under different economic conditions, thereby exhibiting
rich and novel persuasion dynamics. Our characterization in Section
\ref{sec:persuasion_dynamics} describes precisely how the sender
uses them in different equilibria.

\paragraph*{Related literature.}

 This paper primarily contributes
to the Bayesian persuasion literature that began with \citet{kamenica/gentzkow:11},
by studying the problem in a dynamic environment. Several recent papers  also consider dynamic models \citep[e.g.,][]{brocas2007influence,kremer/mansour/perry:14,au:15,ely:17,renault/solan/vieille:17,che/horner:2018,Henry2017,ely/szydlowski:2017,bizotto/rudiger/vigier:16,orlov2020persuading,marinovic2020monitor}. Our focus is different from most of these papers since we consider gradual production of information \emph{and} assume that there is no commitment.\footnote{\citet{orlov2020persuading} characterize an equilibrium that resembles some aspects of our equilibrium in a model where the sender (agent) faces no constraint in the release of information. In particular, they show that the sender may ``pipet'' information---release information gradually---in a way that resembles our {confidence-building} ($R$-drifting) strategy. The resemblance is more apparent than fundamental, however. In their main model, the sender intrinsically prefers the receiver to delay exercise of a real option; that is, the delay of the receiver's action per se is desired by the sender. She can fully reveal the state instantaneously but chooses to delay release of information in order to incentivize the receiver to wait longer. In our model, the sender has no intrinsic preferences for delay and provides information only to persuade the receiver to take a particular final action.}

Two papers closest to ours in this regard are \citet{brocas2007influence}
and \citet{Henry2017}, who restrict the set of feasible experiments
so that information arrives gradually. The former considers a binary
signal in a discrete-time setting, and the latter employs a drift-diffusion
model in a continuous-time setting.\footnote{\citet{McClellan2017} and \citet{escudeludvig} also study dynamic persuasion in drift diffusion models. \citet{McClellan2017} characterizes the optimal dynamic approval mechanism under full commitment. \citet{escudeludvig} consider a sender dynamically optimizing against a receiver who chooses a series of actions myopically.} Unlike our model, the receiver
in their models cannot stop listening and take an action at any time:
he can move only after the sender stops experimenting \citep{brocas2007influence}
or applies for approval \citep{Henry2017}. This modeling difference
reflects interests in different economic problems/contexts; for example,
\citet{Henry2017} focus on regulatory approval, while we study persuasive
\emph{communication}. However,  the difference leads to very different persuasion outcomes: in their models, complete persuasion failure never occurs,
and there exists a unique equilibrium.\footnote{\label{fn:HO}\citet{Henry2017} consider three regimes that differ in the players' commitment power. Their informer-authority regime
corresponds to the sender-optimal dynamic outcome, in that the sender
stops as soon as the belief reaches the minimal point at which the
receiver is willing to take action $r$ (approves the project). It
is easy to show that in this case, if the receiver could reject/accept
the project unilaterally at any time, and discounted his future payoff
or incurred a flow cost as in our model, he would take an action immediately
without listening, and persuasion {would} fail completely. Their
``no-commitment'' regime is similar to our model, but with the crucial
difference that the sender does not have the option to ``pass,"
that is, to stop experimenting without abandoning the project. This
feature allows the receiver (e.g., a drug approver) to force the sender
to keep experimenting, resulting in the ``receiver-optimal'' persuasion as the unique equilibrium outcome. If ``passing'' were an available option as we assume in our model, multiple equilibria supported by ``virtuous cycles'' of beliefs would arise even in their drift-diffusion model, producing a range of persuasion outcomes and ultimately leading to the same kind of result as our Theorem \ref{thm:folk} (see Footnote \ref{fn:Brownian-Motion} below). Finally, their evaluator-authority case is obtained when the receiver can commit to an acceptance threshold.} Another important difference is that the sender in their models does not enjoy the richness and control of information structures: in both papers, the sender decides simply whether to continue or not, and has no influence over the type of information generated.

The receiver's problem in our paper involves a stopping problem, which has been studied extensively in the single agent context, beginning with \citet{wald:47} and \citet{arrow/blackwell/girshick:49}. In particular, \citet{nikandrova/pancs:15}, \citet{CM2019} and \citet{Mayskaya2016}
study an agent's stopping problem when she acquires information through Poisson experiments.\footnote{The Wald stopping problem has also been studied with drift-diffusion learning \citep[e.g.,][]{moscarini/smith:01,Ke2016,fudenberg:18},
and in a model that allows for general endogenous experimentation
\citep[see][]{Zhong2022}.} \citet{CM2019} introduced the general class of Poisson experiments
adopted in this paper. However, the generality is irrelevant in their
model, because unlike here, the decision maker optimally chooses only
between two conclusive experiments (i.e., never chooses a non-conclusive
experiment).

The  paper is organized as follows. Section \ref{sec:model} introduces
the model. Section \ref{sec:illustate} illustrates the main ideas of our equilibria. Sections \ref{sec:folk} and \ref{sec:persuasion_dynamics} characterize our MPE strategies and study their payoff implications. Section \ref{sec:conclusion} concludes.

\section{Model}\label{sec:model}

We consider a game in which a \emph{Sender} (``she'') wishes to
persuade a \emph{Receiver} (``he''). There is an unknown state $\omega$
which can be either $L$ (``left'') or $R$ (``right''). The receiver
ultimately takes a binary action $\ell$ or $r$, which yields the
following payoffs:
\begin{center}
\begin{tabular}{ccc}
\multicolumn{3}{c}{\textbf{Payoffs for the sender and the receiver}}\tabularnewline
\hline
states/actions  & $\ell$  & $r$ \tabularnewline
\hline
$L$  & $(0,u_{\ell}^{L})$  & $(v,u_{r}^{L})$ \tabularnewline
$R$  & $(0,u_{\ell}^{R})$  & $(v,{u_{r}^{R}})$ \tabularnewline
\hline
\end{tabular}
\par\end{center}

\noindent The receiver gets $u_{a}^{\omega}$ if he takes action $a\in\{\ell,r\}$
when the state is $\omega\in\{L,R\}$. The sender's payoff depends
only on the receiver's action: she gets $v$ if the receiver takes
$r$ and zero otherwise. We assume $u_{\ell}^{L}>\max\{u_{r}^{L},0\}$
and $u_{r}^{R}>\max\{u_{\ell}^{R},0\}$, so that the receiver prefers
to match the action with the state, and also $v>0$, so that the sender
prefers action $r$ to action $\ell$. Both players begin with a common
prior $p_{0}$ that the state is $R$, and use Bayes rule to update
their beliefs.

\begin{figure}
\raggedright{}\centering{}\beginpgfgraphicnamed{KGexplain}
    \begin{tikzpicture}[scale=1]

\draw[line width=0.5pt] (0, 4.3) -- (0,0) -- (6,0) -- (6,4.3);

\draw[dashed,line width=0.5pt] (0,0) -- (3,4) -- (6,4);

\draw[loosely dotted] (3,0)--(3,4);
\draw[loosely dotted] (1.5,0)--(1.5,2);
\draw[loosely dotted] (0,4)--(3,4);

\draw[dashdotted,line width=1pt] (0,0)--(6,4);

\fill (0,0) node[below] {\footnotesize{$0$}};
\fill (1.5,-0.1) node[below] {\footnotesize{$p_{0}$}};
\fill (3,0) node[below] {\footnotesize{$\hat p$}};
\fill (6,0) node[below] {\footnotesize{$1$}};

\fill (0,4) node[left] {\footnotesize{$v$}};

\draw[line width=0.5pt] (0,0) -- (3,0);
\draw[line width=0.5pt] (3,4) -- (6,4);

\fill (canvas cs:x=0cm,y=0cm) circle (2pt);
\fill (canvas cs:x=1.5cm,y=2cm) circle (2pt);
\fill (canvas cs:x=3cm,y=4cm) circle (2pt);
\fill (canvas cs:x=6cm,y=4cm) circle (2pt);

\draw (3,4.6) node {\textbf{Sender}};

\draw[xshift=7.5cm] (0, 4.3) -- (0,0) -- (6,0) -- (6,4.3);

\draw[xshift=7.5cm,line width=0.5pt] (0,3.5) -- (3,2) -- (6,4);

\draw[xshift=7.5cm,dashed,color=gray](0,0)--(3,2)--(6,0.5);

\draw[xshift=7.5cm] (3,4.6) node {\textbf{Receiver}};

\fill[xshift=7.5cm] (0,0) node[below] {\footnotesize{$0$}};
\fill[xshift=7.5cm] (1.5,-0.1) node[below] {\footnotesize{$p_{0}$}};
\fill[xshift=7.5cm] (3,0) node[below] {\footnotesize{$\hat p$}};
\fill[xshift=7.5cm] (6,0) node[below] {\footnotesize{$1$}};

\draw[xshift=7.5cm,loosely dotted] (3,0)--(3,2);
\draw[xshift=7.5cm,loosely dotted] (1.5,0)--(1.5,2.75)--(0,2.75);

\fill[xshift=7.5cm] (0,3.5) node[left] {\footnotesize{$u_{\ell}^{L}$}};
\fill[xshift=7.5cm] (0,0) node[left] {\footnotesize{$u_{r}^{L}$}};
\fill[xshift=7.5cm] (6,0.5) node[right] {\footnotesize{$u_{\ell}^{R}$}};
\fill[xshift=7.5cm] (6,4) node[right] {\footnotesize{$u_{r}^{R}$}};

\draw[xshift=7.5cm,dashdotted,line width=0.5pt] (0,3.5)--(6,4);

\fill[xshift=7.5cm] (canvas cs:x=0cm,y=3.5cm) circle (2pt);
\fill[xshift=7.5cm] (canvas cs:x=3cm,y=2cm) circle (2pt);
\fill[xshift=7.5cm] (canvas cs:x=6cm,y=4cm) circle (2pt);
\fill[xshift=7.5cm] (canvas cs:x=1.5cm,y=2.75cm) circle (2pt);

\end{tikzpicture}
\endpgfgraphicnamed
\caption{\label{fig:KG} Payoffs from static persuasion. Solid curves: payoffs without persuasion (information). Dashed curve: the sender's expected payoff in the KG solution. Dash-dotted curves: payoffs under a fully revealing experiment.}
\end{figure}
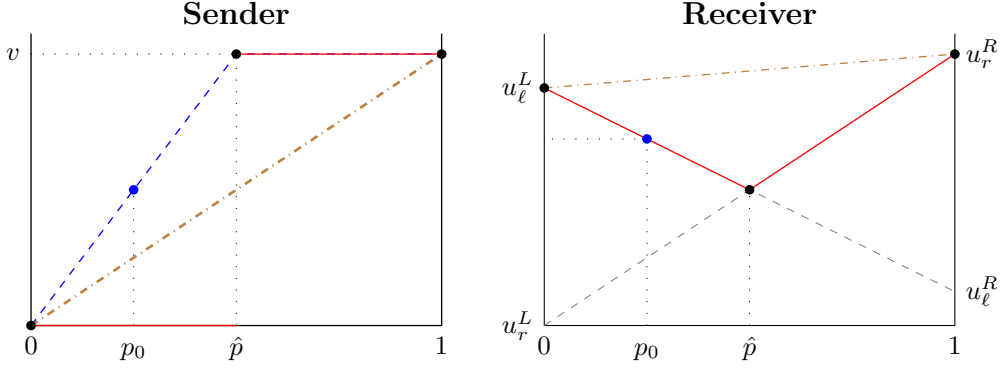

\paragraph*{KG Benchmark.} By now, it is well understood how the sender optimally persuades the
receiver if she can commit to an experiment without any restrictions.
For each $a\in\{\ell,r\}$, let $U_{a}(p)$ denote the receiver's
expected payoff when he takes action $a$ with belief $p$. In addition,
let $\hat{p}$ denote the belief at which the receiver is indifferent
between actions $\ell$ and $r$, that is, $U_{\ell}(\hat{p})=U_{r}(\hat{p})$.\footnote{Specifically, for each $p\in[0,1]$, $U_{\ell}(p):=pu_{\ell}^{R}+(1-p)u_{\ell}^{L}$
and $U_{r}(p):=pu_{r}^{R}+(1-p)u_{r}^{L}$. Therefore, $\hat{p}=\left(u_{\ell}^{L}-u_{r}^{L}\right)/\left(u_{r}^{R}-u_{\ell}^{R}+u_{\ell}^{L}-u_{r}^{L}\right)$,
which is well-defined in $(0,1)$ under our assumptions on the receiver's
payoffs.}

If the sender provides no information, then the receiver takes action
$r$ when $p_{0}\geq\hat{p}$. Therefore, persuasion is necessary
only when $p_{0}<\hat{p}$. In this case, the KG solution prescribes
an experiment that induces only two posteriors, $q_{-}=0$ and $q_{+}=\hat{p}$.
The former leads to action $\ell$, while the latter results in action
$r$. This experiment is optimal for the sender, because $\hat{p}$
is the minimum belief necessary to trigger action $r$, and setting
{$q_{-}=0$} maximizes the probability of generating $\hat{p}$,
and thus action $r$. The resulting payoff for the sender is $p_{0}v/\hat{p}$,
as given by the dashed line in the left panel of Figure \ref{fig:KG}.
The flip side is that the receiver enjoys no rents from persuasion;
his payoff is $\mathcal{U}(p_{0}):=\max\{U_{\ell}(p_{0}),U_{r}(p_{0})\}$, the
same as if no information were provided, as depicted in the right
panel of Figure \ref{fig:KG}.

\paragraph*{Dynamic model.}

We consider a dynamic version of the above Bayesian persuasion problem. Time flows continuously starting at 0. Unless the game has ended, at each point in time $t\ge0$, the sender may perform an experiment at a constant flow cost $c$ from a feasible set, which will be described precisely below, or \emph{pass}---not running any experiment and not incurring the flow cost $c$.\footnote{Passing enables the sender to stop experimenting at no cost. As will be seen, the experimentation always costs $c>0$ even at low intensity (informativeness). While this involves a form of discontinuity, it is largely for analytic convenience. Our results remain unchanged even if the cost is proportional to the intensity of the experiment \citep[see][]{Che2021}. One may also wonder what would happen if passing incurs the same cost $c$ as experimentation---a natural assumption if $c$ is interpreted as the waiting cost rather than the experimentation cost.  Our main results would still go through under this assumption, except for some details of the equilibrium characterization. Without the sender being able to freely stop experimenting, she would never give up on persuading, so the lower boundary of the experimentation region, denoted by $p_{\ast}$ later, is always determined by the receiver's incentives, as in Proposition \ref{prop:C2_hold}.} Just as it is costly for the sender to produce information, it is also costly for the receiver to process it. Specifically, if the sender experiments, then the receiver also pays the same flow cost and observes the experiment and its outcome. After that, he decides whether to take an irreversible action ($\ell$ or $r$), or to wait and listen to the information provided by the sender in the next instant. The former ends the game, while the latter lets the game continue.

There are two notable modeling assumptions. First, the receiver can
stop listening to the sender and take a game-ending action at any
point in time. This is the fundamental difference from KG, wherein
the receiver is allowed to take an action only after the sender finishes
her information provision. Second, the players' flow costs are assumed to be the same. This is, however, just a normalization
which allows us to directly compare the players' payoffs, and all
subsequent results can be reinterpreted as relative to each player's
individual flow cost.\footnote{\label{fn:asymmetric_costs}Suppose that the sender's cost is given
by $c_{s}$, while that of the receiver is $c_{r}$. Such a model
is equivalent to our normalized one in which $c_{r}^{\prime}=c_{s}^{\prime}=c_{r}$
and $v^{\prime}=v(c_{r}/c_{s})$. When solving the model for {a}
fixed set of parameters ($u_{a}^{\omega},v,c,\lambda$), this normalization
does not affect the results. If we let $c$ tend to $0$, we are implicitly
assuming that the sender's and receiver's (unnormalized) costs, $c_{s}$
and $c_{r}$, converge to zero at the same rate. See Footnote \ref{fn:relative_cost_discuss} for a relevant discussion.}

\paragraph*{Feasible experiments.}

We consider a general class of experiments whose informativeness per unit time is bounded in a proper way. Formally, we let $p_{t}$ denote the belief that $\omega=R$ at time $t$ and represent an experiment by a {\it regular} martingale process $\langle p_{t} \rangle$---i.e., a c\`{a}dl\`{a}g martingale process over $X:=[0,1]$ that is progressively measurable with respect to its natural filtration $\{\mathcal{F}_{t}\}$---with countably many discontinuities and a deterministic continuous path at each point in history. Its martingale property follows from the law of iterated expectations (or Bayes plausibility). We let $\mathcal{P}$  denote the set of all regular martingale processes.\footnote{The requirement of deterministic continuous path means that $\mathcal{P}$ does not include diffusion processes such as Brownian motion. But the class $\mathcal{P}$ encompasses a large class of jump (Poisson) processes. The implications of relaxing this requirement for our results (i.e., whether the sender would prefer a belief process failing this requirement to Poisson processes we allow in our model) remain an open question.}

For any $q\neq p_{-}:=\lim_{t'\uparrow t} p_{t'}$, let $\lambda^{\omega}(q,p_{t-}):=\lim_{dt\to 0} \mathbb{P}[p_{t}=q|p_{t-dt}, \omega]/{dt}$ denote the rate at which the belief changes from $p_{t-}$ to $q$ in state $\omega$. The set of feasible experiments is then defined as:
\begin{equation*}
    \mathcal{P}^{\ast}:=\left\{\langle p_{t} \rangle\in\mathcal{P}: \sum_{q\ne p_{t-}} |\lambda^R(q,p_{t-}) -\lambda^L(q,p_{t-})|\le \lambda, \text{ for all $t$ and $p_{t-}$}\right\}.
\end{equation*}
The set $\mathcal{P}^{\ast}$ includes all Poisson processes whose state-contingent jump rates $(\lambda^{L},\lambda^{R})$ satisfy $|\lambda^{L}-\lambda^{R}|\leq \lambda$ at each point in history; the feasible arrival rates are depicted by the shaded area in Figure \ref{fig:feasible-experiments}. It also includes all mixtures of those Poisson experiments.

\begin{figure}[htb]
 \raggedright{}\centering{}
\begin{tikzpicture}[scale=0.65]
		
 		\draw [<->,line width=0.5pt] (0,8)--(0,0)--(8,0);
 		\fill (0,0) node[below] {\footnotesize{$0$}};
 		\fill (0,8) node[left] {\footnotesize{$\lambda^{R}$}};
 		\fill (8,0) node[below]{\footnotesize{$\lambda^{L}$}};
		
 		\draw [line width=0.5pt,dotted] (0,0)--(8,8);
		
 		\fill (2.8,0) node [below] {\footnotesize{$\lambda$}};
 		\fill (0,3) node [left] {\footnotesize{$\lambda$}};

 		\fill[gray!50!white,fill opacity=0.5] (0,0)--(3,0)--(8,5)--(8,8)--(5,8)--(0,3)--(0,0);
		
 		\draw[line width=0.5pt,dashed] (0,3)--(5,8);

 		\draw[line width=0.5pt,dotted] (0,0)--(3.75,6.75);

         \draw[line width=0.5pt,dotted] (3.75,0)--(3.75,6.75)--(0,6.75);
         \fill (3.75,0) node[below] {\footnotesize{$\mu$}};
         \fill (0,6.75) node[left]{\footnotesize{$\lambda+\mu$}};

         \filldraw (0,3) circle [radius=3pt];	
		
 		\draw[fill=white] (3.75,6.75) circle (3pt);
		
 		\end{tikzpicture}
\caption{\label{fig:feasible-experiments} Arrival rates of feasible Poisson
experiments.}
\end{figure}
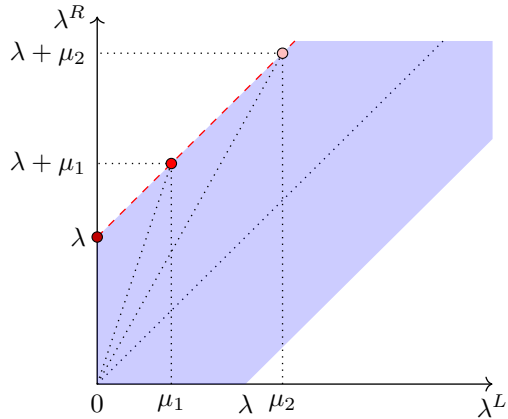
In fact, any information in  our class can be generated by ``diluting'' a conclusive Poisson signal arriving at rate  $\lambda$.  Consider  a conclusive signal that arrives in state $R$ at rate $\lambda$, depicted by a  black dot in Figure \ref{fig:feasible-experiments}. One can then add a white noise arriving in {\it both} states at some rate $\mu$ to this conclusive signal. The resulting signal, depicted in the figure by a  white dot, then arrives more frequently at rates $(\mu, \lambda + \mu)$ but is less precise, moving the belief only to posterior $q=\frac{p(\lambda+\mu)}{(1-p)\mu+p(\lambda+\mu)}(<1)$. The constant bound for the arrival rate differences means that the constraint on   flow information is independent of a prior, or {\it experimental}, as defined by \cite{denti2021experimental}; this stands in contrast to   other models such as {\it rational inattention} which assumes (prior-dependent) Shannon information cost or capacity.

\begin{lem} \label{lem:info}
An experiment $\langle p_t\rangle$ is feasible (i.e., $\langle p_t\rangle \in \mathcal{P}^{\ast}$) if and only if the following property holds at each point in history: there exists $\alpha:[0,1]\to[0,1]$ such that $\sum_{q\neq p}\alpha(q)\leq 1$;
\begin{itemize}
\item [(a)] for any $q\neq p$, the arrival rate of posterior belief $q$ given $p_{t-}=p$ is equal to
\begin{equation*}
    \alpha(q) \frac{ \lambda p(1-p)}{|q-p|}; \text{and}
\end{equation*}
\item [(b)] conditional on no jump, the belief  drifts according to
\begin{equation*}
     \dot{p}=-\left(\sum_{q>p}\alpha(q)-\sum_{q<p}\alpha(q)\right)\lambda p(1-p).
\end{equation*}
\end{itemize}
\end{lem}
\begin{proof}
See Appendix \ref{appendix:flow_bound}.
\end{proof}

Lemma \ref{lem:info} shows that a feasible flow experiment can be represented by the shares $\alpha$ of a unit ``capacity'' allocated to  Poisson experiments that trigger jumps to alternative posterior beliefs $q$, at rates $\alpha(q) \frac{\lambda p(1-p)}{|q-p|}$.  The jump rate in Part (a) simplifies to an expression familiar from the existing literature when the sender triggers a single jump with $\alpha(q)=1$  to  conclusive news  with either $q=0$ or $q=1$. For instance, conclusive $R$-evidence
($q=1$) is obtained at the rate of $\lambda p$, as is assumed in
``good'' news models \citep[see, e.g.,][]{keller:05}. Likewise,
conclusive $L$-evidence ($q=0$) is obtained at the rate of $\lambda(1-p)$,
as is assumed in ``bad'' news models \citep[see, e.g.,][]{keller2015breakdowns}.
Our model allows for such conclusive news, but it also allows for arbitrary non-conclusive news with $q\in(0,1)$, as well as any arbitrary mixture among such experiments. Further,  our information constraint captures the intuitive idea that more accurate information takes longer to generate. For example, assuming $q>p$, the arrival rate increases as the news becomes less precise ($q$ falls), and it approaches infinity
as the news becomes totally uninformative (i.e., as $q$ tends to $p$). Lastly, limited arrival rates  capture an important feature of our model that any meaningful persuasion takes time and requires delay.

Part (b) describes the law of motion governing the drift of beliefs when no jump occurs. Strikingly,  the drift rate depends only on the difference between the fractions of the capacity allocated to ``right'' versus ``left''
Poisson signals. That is, the rate does not depend on the precision
$q$ of the news in the individual experiments. The reason is
that the precision of news and its arrival rate offset each other,
leaving the drift rate unaffected.\footnote{Suppose $q>p$. This means that the sender has chosen $\nu^{R}=\lambda$ for the informative signal and $\mu\ge0$ for the noise. It is clear that $\mu$ does not affect the updating of the state since the noise
arrives at the same rate in both states.} This feature makes the analysis tractable while at the same time generalizing conclusive Poisson models in an intuitive way.

\begin{figure}
\raggedright{}\centering{}\beginpgfgraphicnamed{fig:feasibleexs}
\begin{tikzpicture}[scale=1.05]

			\node[left] at (-0.2,3) {$R$-drifting, targeting $0$:};
			\draw[|-|, thick] (0, 3) -- (8,3);
			\node[below] at (0, 2.9) {\footnotesize{$0$}};
			\node[below] at (8, 2.9) {\footnotesize{$1$}};

			\filldraw[fill=white] (0,3) circle [radius=3pt];
			\node[below] at (4, 2.9) {\footnotesize{$p$}};
			\filldraw[fill=white] (4,3) circle [radius=3pt];
			\draw[>={Stealth[length=6pt]},->>,thick] (4.1,3) -- (4.6,3);

			\draw[>={Stealth[length=6pt,width=4pt]},->,thick] (3.92,3.08) .. controls (2.75,3.7) and (1.25,3.7) .. (0.08,3.08);

			\node[left] at (-0.2,1.5) {$L$-drifting, targeting $q$:};
			
			\draw[|-|, thick] (0, 1.5) -- (8,1.5);
			\node[below] at (0, 1.4) {\footnotesize{$0$}};
			\node[below] at (8, 1.4) {\footnotesize{$1$}};

			\node[below] at (4, 1.4) {\footnotesize{$p$}};
			\filldraw[fill=white] (4,1.5) circle [radius=3pt];
			\node[below] at (7, 1.4) {\footnotesize{$q$}};
			\filldraw[fill=white] (7,1.5) circle [radius=3pt];
			\draw[>={Stealth[length=6pt]},->>,very thick] (3.9,1.5) -- (3.4,1.5);

			\draw[>={Stealth[length=6pt,width=4pt]},->,thick] (4.08,1.58) .. controls (5,2.1) and (6,2.1) .. (6.92,1.58);

			\node[left] at (-0.2,0) {Stationary, targeting $0$ and $q$:};
			
			\draw[|-|, thick] (0, 0) -- (8,0);
			\node[below] at (0, -0.1) {\footnotesize{$0$}};
			\node[below] at (8, -0.1) {\footnotesize{$1$}};

			\node[below] at (4, -0.1) {\footnotesize{$p$}};
			\filldraw[fill=white] (4,0) circle [radius=3pt];
			\node[below] at (7, -0.1) {\footnotesize{$q$}};
			\filldraw[fill=white] (7,0) circle [radius=3pt];
			
			\draw[>={Stealth[length=6pt,width=6pt]},->,thick] (4.08,0.08) .. controls (5,0.6) and (6,0.6) .. (6.92,0.08);

			\filldraw[fill=white] (0,0) circle [radius=3pt];

			\draw[>={Stealth[length=6pt,width=4pt]},->,thick] (3.92,0.08) .. controls (2.75,0.7) and (1.25,0.7) .. (0.08,0.08);

\end{tikzpicture}
\endpgfgraphicnamed\caption{\label{fig:feasible_exs} Three prominent feasible experiments.}
\end{figure}
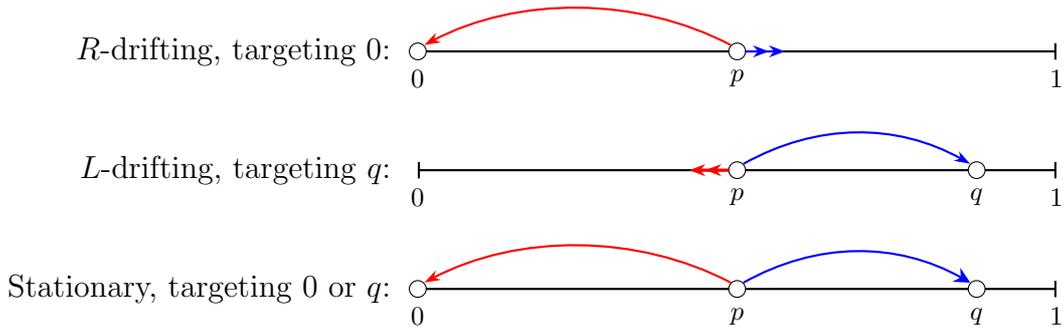

Among many feasible experiments, the following three, visualized in
Figure \ref{fig:feasible_exs}, will prove particularly relevant for
our purposes. They formalize the three modes of persuasion discussed
in the introduction.
\begin{itemize}
\item \textbf{$R$-drifting experiment} (confidence building):  $\alpha(0)=1$. The sender devotes all her capacity to a Poisson
experiment with (posterior) jump target $q=0$. In the absence of a jump, the posterior drifts to the right, at rate $\dot{p}=\lambda p(1-p)$.

\item \textbf{$L$-drifting experiment} (confidence spending):  $\alpha(q)=1$ for some $q>p$. The sender devotes all her capacity
to a Poisson experiment with jumps targeting some posterior $q>p$.
The precise jump target $q$ will be specified in our equilibrium
construction. In the absence of a jump, the posterior drifts to the
left, at rate $\dot{p}=-\lambda p(1-p)$.

\item \textbf{Stationary experiment} (confidence preserving): $\alpha(0)=\alpha(q)=1/2$ for some $q>p$. The sender assigns an equal share of her capacity to an experiment targeting $0$ and one targeting $q$. Absent jumps, the posterior remains unchanged.
\end{itemize}

\paragraph*{Solution concept.} We study (pure-strategy) Markov Perfect equilibria (MPE, hereafter) of this dynamic game in which both players' strategies depend only on the current belief $p$.\footnote{Naturally, this solution concept limits the use of (punishment) strategies depending on the payoff-irrelevant part of the histories, and serves to discipline strategies off the equilibrium path. For non-Markov equilibria, see our discussion in Section \ref{sec:conclusion}.} Formally, a profile of Markov strategies specifies for each $p\in[0,1]$, a flow experiment $\sigma^{S}(p)=(\alpha(q;p))_{q\in[0,1]}$ chosen by the sender, and an action $\sigma^{R}(p)\in\{\ell,r,\mbox{wait}\}$
chosen by the receiver. Given $\sigma=(\sigma^{S},\sigma^{R})$ and prior belief $p_{0}$, let $p_{t}$ denote the belief at time $t$ induced by the strategy profile and $\tau$ denote the stopping time at which the receiver takes action $\ell$ or $r$. Then, the sender's expected payoff is given by
\[
V^{\sigma}(p_{0})=v\,\mathbb{P}\left[\sigma^{R}(p_{\tau})=r\middle|p_{0}\right]-{c\mathbb{E}\left[\int_{0}^{\tau}\boldsymbol{1}_{\left\{\sum\alpha(q;p_{t})>0\right\} }dt\middle|p_{0}\right]},
\]
while the receiver's expected payoff is given by
\[
U^{\sigma}(p_{0})=\mathbb{E}\left[U_{\sigma^{R}(p_{\tau})}(p_{\tau})|p_{0}\right]-{c\mathbb{E}\left[\int_{0}^{\tau}\boldsymbol{1}_{\left\{\sum\alpha(q;p_{t})>0\right\} }dt\middle|p_{0}\right]}.
\]
A strategy profile $\sigma=(\sigma^{S},\sigma^{R})$ is {\it admissible} if the law of motion governing the belief evolution is well defined (see Appendix \ref{appendix:admissible} for detail) and the stopping time $\tau$ is also well defined. Let $\Sigma$ denote the set of all admissible strategy profiles.

\begin{defn}[Markov Perfect Equilibrium]\label{def:mpe}
A  strategy profile $\sigma=(\sigma^{S},\sigma^{R})\in \Sigma$
is a \emph{Markov perfect equilibrium} (MPE) if
\begin{enumerate}
\item[(i)] $V^{\sigma}(p)\ge V^{\hat{\sigma}}(p)$ for all $p\in[0,1]$ and   $\hat{\sigma}=(\hat{\sigma}^{S},\sigma^{R})\in \Sigma$,

\item[(ii)] $U^{\sigma}(p)\ge U^{\hat{\sigma}}(p)$ for all $p\in[0,1]$ and   $\hat{\sigma}=(\sigma^{S},\hat{\sigma}^{R})\in \Sigma$, and
\item[(iii)] for any $p$ such that the receiver stops (i.e., $\sigma^{R}(p)\in\{\ell,r\}$), \hfill{}(refinement)
\[
\sigma^{S}(p)\in\arg\max_{\alpha(\cdot;p)}\sum_{q}\alpha(q,p)\frac{\lambda p(1-p)}{\left|q-p\right|}\left(V^{\sigma}(q)-\boldsymbol{1}_{\left\{ \sigma^{R}(p)=r\right\} }v\right)-\boldsymbol{1}_{\left\{\sum\alpha(q;p_{t})>0\right\} }c.
\]
\end{enumerate}
\end{defn}

Whereas (i) and (ii) are obvious equilibrium requirements, (iii) imposes a restriction that captures the spirit of ``perfection'' in our continuous-time framework. To see its role clearly, suppose that the receiver would choose action $\ell$ unless the sender changes the belief significantly by running a flow experiment. In discrete time, the sender would simply choose a flow experiment that maximizes her expected payoff. In continuous time, however, the sender's strategy at such a point is inconsequential for her payoff; with probability one, the game would end with the receiver taking action $\ell$. With no further restriction on the sender's strategy, this continuous time peculiarity leads to severe but uninteresting equilibrium multiplicity (see Footnote \ref{fn:refinement_Thoerem1}). Property (iii) enables us to avoid the problem, by requiring the sender to choose a strategy that maximizes her \emph{instantaneous payoff normalized by $dt$} in the stopping region; it can be seen as selecting an MPE that is robust to a discrete-time approximation.

\section{Illustration: \emph{Persuading the Receiver to Listen}}

\label{sec:illustate}

We begin by illustrating the key issue facing the sender: \emph{persuading
the receiver to listen}. To this end, consider any prior $p_{0}<\hat{p}$
so that persuasion is not trivial and suppose that the sender repeatedly
chooses $R$-drifting experiments with jumps targeting $q=0$ until
the posterior either jumps to $0$ or drifts to $\hat{p}$, as depicted
on the horizontal axis in Figure \ref{fig:KG-Ldrifting}. This strategy
exactly replicates the KG solution (in the sense that it yields the
same probabilities of reaching the two posteriors, $0$ and $\hat{p}$),
provided that the receiver listens to the sender for a sufficiently
long time.

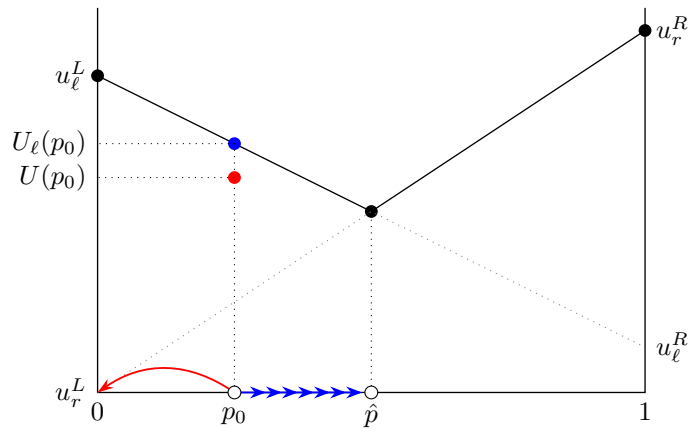
\begin{figure}[htb]
\centering{} \beginpgfgraphicnamed{LdrifingKG}
\begin{tikzpicture}[scale=0.6]
\draw[line width=0.5pt] (0, 8.5) -- (0,0) -- (12,0) -- (12,8.5);

\draw[line width=0.5pt] (0,7) -- (6,4) -- (12,8);

\fill (canvas cs:x=0cm,y=7cm) circle (4pt);
\fill (canvas cs:x=6cm,y=4cm) circle (4pt);
\fill (canvas cs:x=12cm,y=8cm) circle (4pt);
\fill (canvas cs:x=3cm,y=5.5cm) circle (4pt);

\draw[dotted,color=gray](0,0)--(6,4)--(12,1);

\draw[dotted] (6,0)--(6,4);
\draw[dotted] (3,0)--(3,5.5)--(0,5.5);

\fill (0,7) node[left] {\footnotesize{$u_{\ell}^{L}$}};
\fill (0,0) node[left] {\footnotesize{$u_{r}^{L}$}};
\fill (12,1) node[right] {\footnotesize{$u_{\ell}^{R}$}};
\fill (12,8) node[right] {\footnotesize{$u_{r}^{R}$}};

\fill (0,0) node[below] {\footnotesize{$0$}};
\fill (3,-0.1) node[below] {\footnotesize{$p_{0}$}};
\fill (6,0) node[below] {\footnotesize{$\hat p$}};
\fill (12,0) node[below] {\footnotesize{$1$}};

\draw[>={Stealth[length=6pt]},->>>>>>>,line width=0.7pt] (3,0) -- (5.8,0);

\draw[>={Stealth[length=6pt,width=4pt]},->,line width=0.7pt] (3.0,0) .. controls (2,0.7) and (1,0.7) .. (0,0);

\fill (canvas cs:x=3cm,y=4.75cm) circle (4pt);

\fill (0,5.5) node[left] {\footnotesize{$U_{\ell}(p_{0})$}};

\draw[dotted] (3,4.75)--(0,4.75);
\fill (0,4.75) node[left] {\footnotesize{$U(p_{0})$}};

\filldraw[fill=white] (3,0) circle [radius=4pt];

\filldraw[fill=white] (6,0) circle [radius=4pt];

\end{tikzpicture}
\endpgfgraphicnamed\caption{\label{fig:KG-Ldrifting} Replicating the KG outcome through $R$-drifting
experiments.}
\end{figure}

\emph{But will the receiver wait until the belief reaches $0$ or $\hat{p}$?} The answer is no. The KG experiment leaves no rents for the receiver without listening costs, and thus with listening costs the receiver will be strictly worse off than if he picks $\ell$ immediately. In Figure \ref{fig:KG-Ldrifting}, the receiver's expected gross payoff from the static KG experiment is $U_{\ell}(p_{0})$. Due to the listening costs, the receiver's expected payoff under the \emph{dynamic KG} strategy, denoted here by $U(p_{0})$, is strictly smaller than $U_{\ell}(p_{0})$. In other words, the dynamic strategy implementing the KG solution cannot persuade the receiver to wait and listen, so it does not permit any persuasion.\footnote{The KG outcome can also be replicated by other dynamic strategies. For instance, the sender could repeatedly choose a stationary strategy with jumps targeting $0$ and $\hat{p}$ until either jump occurs. However, this (and in fact, any other) strategy would not incentivize the receiver to listen, for the same reason as in the case of repeating $R$-drifting experiments.} Indeed, this problem leads to the existence of a no-persuasion MPE, regardless of the listening cost.

\begin{thm}[Persuasion Failure]
\label{thm:persuasion-failure} For any $c>0$, there exists an MPE
in which no persuasion occurs, that is, for any $p_{0}$, the receiver
immediately takes either action $\ell$ or $r$.
\end{thm}
\begin{proof}
Consider the following strategy profile: \textit{the receiver chooses
$\ell$ for $p<\hat{p}$ and $r$ for $p\ge\hat{p}$;} and \textit{the
sender chooses the $L$-drifting experiment with jump target $\hat{p}$
for all $p\in[\hat{\pi}_{\ell L},\hat{p})$ and passes for all $p\notin[\hat{\pi}_{\ell L},\hat{p})$},
where the cutoff $\hat{\pi}_{\ell L}$ is the belief at which the
sender is indifferent between the $L$-drifting experiment and stopping (followed by $\ell$).\footnote{Specifically, $\hat{\pi}_{\ell L}$ equates the sender's flow cost
$c$ to the instantaneous benefit from the $L$-drifting experiment:
\[
c=\frac{\lambda\hat{\pi}_{\ell L}(1-\hat{\pi}_{\ell L})}{\hat{p}-\hat{\pi}_{\ell L}}v,
\]
where the right-hand side is the sender's benefit $v$ from persuasion
multiplied by the rate at which the rightward jump to $\hat{p}$ occurs
(under the $L$-drifting experiment) at belief $\hat{\pi}_{\ell L}$. Solving the
equation yields
\[
\hat{\pi}_{\ell L}=\frac{1}{2}+\frac{c}{2\lambda v}-\sqrt{\left(\frac{1}{2}+\frac{c}{2\lambda v}\right)^{2}-\frac{c\hat{p}}{\lambda v}}.
\]
}

In order to show that this strategy profile is indeed an equilibrium,
first consider the receiver's incentives given the sender's strategy.
If $p\not\in[\hat{\pi}_{\ell L},\hat{p})$, then the sender never
provides information, so the receiver has no incentive to wait, and
will take an action immediately. If $p\in[\hat{\pi}_{\ell L},\hat{p})$,
then the sender never moves the belief into the region where the receiver
\emph{strictly} prefers to take action $r$ (i.e., strictly above
$\hat{p}$). This implies that the receiver's expected payoff is equal
to $U_{\ell}(p_{0})$ minus any listening cost she may incur. Therefore,
again, it is optimal for the receiver to take an action immediately.

Now consider the sender's incentives given the receiver's strategy.
If $p\geq\hat{p}$, then it is trivially optimal for the sender to
pass. Now suppose that $p<\hat{p}$. Our refinement, (iii) in Definition \ref{def:mpe}, requires that the sender choose a flow experiment that maximizes her instantaneous payoff, which is given by\footnote{The objective function follows from the fact that under the given strategy profile, the sender's value function is $V(p)=v$ if $p\geq\hat{p}$ and $V(p)=0$ otherwise; and when the target posterior is $q$, a Poisson jump occurs at rate $\lambda p(1-p)/|q-p|$.}
\[
\max_{\alpha(\cdot;p)}\sum_{q\neq p}\alpha(q;p)\lambda\frac{p(1-p)}{|q-p|}\mathbf{1}_{\{q\geq\hat{p}\}}v-\boldsymbol{1}_{\left\{\sum\alpha(q;p)>0\right\} }c\text{ subject to }\sum_{q\neq p}\alpha(q;p)\leq1.
\]
If the sender chooses any nontrivial experiment, its jump target must be $q=\hat{p}$. {Hence the sender's best response is either to maximize the jump rate to $\hat{p}$ (i.e., $\alpha(\hat p;p)=1$) or to pass.} The former is optimal if and only if $\frac{\lambda p(1-p)}{\hat{p}-p}v\ge c$, or equivalently $p\ge\hat{\pi}_{\ell L}$.\footnote{\label{fn:refinement_Thoerem1} Absent (iii) in Definition \ref{def:mpe}, there are many additional equilibria in which, in the stopping region, the sender may simply refuse to experiment or adopt an arbitrary Poisson experiment with jumps targeting beliefs other than $\hat p$ within the same stopping region. None of these alternative equilibria survive in the corresponding discrete-time setting. Our refinement allows us to select the continuous-time limit of the unique discrete-time no-persuasion equilibrium, and Theorem \ref{thm:persuasion-failure} holds \emph{despite} this refinement.}
\end{proof}

The no-persuasion equilibrium constructed in the proof showcases a
total collapse of trust between the two players. The receiver does
not trust the sender to convey valuable information (i.e., to choose
an experiment targeting $q>\hat{p}$), so he refuses to listen to
her. This attitude makes the sender desperate for a quick breakthrough;
she tries to achieve persuasion by targeting just $\hat{p}$, which
is indeed not enough for the receiver to be willing to wait.

\emph{Can trust be restored? In other words, can the sender ever persuade
the receiver to listen to her?} She certainly can, if she can commit
to a dynamic strategy, that is, if she can credibly promise to provide
more information in the future. Consider the following modification
of the dynamic KG strategy discussed above: the sender repeatedly
chooses $R$-drifting experiments with jumps targeting zero, until
either the jump occurs or the belief reaches $p^{*}>\hat{p}$. If
the receiver waits until the belief either jumps to $0$ or reaches
$p^{\ast}$, then her expected payoff is equal to\footnote{\label{fn:R_expected_waiting}To understand this explicit solution, first notice that under the prescribed strategy profile, the receiver takes action $\ell$ when
$p$ jumps to $0$, which occurs with probability $(p^{\ast}-p)/p^{\ast}$, and action $r$ when $p$ reaches $p^{*}$, which occurs with probability $p/p^{\ast}$. The last term captures the total expected listening cost. The length of time $\tau$ it takes for $p$ to reach $p^{\ast}$ absent jumps is derived as follows:
\[
p^{\ast}=\frac{p}{p+(1-p)e^{-\lambda\tau}}\Leftrightarrow\tau=\frac{1}{\lambda}\log\left(\frac{p^{\ast}}{1-p^{\ast}}\frac{1-p}{p}\right).
\]
Hence, the total listening cost is equal to
\[
(1-p)\int_{0}^{\tau}ctd\left(1-e^{-\lambda t}\right)+\left(p+(1-p)e^{-\lambda\tau}\right)c\tau=\left(p\log\left(\frac{p^{\ast}}{1-p^{\ast}}\frac{1-p}{p}\right)+1-\frac{p}{p^{\ast}}\right)\frac{c}{\lambda}.
\]
}
\[
U_{R}(p)=\frac{p^{\ast}-p}{p^{\ast}}u_{\ell}^{L}+\frac{p}{p^{\ast}}U_{r}(p^{\ast})-\left(p\log\left(\frac{p^{\ast}}{1-p^{\ast}}\frac{1-p}{p}\right)+1-\frac{p}{p^{\ast}}\right)\frac{c}{\lambda}.
\]
Importantly, if $p^{\ast}$ is sufficiently large relative to $c$, then $U_{R}(p)$ (the dashed curve in Figure \ref{fig:R-persuasion})
stays above $\max\{U_{\ell}(p),U_{r}(p)\}$ (the solid kinked curve) while $p$ drifts toward $p^{\ast}$, so the receiver prefers to wait. Intuitively, unlike in the KG solution, this ``more generous'' persuasion
scheme promises the receiver enough rents that make it worth listening
to.
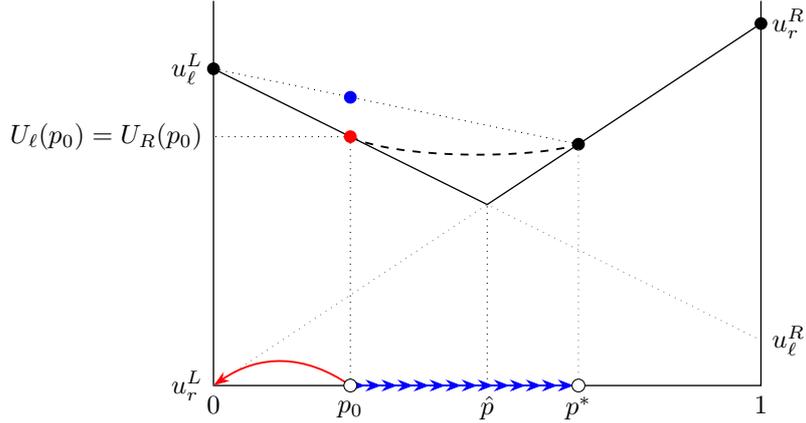
\begin{figure}
\centering{}\beginpgfgraphicnamed{Rdrifingcommitment}
\begin{tikzpicture}[scale=0.6]
\draw[line width=0.5pt] (0, 8.5) -- (0,0) -- (12,0) -- (12,8.5);

\draw[line width=0.5pt] (0,7) -- (6,4) -- (12,8);

\fill (canvas cs:x=0cm,y=7cm) circle (4pt);
\fill (canvas cs:x=12cm,y=8cm) circle (4pt);
\fill (canvas cs:x=3cm,y=6.37cm) circle (4pt);

\draw[dotted,color=gray](0,0)--(6,4)--(12,1);

\draw[dotted] (6,0)--(6,4);
\draw[dotted] (3,0)--(3,5.5)--(0,5.5);

\fill (0,7) node[left] {\footnotesize{$u_{\ell}^{L}$}};
\fill (0,0) node[left] {\footnotesize{$u_{r}^{L}$}};
\fill (12,1) node[right] {\footnotesize{$u_{\ell}^{R}$}};
\fill (12,8) node[right] {\footnotesize{$u_{r}^{R}$}};

\fill (0,0) node[below] {\footnotesize{$0$}};
\fill (3,-0.1) node[below] {\footnotesize{$p_{0}$}};
\fill (6,0) node[below] {\footnotesize{$\hat p$}};
\fill (12,0) node[below] {\footnotesize{$1$}};
\fill (8,0) node[below] {\footnotesize{$p^{\ast}$}};

\draw[>={Stealth[length=6pt]},->>>>>>>>>>>>>>,line width=0.7pt] (3,0) -- (7.9,0);

\draw[>={Stealth[length=6pt,width=4pt]},->,line width=0.7pt] (3.0,0) .. controls (2,0.7) and (1,0.7) .. (0,0);

\draw[dashed,line width=0.7pt,color=black] (3,5.5) .. controls (4,5) and (7,5)  .. (8,5.333);
\fill (canvas cs:x=8cm,y=5.333cm) circle (4pt);
\fill (canvas cs:x=3cm,y=5.5cm) circle (4pt);

\draw[dotted,color=gray](8,0)--(8,5.333);
\draw[dotted] (0,7)--(8,5.333);

\fill (0,5.5) node[left] {\footnotesize{$U_{\ell}(p_{0})=U_{R}(p_{0})$}};

\filldraw[fill=white] (3,0) circle [radius=4pt];

\filldraw[fill=white] (8,0) circle [radius=4pt];

\end{tikzpicture}

\endpgfgraphicnamed \caption{\label{fig:R-persuasion} Persuasive $R$-drifting experiments}
\end{figure}

If $c$ is sufficiently small, the required belief target $p^{*}$
need not exceed $\hat{p}$ by much. In fact, $p^{*}$ can be chosen
to converge to $\hat{p}$ as $c\to0$. In this fashion, a dynamic
persuasion strategy can be constructed to approximate the KG solution when $c$ is sufficiently small.

At first glance, this strategy seems unlikely to work without the
sender's commitment power. \emph{How can she credibly continue her
experiment even after the posterior has risen past $\hat{p}$? Why
not simply stop at the posterior $\hat{p}$---the belief that should
have convinced the receiver to choose $r$?} Surprisingly, however,
the strategy works even without commitment. This is because the equilibrium
beliefs generated by the Markov strategies themselves can provide
a sufficient incentive for the sender to continue beyond $\hat{p}$.
We already argued that, with a suitably chosen $p^{*}>\hat{p}$, the
receiver is incentivized to wait past $\hat{p}$, due to the ``optimistic''
equilibrium belief that the sender will continue to experiment until
a much higher belief $p^{*}$ is reached. Crucially, this optimism
in turn incentivizes the sender to carry out her strategy:\footnote{We will show in Section \ref{subsec:Equilibrium-Characterization}
that under certain conditions, using $R$-drifting experiments is
not just better than passing but also the optimal strategy (best response),
given that the receiver waits. Here, we illustrate the possibility
of persuasion for this case. The logic extends to other cases where
the sender optimally uses different experiments to persuade the receiver.} were she to deviate and, say, pass at $q=\hat{p}$, the receiver
would simply wait (instead of choosing $r$), believing that the sender
will shortly resume her $R$-drifting experiments after the ``unexpected'' pause. Given this response, the sender cannot gain from deviating:
she cannot convince the receiver to ``prematurely'' choose $r$.
To summarize, the sender's strategy instills optimism in the receiver
that makes him wait and listen, and this optimism, or the \emph{power of beliefs}, in turn incentivizes the sender to carry out the strategy.

The above power-of-beliefs logic extends beyond the Poisson model we employ here,\footnote{\label{fn:Brownian-Motion}Consider \citet{Henry2017}'s model in which the belief, as expressed by the log likelihood ratio $s=\ln(p/(1-p))$, follows a Brownian motion with a drift given by the state. In keeping with our model, suppose at each point in time the sender either experiments or passes, and the receiver chooses $\ell,r$, or ``wait,'' with the flow cost $c$ incurred on both sides if the sender experiments and the receiver waits. As noted in Footnote \ref{fn:HO}, this model is similar to \citet{Henry2017}'s no-commitment regime, except that our sender has the option to pass without ending the game and the receiver incurs a flow cost. A simple MPE is then characterized by two stopping bounds, $s_{*}\le\hat{s}:=\ln(\hat{p}/(1-\hat{p}))$
and $s^{*}\ge\hat{s}$, such that the sender experiments and the receiver waits if and only if $s\in(s_{*},s^{*})$. Our ``power of beliefs'' argument would imply that a range of persuasion targets $s^{*}$ are
supported as MPE for $c>0$ sufficiently low, and that range would span the entire $(\hat{s},\infty)$ as $c\to0$.} but it does depend on subtle details of the model. For example, consider a variation of the model in which the sender becomes unable to provide further information at some (Poisson distributed) random time. If the event is also observable to the receiver, then the above logic applies unchanged. If it is unobservable to the receiver, however, the logic no longer holds: no matter how unlikely the event is, the sender will stop providing information as soon as the belief rises above $\hat p$, unraveling any persuasion equilibrium. Likewise, with a deadline at which the receiver should take an action, the power-of-beliefs logic survives if the arrival of the deadline is stochastic but fails if the deadline is deterministic. See Section \ref{sec:conclusion} for discussions on a few other relevant features.

\section{Persuasion Equilibria}\label{sec:folk}

The equilibrium logic outlined in the previous section applies not
just to strategy profiles that approximate the KG solution, but also
to other strategy profiles with a persuasion target $p^{\ast}\in(\hat{p},1)$.
Building upon this observation, we establish a folk-theorem-like result: \emph{any sender (receiver) payoff between the KG solution and full revelation} can be supported as an MPE payoff in the limit as $c$ tends to $0$.

\begin{thm}
\label{thm:folk} Fix any prior $p_{0}\in(0,1)$.
\begin{itemize}
\item[(a)] For any sender payoff $V\in\left(p_{0}v,\min\{p_{0}/\hat{p},1\}v\right)$,
if $c$ is sufficiently small, there exists an MPE in which the sender obtains $V$.
\item[(b)] For any receiver payoff $U\in\left(\mathcal{U}(p_{0}),p_{0}u_{r}^{R}+(1-p_{0})u_{\ell}^{L}\right)$,
if $c$ is sufficiently small, there exists an MPE in which the receiver
achieves $U$.
\end{itemize}
\end{thm}
The proof of Theorem \ref{thm:folk} follows from the equilibrium
constructions of Propositions \ref{prop:C2_hold} and \ref{prop:C2_fail}
in Section \ref{subsec:Equilibrium-Characterization}. The main argument
for the proof is outlined below.

Figure \ref{fig:folk2} depicts how the set of implementable payoffs
for each player varies according to $p_{0}$ in the limit as $c$
tends to $0$. Theorem \ref{thm:folk} states that any payoffs in
the   shaded areas can be implemented in an MPE, provided
that $c$ is sufficiently small. In the left panel, the upper bound
for the sender's payoff is given by the KG-optimal payoff $\min\{p_{0}/\hat{p},1\}v$,
and the lower bound is given by the sender's payoff from full revelation
$p_{0}v$. For the receiver, by contrast, full revelation defines
the upper bound $p_{0}u_{r}^{R}+(1-p_{0})u_{\ell}^{L}$, whereas the
KG payoff, which leaves no rent for the receiver, is given by $\mathcal{U}(p_{0})$.

\begin{figure}
\beginpgfgraphicnamed{folktheorem2} \centering{}
\begin{tikzpicture}[scale=1]

\draw[line width=0.5pt] (0, 4.3) -- (0,0) -- (6,0) -- (6,4.3);
\fill[pattern=dots,line width=0.5pt] (0,0) -- (6,4)--(3,0) -- (0,0);
\draw[line width=0.5pt] (0,0) -- (6,4)--(3,0) -- (0,0);
\fill[color=gray,fill opacity=0.3,draw=black,line width=0.2pt] (0,0) -- (3,4) -- (6,4)--(0,0);

\draw[line width=1pt] (0,0) -- (3,0);
\draw[line width=1pt] (3,4) -- (6,4);

\draw[dotted] (3,0)--(3,4);
\draw[dotted] (1.5,0)--(1.5,2);
\draw[dotted] (0,4)--(3,4);
\draw[dotted] (0,2)--(1.5,2);
\draw[dotted] (0,1)--(1.5,1);
\draw[dashed,line width=0.5pt] (0,0)--(4.5,4);
\draw[dotted] (4.5,0)--(4.5,4);

\fill (0,0) node[below] {\footnotesize{$0$}};
\fill (1.5,-0.1) node[below] {\footnotesize{$p_{0}$}};
\fill (3,0) node[below] {\footnotesize{$\hat p$}};
\fill (4.5,0) node[below] {\footnotesize{$p^{\ast}$}};
\fill (6,0) node[below] {\footnotesize{$1$}};

\fill (0,4) node[left] {\footnotesize{$v$}};
\fill (0,2) node[left] {\footnotesize{$\frac{p_{0}}{\hat{p}}v$}};
\fill (0,1) node[left] {\footnotesize{$p_{0}v$}};

\draw (3,4.6) node {\textbf{Sender}};

\draw[xshift=7.5cm] (0, 4.3) -- (0,0) -- (6,0) -- (6,4.3);

\fill[xshift=7.5cm,color=gray,fill opacity=0.3,draw=black,line width=0.2pt] (0,3.5) -- (3,2) -- (6,4)--(0,3.5);

\draw[xshift=7.5cm,dashed,color=gray](0,0)--(3,2)--(6,0.5);

\draw[xshift=7.5cm] (3,4.6) node {\textbf{Receiver}};

\draw[xshift=7.5cm,line width=1pt] (0,3.5)--(3,2)--(6,4);

\fill[xshift=7.5cm] (0,0) node[below] {\footnotesize{$0$}};
\fill[xshift=7.5cm] (1.5,-0.1) node[below] {\footnotesize{$p_{0}$}};
\fill[xshift=7.5cm] (3,0) node[below] {\footnotesize{$\hat p$}};
\fill[xshift=7.5cm] (4.5,0) node[below] {\footnotesize{$p^{\ast}$}};
\fill[xshift=7.5cm] (6,0) node[below] {\footnotesize{$1$}};

\draw[xshift=7.5cm,dotted] (3,0)--(3,2);
\draw[xshift=7.5cm,dotted] (1.5,2.75)--(0,2.75);
\draw[xshift=7.5cm,dotted] (1.5,0)--(1.5,3.625);
\draw[xshift=7.5cm,dotted] (4.5,0)--(4.5,3);
\draw[xshift=7.5cm,dashed,line width=0.5pt] (0,3.5)--(4.5,3);

\fill[xshift=7.5cm] (0,2.75) node[left] {\footnotesize{$\mathcal{U}(p_{0})$}};
\fill[xshift=7.5cm] (0,3.5) node[left] {\footnotesize{$u_{\ell}^{L}$}};
\fill[xshift=7.5cm] (0,0) node[left] {\footnotesize{$u_{r}^{L}$}};
\fill[xshift=7.5cm] (6,0.5) node[right] {\footnotesize{$u_{\ell}^{R}$}};
\fill[xshift=7.5cm] (6,4) node[right] {\footnotesize{$u_{r}^{R}$}};

\end{tikzpicture}

\endpgfgraphicnamed \caption{\label{fig:folk2} Implementable payoff set for each player at each
$p_{0}$.}
\end{figure}
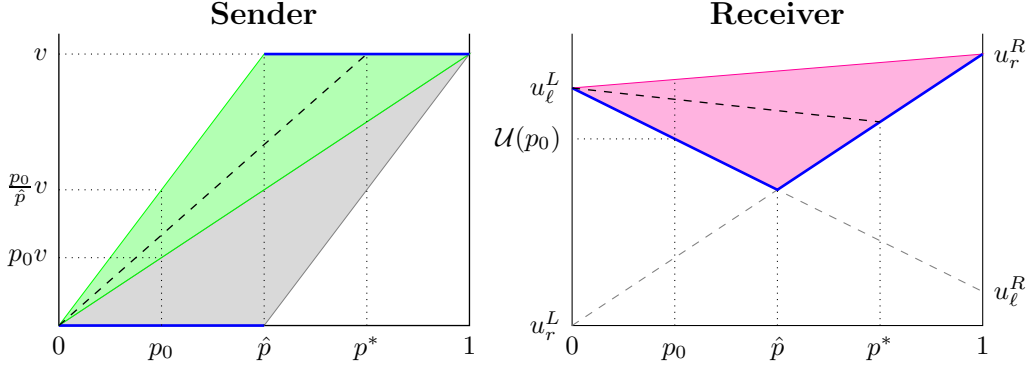

Note that Theorem \ref{thm:folk} is silent about payoffs in the
 dotted region. In the static KG environment, these payoffs can be
achieved by the (sender-pessimal) experiment that splits the prior
$p$ into two posteriors, $1$ and $q\in[0,\hat{p}]$. The following
theorem shows that the sender's payoffs in this region cannot be supported
as an MPE payoff for a sufficiently small $c>0$ (even without invoking
our refinement).
\begin{thm}
\label{thm:infeasible_MPE_payoffs} If $p_{0}\leq\hat{p}$, then the
sender's payoff in any MPE is either equal to $0$ or at least $p_{0}v-2c/\lambda$.
If $p_{0}>\hat{p}$, then the sender's payoff in any MPE is at least
$p_{0}v-2c/\lambda$.
\end{thm}
\begin{proof}
Fix $p_{0}\leq\hat{p}$, and consider any MPE. If the receiver's strategy
is to wait at $p_{0}$, then the sender can always adopt the stationary
strategy with jump targets $0$ and $1$, which will guarantee her
a payoff of $p_{0}v-2c/\lambda$.\footnote{\label{fn:S_expected_waiting}In order to understand this payoff, notice that the strategy fully
reveals the state, and thus the sender gets $v$ only in state $R$.
In addition, in each state, a Poisson jump occurs at rate $\lambda/2$,
and thus the expected waiting time equals $2/\lambda$, which is multiplied
by $c$ to obtain the expected cost.} If the receiver's strategy is to stop at $p_{0}$, then the receiver
takes action $\ell$ immediately, in which case the sender's payoff
is equal to $0$. Therefore, the sender's expected payoff is either
equal to $0$ or above $p_{0}v-2c/\lambda$.

Now suppose $p_{0}>\hat{p}$, and consider any MPE. As above, if $p_{0}$
belongs to the waiting region, then the sender's payoff must be at least $p_{0}v-2c/\lambda$. If $p$ belongs to the stopping region, then the sender's payoff is equal to $v$. In either case, the sender's
payoff is at least $p_{0}v-2c/\lambda$.
\end{proof}
We prove  {Theorem \ref{thm:folk}} by constructing MPEs with a particularly
simple structure:
\begin{defn}
\label{def:simple_MPE} A Markov perfect equilibrium is a \emph{simple
MPE} (henceforth, SMPE) if there exist $p_{\ast}\in(0,\hat{p})$ and
$p^{\ast}\in(\hat{p},1)$ such that the receiver chooses action $\ell$
if $p<p_{\ast}$, waits if $p\in(p_{\ast},p^{\ast})$, and chooses
action $r$ if $p\geq p^{\ast}$.\footnote{We do not restrict the receiver's decision at the lower bound $p_{\ast}$,
so that the waiting region can be either $(p_{*},p^{*})$ or $[p_{*},p^{*})$.
Requiring $W=(p_{*},p^{*})$ can lead to non-existence of an SMPE in Proposition \ref{prop:C2_hold}. Requiring $W=[p_{*},p^{*})$
can lead to non-admissibility of the sender's best response in Proposition
\ref{prop:C2_fail}.}
\end{defn}
In other words, in an SMPE, the receiver waits for more information if $p\in W$ and takes an action, $\ell$ or $r$, otherwise, where $W=(p_{*},p^{*})$ or $W=[p_{*},p^{*})$ denotes the \emph{waiting region:}
\[
\underset{\hspace{-0.35cm}p=0}{|}\overbrace{\text{------------------------}}^{\ell}{\scriptstyle p_{*}}\overbrace{\text{------------------------}}^{\text{wait}}{\scriptstyle {p}^{*}}\overbrace{\text{------------------}}^{r}\underset{1}{|}
\]
While this is the most natural equilibrium structure, we do not exclude
possible MPEs that violate this structure. Whether such non-simple
MPEs exist or not is irrelevant for our results. While we construct
SMPEs to establish  {Theorem \ref{thm:folk}}, Theorem \ref{thm:infeasible_MPE_payoffs}
is valid for \emph{all} MPEs. Finally, we continue to require our
refinement with SMPEs.

To prove  {Theorem \ref{thm:folk}}, we begin by fixing $p^{\ast}\in(\hat{p},1)$.
Then, for each $c$ sufficiently small, we identify a unique value
of $p_{\ast}$ for which an SMPE can be constructed. We then show
that as $c\rightarrow0$, $p_{\ast}$ approaches $0$ as well (see
Propositions \ref{prop:C2_hold} and \ref{prop:C2_fail} in Section
\ref{subsec:Equilibrium-Characterization}). This implies that given
$p^{\ast}$, the limit SMPE spans the sender's payoffs on the line
segment that connects $(0,0)$ and $(p^{\ast},v)$---the dashed line
in the left panel of Figure \ref{fig:folk2}---and the receiver's
payoffs on the line segment that connects $(0,u_{\ell}^{L})$ and
$(p^{\ast},U_{r}(p^{\ast}))$ in the right panel. By varying $p^{\ast}$
from $\hat{p}$ to $1$, we can cover the entire shaded areas in Figure
\ref{fig:folk2}. Note that with this construction and the uniqueness
claims in Propositions \ref{prop:C2_hold} and \ref{prop:C2_fail},
we also obtain a characterization of feasible \emph{payoff vectors}
$(V,U)$ for the sender and receiver that can arise in an SMPE in
the limit as $c$ tends to $0$. We state this in the following corollary.
\begin{cor}
\label{prop:payoff_vectors}For any prior $p_{0}\in[0,1]$, in the
limit as $c$ tends to $0$, the set of SMPE payoff vectors $(V,U)$
is given by
\[
\left\{ (V,U)\middle|\exists p^{*}\in\left[\max\left\{ p_{0},\hat{p}\right\} ,1\right]:\:V=\frac{p_{0}}{p^{*}}v,\:U=\frac{p_{0}}{p^{*}}U_{r}(p^{*})+\frac{p^{*}-p_{0}}{p^{*}}u_{\ell}^{L}\right\} ,
\]
with the addition of the no-persuasion payoff vector $(0,U(p_{0}))$
for $p_{0}<\hat{p}$.
\end{cor}

\section{Persuasion Dynamics}

\label{sec:persuasion_dynamics}

In this section, we provide a full description of SMPE strategy profiles
and illustrate the resulting equilibrium persuasion dynamics. We first
explain why the sender optimally uses the three modes of persuasion
discussed in the Introduction and Section \ref{sec:model}. Then,
using them as building blocks, we construct full SMPE strategy profiles.

\subsection{Modes of Persuasion}

\label{subsec:modes_persuasion}

Fix an SMPE with two threshold beliefs $p_{*}$ and $p^{*}$, where
$p_{*}<\hat{p}<p^{*}$. We investigate the sender's optimal persuasion/experimentation behavior at any belief $p\in(0,1)$ in that equilibrium.

Suppose that the sender runs a flow experiment that targets $q\neq p$ when the current belief is $p$. Then, by Lemma \ref{lem:info}, the belief jumps to $q$ at rate $\lambda p(1-p)/|q-p|$ and, absent jumps, moves continuously according to $\dot{p}=-\text{sgn}(q-p)\lambda p(1-p)$, where $\text{sgn}(x)$ denotes the signum function. Therefore, her flow benefit is given by
\[
v(p;q):=\lambda\frac{p(1-p)}{|q-p|}(V(q)-V(p))-\text{sgn}(q-p)\lambda p(1-p)V^{\prime}(p),
\]
where $V(\cdot)$ is the sender's value
of playing the candidate equilibrium strategy.\footnote{Note that the sender's value function may not be everywhere differentiable.
We ignore this here to give a simplified argument illustrating the
properties of the optimal strategy for the sender. The formal proofs
can be found in Appendix \ref{sec:Proofs-of-Propositions}.} Specifically, for $q>p$, the flow benefit consists of the value
increase from a breakthrough which arises at rate $\lambda\frac{p(1-p)}{|q-p|}$
(the first term) and the decay of value in its absence (the second
term). For $q<p$, the first term captures the value decrease
from a breakdown, while the second term represents the gradual appreciation
in its absence.

At each point in time, the sender can choose any countable mixture
over experiments. Therefore, at each $p$, her flow benefit from \emph{optimal}
persuasion is equal to
\begin{equation}
v(p):=\max_{\alpha(\cdot;p)}\sum_{q}\alpha(q;p)v(p;q)\text{ subject to }\sum_{q}\alpha(q;p)\le1.\label{eq:HJB_general_main}
\end{equation}
The function $v(p)$ represents the gross flow value from experimentation.
It plays an important role in characterizing the sender's strategy
in the stopping region as well as in the waiting region. If $p\geq p^{\ast}$,
then the receiver takes action $r$ immediately, and thus $V(p)=v$
for all $p\geq p^{\ast}$. It follows that $v(p)=0<c$, so it is optimal
for the sender to pass, which is intuitive. If $p<p_{\ast}$ then
the sender has only one instant to persuade the receiver, and therefore
she experiments only when $v(p)\geq c$: if $v(p)<c$, persuasion
is so unlikely that she prefers to pass, or more intuitively, gives
up on persuasion.

In the waiting region $p\in(p_{*},p^{*})$, the sender must have an
incentive to experiment, which suggests that $v(p)\geq c$.\footnote{Suppose that $v(p)<c$. Then, the sender strictly prefers passing
forever to conducting any experiment at $p$ followed by the optimal
continuation. This implies that the value function must be $V(p)=0$---the
value of passing forever. Hence, we must have $v(p)\ge c$ whenever
$V(p)>0$, which holds if $p\in W$.} In particular, when the sender's equilibrium strategy involves experimentation,
her value function is characterized by the Hamilton-Jacobi-Bellman
(HJB) equation, which means that $V(p)$ is adjusted so that $v(p)=c$
holds.

The following proposition simplifies the potentially daunting task
of characterizing the sender's optimal experiment at each belief in
\eqref{eq:HJB_general_main}, to searching among a small subset of feasible experiments.
\begin{prop}
\label{prop:persuasion_modes} Consider an SMPE where the receiver's
strategy is given by $p_{*}<\hat{p}<p^{*}$.
\begin{enumerate}
\item[(a)] For all $p\in(0,1)$, there exists a best response that  {involves at most two distinct Poisson jumps, one to $q_{1}(>p)$ at rate $\alpha_{1}:=\alpha(q_{1};p)$ and the other to $q_{2}(<p)$ at rate $\alpha_{2}:=\alpha(q_{2};p)$}.
\item[(b)] Suppose that $V(\cdot)$ is nonnegative, increasing, and strictly
convex over $(p_{\ast},p^{\ast}]$, and $V(p_{\ast})/p_{\ast}\leq V^{\prime}(p_{\ast})$.
Then, the best response in part (a) has
\begin{enumerate}
\item for $p\in(p_{\ast},p^{\ast})$, $\alpha_{1}+\alpha_{2}=1$ with $q_{1}=p^{\ast}$
and $q_{2}=0$;
\item for $p<p_{\ast}$, either the sender passes, or $\alpha_{1}=1$ and $q_{1}=p_{\ast}$ or $q_{1}=p^{\ast}$;
\item for $p>p^{*}$, the sender passes.
\end{enumerate}
\end{enumerate}
\end{prop}
For part (a) of Proposition \ref{prop:persuasion_modes}, notice that
the right-hand side in equation \eqref{eq:HJB_general_main} is linear
in each $\alpha(q;p)$ and the constraint $\sum_{q}\alpha(q;p)\le1$
is also linear. Therefore, by the standard linear programming logic,
there exists a solution that makes use of at most two experiments,
one below $p$ and the other above $p$.\footnote{One may wonder why we allow for two experiments. In fact, linearity
implies that there exists a maximizer that puts all weight on a single
experiment. But to obtain an \emph{admissible} Markov strategy, using
two experiments is sometimes necessary. For example, if $p$ is an
absorbing belief, then admissibility requires that the stationary
strategy be used at that belief, requiring two experiments. See Appendix \ref{appendix:admissible} for details.} This result implies that
\begin{equation}
v(p)=\max_{(\alpha_{1},q_{1}),(\alpha_{2},q_{2})}\lambda p(1-p)\left[\alpha_{1}\frac{V(q_{1})-V(p)}{q_{1}-p}-\alpha_{2}\frac{V(p)-V(q_{2})}{p-q_{2}}-(\alpha_{1}-\alpha_{2})V^{\prime}(p)\right],\label{eq:simplified_HJB_two}
\end{equation}
subject to $\alpha_{1}+\alpha_{2}\le1$ and $q_{2}<p<q_{1}$.

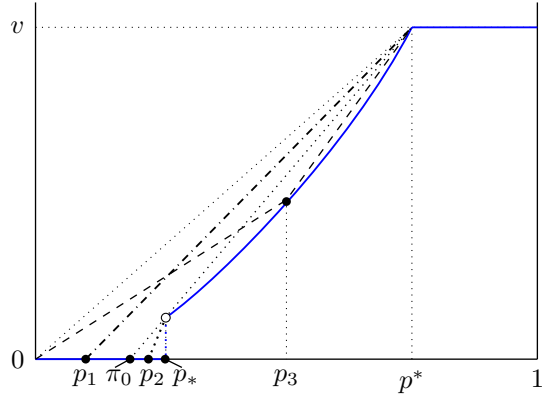
\begin{figure}
 \centering{}\beginpgfgraphicnamed{fig:HJBsimplified}
\begin{tikzpicture}[scale=1.1]
\draw[line width=0.5pt] (0, 4.3) -- (0,0);
\draw[line width=0.5pt] (0,0) -- (6,0);
\draw[line width=0.5pt] (6,4.3)-- (6,0);

\draw[dotted] (0,0) -- (4.5,4);
\draw[dotted] (0,4)--(6,4);
\draw[dotted] (4.5,0)--(4.5,4);


\draw[line width=0.7pt] (1.5572,0.5) .. controls (2.25,1) and (3.75,2.5447) .. (4.5,4.0000);

\draw[dotted,line width=0.7pt] (1.5572,0.5)--(1.5572,0);

\draw[line width=0.7pt] (0,0)--(1.5572,0);


\draw[line width=0.7pt] (4.5,4)--(6,4);

\fill (0,0) node[left] {\footnotesize{$0$}};
\fill (canvas cs:x=1.55cm,y=0cm) circle (1.5pt);
\draw[line width=0.5pt] (1.55,0)--(1.75,-0.1);
\fill (1.8,0) node[below] {\footnotesize{$p_{\ast}$}};

\fill (4.5,0) node[below] {\footnotesize{$p^{\ast}$}};
\fill (6,0) node[below] {\footnotesize{$1$}};
\fill (0,4) node[left] {\footnotesize{$v$}};

\fill (canvas cs:x=0.6cm,y=0cm) circle (1.5pt);
\fill (0.6,0) node[below] {\footnotesize{$p_{1}$}};
\draw[dashdotted,line width=0.7pt] (0.6,0)--(4.5,4);


\draw[dotted] (1.5572,0.5)--(1.5572,0);
\fill (canvas cs:x=1.35cm,y=0cm) circle (1.5pt);
\fill (1.4,0) node[below] {\footnotesize{$p_{2}$}};
\draw[dotted,line width=0.8pt] (1.35,0)--(1.5572,0.5);

\fill (canvas cs:x=3cm,y=1.9cm) circle (1.5pt);
\draw[dotted] (3,1.92)--(3,0);
\draw[dashed,line width=0.5pt] (0,0)--(3,1.92)--(4.5,4);
\fill (3,0) node[below] {\footnotesize{$p_{3}$}};

\fill (canvas cs:x=1.13cm,y=0cm) circle (1.5pt);
\fill (1,0) node[below] {\footnotesize{$\pi_{0}$}};
\draw[dotted,line width=0.5pt] (1.13,0)--(4.5,4);
\draw[line width=0.5pt] (1.13,0)--(0.95,-0.1);

\filldraw[fill=white] (1.5572,0.5) circle [radius=1.5pt];

\end{tikzpicture}
\endpgfgraphicnamed\caption{\label{fig:HJB_simplified} Optimal Poisson jump targets for different
values of $p$. The solid curve represents the sender's value function
in an SMPE with $p_{\ast}$ and $p^{\ast}$.}
\end{figure}

Part (b) of Proposition \ref{prop:persuasion_modes} states that if
$V(\cdot)$ satisfies the stated properties, which will be shown to
hold in equilibrium later, then there are only three candidates for
optimal Poisson jump targets, $0$, $p_{\ast}$, and $p^{\ast}$,
regardless of $p\in(0,p^{\ast})$. As illustrated in Figure \ref{fig:HJB_simplified},
the RHS of \eqref{eq:simplified_HJB_two} boils down to choosing $q_{1}>p$
to maximize the slope of $V$ between $q_{1}$ and $p$ (i.e., the
first fraction) or choosing $q_{2}<p$ to minimize the slope of $V$
between $q_{2}$ and $p$ (i.e., the second fraction). In the waiting
region, the former strategy leads to $q_{1}=p^{*}$ whereas the latter
strategy leads to $q_{2}=0$ (see $p_{3}$ and the dashed lines in
Figure \ref{fig:HJB_simplified}).\footnote{Note that $q_{1}>p^{\ast}$ yields a lower slope than $q_{1}=p^{\ast}$;
intuitively, the sender would be wasting her persuasion rate if she
targets above $p^{\ast}$. Meanwhile, when $p\in(p_{\ast},p^{\ast})$,
$q_{2}=p_{\ast}$ yields a higher slope than $q_{2}=0$, given $V(p_{\ast})/p_{\ast}\leq V^{\prime}(p_{\ast})$.} Similarly, if $p<p_{\ast}$ then $q_{2}=0$ is optimal and $q_{1}$
is either $p_{\ast}$ (see $p_{2}$ and the dotted line) or $p^{\ast}$
(see $p_{1}$ and the dash-dotted line).

Proposition \ref{prop:persuasion_modes} implies that the sender makes
use of the following three modes of persuasion at each $p<p^{\ast}$.
{\begin{itemize}
    \item Confidence building: \texorpdfstring{$R$}{R}-drifting experiment with jump target $0$.
    \item Confidence spending: \texorpdfstring{$L$}{L}-drifting experiment with jump target $q_{1}=p^{\ast}$ or possibly $q_{1}=p_{\ast}$ if $p<p_{\ast}$.
    \item Confidence preserving: stationary experiment with jump targets $q_{1}=p^{\ast}$ and $q_{2}=0$.
\end{itemize}}

Two aspects determine the sender's choice over these experiments in
her optimal strategy. First, strategies may differ in the distributions
over final posteriors they induce. In particular, they may differ
in the probability of persuasion (i.e., of the belief reaching $p^{*}$).
Second, and more interestingly, they may differ in the time it takes
for the sender to conclude persuasion. While the former feature has
been studied extensively by the static persuasion models, the latter
feature is novel here and is crucial for shaping the precise persuasion
dynamics.

To be concrete, compare the confidence-building strategy that uses
the $R$-drifting experiment (with jump target $0$) until the belief
reaches $p^{*}$, with the confidence-preserving strategy that uses
the stationary experiment (with jump targets $q_{1}=p^{\ast}$ and
$q_{2}=0$) until a jump occurs. Starting from any belief $p\in(p_{*},p^{*})$,
both strategies eventually lead to a posterior of $0$ or $p^{*}$,
with identical probabilities. Hence they yield the same outcome for
the two players, \emph{except for the time it takes} for the persuasion
process to conclude. Clearly, the sender wishes to minimize that time,
which explains her choice between the two modes of persuasion. Intuitively,
if the current belief is close to the persuasion target $p^{*}$,
then confidence building (i.e., $R$-drifting) takes less time on
average than confidence preserving (i.e., stationary), since the former
concludes persuasion within a short period of time, whereas the latter may take a long time and thus proves costly.\footnote{The expected persuasion costs associated with $R$-drifting and stationary strategies, which can be computed as illustrated in Footnotes \ref{fn:R_expected_waiting} and \ref{fn:S_expected_waiting}, are respectively given by
\begin{equation*}
    C_{+}(p;p^{\ast})=\frac{c}{\lambda}\left(p\log\left(\frac{p^{\ast}}{1-p^{\ast}}\frac{1-p}{p}\right)+1-\frac{p}{p^{\ast}}\right)\text{ and }C_{S}(p)=\frac{c}{\lambda}=\frac{2(p^{\ast}-p)}{p^{\ast}(1-p)}.
\end{equation*}
It can be shown that $C_{+}(p^{\ast};p^{\ast})=C_{S}(p^{\ast})$ and $2C_{+}^{\prime}(p^{\ast};p^{\ast})=C_{S}^{\prime}(p^{\ast})<0$; that is, as $p$ tends to $p^{\ast}$, the expected persuasion cost converges to $0$ faster under $R$-drifting than under stationary strategy.} The opposite is true, however, if the current belief is significantly
away from the persuasion target $p^{*}$. Intuitively, seeking persuasion
by an immediate success is more useful than slowly building up the
receiver's confidence in that case.

The confidence-spending strategy (which uses the $L$-drifting experiment
with jump target $p^{*}$) offers a similar trade-off as confidence preserving vis-a-vis confidence building.
If the current belief is far away from the persuasion target $p^{*}$,
confidence spending involves less time than confidence building. However,
there is another difference. If a success does not arise before the
belief falls to $p_{*}$, persuasion stops and the receiver chooses
$\ell$, before the belief reaches zero. By the familiar logic from
(static) Bayesian  persuasion, this leads to a suboptimal distribution
over posteriors. To avoid this, the sender may in some cases prefer
the confidence-building strategy, or in other cases switch from the
$L$-drifting experiment to the confidence-preserving strategy before
reaching $p_{\ast}$. As will be seen, the confidence-spending strategy
is also used in the stopping region $p<p^{*}$ as a ``Hail Mary pitch''
when the receiver is about to choose $\ell$ an instant later.

\subsection{Equilibrium Characterization\label{subsec:Equilibrium-Characterization}}

We now explain how the sender's equilibrium strategy deploys the three
modes of persuasion introduced in Section \ref{subsec:modes_persuasion},
and provide a full description of the unique SMPE strategy profile
for each set of parameter values and persuasion target $p^{*}$.

The structure of SMPE depends on two conditions. The first condition concerns how demanding the persuasion target $p^{\ast}$ is:
\begin{equation}
\tag{C1}p^{*}\leq\eta\approx0.943.\label{cond:permissive_persuasion_target}
\end{equation}
This condition determines whether the sender always prefers the $R$-drifting
strategy to the stationary strategy or not. The constant $\eta$ is
the largest value of $p^{*}$ such that the sender prefers the former
strategy to the latter for all $p<p^{*}$ (see Appendix \ref{par:Construction-of-eta} for a formal definition). Notice that this condition holds for $p^{*}$
not too large relative to $\hat{p}$; for instance, this is the case
when the sender's equilibrium strategy approximates the KG solution
(as long as $\hat{p}\leq\eta$).

The structure of the sender's equilibrium strategy also depends on
the following condition:
\begin{equation}
\tag{C2}v>U_{r}(p^{\ast})-U_{\ell}(p^{\ast}).\label{cond:low_receiver_rent}
\end{equation}
The left-hand side quantifies the sender's gains when she successfully
persuades the receiver and induces action $r$, while the right-hand
side represents the corresponding gains for the receiver.\footnote{\label{fn:relative_cost_discuss}As explained in Section \ref{sec:model} (see Footnote \ref{fn:asymmetric_costs}),
the payoffs of the two players are directly comparable, because their flow cost $c$ is normalized to be the same. With different flow costs, \eqref{cond:low_receiver_rent} has to be stated
using each player's payoff relative to their flow cost. In the extreme case when the sender's cost is zero but the receiver's is not, \eqref{cond:low_receiver_rent} necessarily holds, and the
equilibria characterized in Proposition \ref{prop:C2_hold} below
always exist. However, the sender is indifferent over all strategies
that yield the same (ex post) distribution of posteriors. Therefore,
the claim of uniqueness in Proposition \ref{prop:C2_hold} no longer
holds.} If \eqref{cond:low_receiver_rent} holds, then the sender has a stronger
incentive to experiment than the receiver has to listen, so the belief
$p_{\ast}$ below which some player wishes to stop is determined by
the receiver's incentives. Conversely, if \eqref{cond:low_receiver_rent}
fails, then the sender is less eager to experiment, and thus $p_{\ast}$
is determined by the sender's incentives.

We first provide an equilibrium characterization for the case where
\eqref{cond:low_receiver_rent} is satisfied.
\begin{prop}
\label{prop:C2_hold} Fix $p^{*}\in(\hat{p},1)$ and suppose that
$v>U_{r}(p^{\ast})-U_{\ell}(p^{\ast})$. For each $c>0$ sufficiently
small, there exists a unique SMPE such that the waiting region has
upper bound $p^{\ast}$. The waiting region is $W=[p_{\ast},p^{\ast})$
for some $p_{\ast}<\hat{p}$, and the sender's equilibrium strategy
is as follows:\footnote{We set $W=[p_{*},p^{*})$ to be a half-open interval, since for beliefs
$p<p_{*}$ close to $p_{*}$, the sender's best response is to target
$q=p_{*}$. Hence existence of the best response requires $p_{*}\in W$.}
\begin{enumerate}
\item[(a)] Suppose the belief is in the waiting region with $p\in[p_{*},p^{*})$.
\begin{enumerate}
\item[(i)] If $p^{\ast}\in(\hat{p},\eta)$, then the sender plays the $R$-drifting
strategy with left-jumps to $0$ for all $p\in[p_{\ast},p^{\ast})$.
\item[(ii)] If $p^{\ast}\in(\eta,1)$,\footnote{Notice that in the knife-edge case when $p^{\ast}=\eta$, there are
two SMPEs, one as in (a.i) and another as in (a.ii). In the latter,
however, $\overline{\pi}_{LR}=\xi$ and the $L$-drifting strategy
is not used in the waiting region. The two equilibria are payoff-equivalent
but exhibit very different dynamic behavior when $p_{0}\in[p_{\ast},\xi]$.} then there exist cutoffs $p_{\ast}<\xi<\overline{\pi}_{LR}<p^{\ast}$
such that for $p\in[p_{\ast},\xi)\cup(\overline{\pi}_{LR},p^{\ast})$,
the sender plays the $R$-drifting strategy with left-jumps to $0$;
for $p=\xi$, she uses the stationary strategy with jumps to $0$
and $p^{\ast}$; and for $p\in(\xi,\overline{\pi}_{LR}]$, she adopts
the $L$-drifting strategy with right-jumps to $p^{\ast}$.
\end{enumerate}
\item[(b)] Suppose the belief is outside the waiting region with $p<p_{*}$.
There exist cutoffs $0<\pi_{\ell L}<\pi_{0}<p_{*}$ such that for
$p\le\pi_{\ell L}$, the sender passes; for $p\in(\pi_{\ell L},\pi_{0})$,
she uses the $L$-drifting strategy with jumps to $q=p^{*}$; and
for $p\in[\pi_{0},p_{\ast})$, she uses the $L$-drifting strategy
with jumps to $q=p_{\ast}$.
\end{enumerate}
The lower bound $p_{*}$ of the waiting region converges to zero as
$c\rightarrow0$.
\end{prop}

Figure \ref{fig:C2_hold} below summarizes the sender's SMPE strategy
in Proposition \ref{prop:C2_hold}, depending on whether $p^{*}<\eta$
or not. If $p^{\ast}\in(\hat{p},\eta)$, then the sender uses only
$R$-drifting experiments in the waiting region $[p_{\ast},p^{\ast})$,
as depicted in the top panel of the figure. If $p^{\ast}>\eta$, then
the sender employs other strategies as well, as described in the bottom
panel of Figure \ref{fig:C2_hold}. For low beliefs close to $p_{*}$,
she starts with $R$-drifting (confindence-building) experiments but
switches to the stationary experiment when the belief reaches $\xi$.
For beliefs above $\xi$, but below $\overline{\pi}_{LR}$, she employs
$L$-drifting (confidence-spending) experiments and also switches
to the stationary experiment when the belief reaches $\xi$.

\begin{figure}[htp]
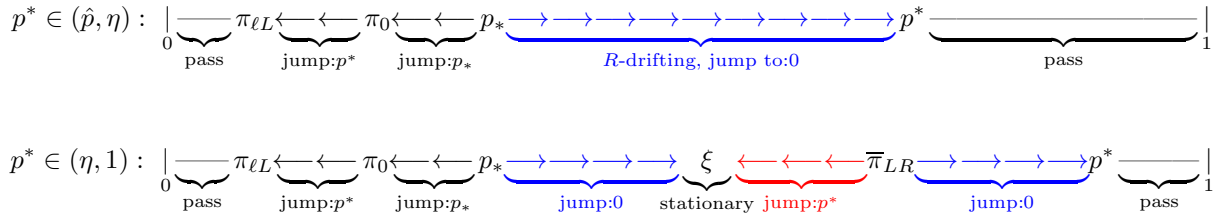

\begin{centering}
{\footnotesize{}{}{}{}{}{}{}{}{}{}{}{}
\[
p^{\ast}\in(\hat{p},\eta):~\underset{0}{|}\underbrace{\vphantom{\pi_{0}}\text{------}}_{\hspace{-20pt}\text{pass}\hspace{-20pt}}\pi_{\ell L}\underbrace{\vphantom{\pi_{0}}\hspace{-5pt}\longleftarrow\longleftarrow}_{\text{jump:}p^{*}}\pi_{0}\underbrace{\vphantom{\pi_{0}}\hspace{-5pt}\longleftarrow\longleftarrow}_{\text{jump:}p_{*}}p_{*}\textcolor{black}{\underbrace{{\longrightarrow\vphantom{\pi_{0}}\hspace{-5pt}\longrightarrow\longrightarrow\longrightarrow\longrightarrow\longrightarrow\longrightarrow\longrightarrow\longrightarrow}}_{\text{\ensuremath{R}-drifting, jump to:}0}}\,p^{*}\underbrace{\vphantom{\pi_{0}}\text{------------------------------}}_{\text{pass}}\underset{1}{|}
\]
}{\footnotesize\par}
\par\end{centering}
\begin{centering}

\par\end{centering}
\begin{centering}
\smallskip{}
 {\footnotesize{}{}{}{}{}{}{}{}{}{}{}{}
\[
p^{\ast}\in(\eta,1):~\underset{0}{|}\underbrace{\vphantom{\pi_{0}}\text{------}}_{\hspace{-20pt}\text{pass}\hspace{-20pt}}\pi_{\ell L}\underbrace{\vphantom{\pi_{0}}\hspace{-5pt}\longleftarrow\longleftarrow}_{\text{jump:}p^{*}}\pi_{0}\underbrace{\vphantom{\pi_{0}}\hspace{-5pt}\longleftarrow\longleftarrow}_{\text{jump:}p_{*}}p_{*}\textcolor{black}{\underbrace{{\longrightarrow\longrightarrow\longrightarrow\vphantom{\pi_{0}}\hspace{-5pt}\longrightarrow}}_{\text{jump:}0}}\underbrace{\vphantom{\pi_{0}}\hspace{-5pt}\xi\hspace{-5pt}}_{\text{\hspace{-0.3cm}stationary\hspace{-0.3cm}}}\textcolor{black}{ \underbrace{{\longleftarrow\vphantom{\pi_{0}}\hspace{-5pt}\longleftarrow\longleftarrow}}_{\text{jump:} p^{*}} }\overline{\pi}_{LR}\textcolor{black}{\underbrace{{\longrightarrow\vphantom{\pi_{0}}\hspace{-5pt}\longrightarrow\longrightarrow\longrightarrow}}_{\text{jump:}0}}p^{*}\underbrace{\vphantom{\pi_{0}}\text{---------}}_{\text{pass}}\underset{1}{|}
\]
}{\footnotesize\par}
\par\end{centering}
\begin{centering}

\par\end{centering}
\caption{The sender's SMPE strategies in Proposition \ref{prop:C2_hold}, that
is, when $v>U_{r}(p^{\ast})-U_{\ell}(p^{\ast})$.\label{fig:C2_hold}}
\end{figure}

To understand these different patterns, recall from Section \ref{subsec:modes_persuasion}
that the $R$-drifting experiment is particularly useful if it does
not take too long to build the receiver's confidence and move the
belief to $p^{*}$. This explains the use of $R$-drifting experiment
when $p$ is rather close to $p^{\ast}$, for $p\in[\overline{\pi}_{LR},p^{\ast})$
if $p^{\ast}\geq\eta$ and for all $p$ in the waiting region if $p^{\ast}<\eta$. If $p^{*}$ is above $\eta$, then for $p$
below $\overline{\pi}_{LR}$, other experiments become optimal. For
$p<\xi$, the sender starts by building confidence, but instead of
continuing with this strategy until $p^{*}$ is reached, she cuts
it short and switches to the stationary strategy when $\xi$ is reached.
At $\xi$, the arrival rate of a jump to $p^{*}$ in the stationary
experiment is sufficiently high to yield a faster persuasion (on average)
than it would take to gradually build confidence to $p^{\ast}$ using
the $R$-drifting strategy. For beliefs $p\in(\xi,\overline{\pi}_{LR})$,
a jump to $p^{*}$ arrives at a higher rate, so that it becomes optimal
to spend confidence and use only the $L$-drifting experiment, rather
than preserving confidence with the stationary experiment.

For an economic intuition, consider a salesperson courting a potentially
interested buyer. If the buyer needs only a bit more reassurance to
buy the product, then the salesperson should carefully build up the
buyer's confidence until the belief reaches $p^{*}$. The salesperson
may still ``slip off'' and lose the buyer (i.e., $p$ jumps down
to $0$). But most likely, the salesperson ``weathers'' that risk
and moves the buyer over the last hurdle (i.e., $q=p^{\ast}$ is reached).
This is exactly what our equilibrium persuasion dynamics describes
when $p_{0}$ is close to $p^{\ast}$. When the buyer does not require
a high degree of confidence to be persuaded ($p^{*}\le\eta$), building
up confidence is the optimal strategy for the salesperson whenever
the buyer is initially willing to listen (i.e., $p_{0}$ is in the
waiting region). By contrast, when $p^{*}>\eta$, the buyer requires
a lot of convincing and there are beliefs where the buyer is rather
uninterested (as in a ``cold call''). Then, the salesperson's optimal
strategy depends on how skeptical the buyer is initially. If $p_{0}\in[\overline{\pi}_{LR},p^{\ast})$,
then it is still an optimal strategy for the salesperson to build
up the buyer's confidence until $p^{\ast}$. If $p_{0}\in(p_{*},\xi)$,
the salesperson first tries to build confidence. If the buyer is still
listening when the belief reaches $\xi$, the seller becomes more
convinced that the buyer can be persuaded, and she starts using a
big pitch that would move the belief to $p^{*}$. For higher beliefs,
she is even more convinced that the buyer can be persuaded quickly,
so she ``spends confidence'' and concentrates all her efforts on
quickly persuading the receiver.

Condition \eqref{cond:low_receiver_rent} means that the lower bound
$p_{*}$ of the waiting region is determined by the receiver's incentive: $p_{\ast}$ is the point at which the receiver is indifferent between taking action $\ell$ immediately and waiting (i.e., $U_{\ell}(p_{\ast})=U(p_{\ast})$,
where $U(p)$ is the receiver's payoff from experimentation). Intuitively, \eqref{cond:low_receiver_rent} suggests that the receiver gains less from experimentation, and is thus less willing to continue, than the sender. Therefore, at the lower bound $p_{\ast}$, the receiver wants
to stop, even though the sender wants to continue persuading the receiver
(i.e., $V(p_{\ast})>0$).

When $p<p_{\ast}$, the sender plays only $L$-drifting experiments,
unless she prefers to pass (i.e., when $p<\pi_{\ell L}$). This is
intuitive, because the receiver takes action $\ell$ immediately unless
the sender generates an instantaneous jump, forcing the sender to
effectively make a ``Hail Mary'' pitch. It is intriguing, though,
that the sender's target posterior can be either $p_{\ast}$ or $p^{\ast}$,
depending on how close $p$ is to $p_{\ast}$: in the sales context
used above, if the buyer is fairly skeptical, then the salesperson
needs to use a big pitch. But, depending on how skeptical the buyer
is, she may try to get enough attention only for the buyer to stay
engaged (targeting $q=p_{\ast}$) or use an even bigger pitch to convince
the buyer to buy outright (targeting $q=p^{\ast}$). If $p$ is just
below $p_{\ast}$ (see $p_{2}$ in Figure \ref{fig:HJB_simplified}),
then the sender can jump into the waiting region at a high rate: recall
that the arrival rate of a jump to $p_{\ast}$ grows to infinity as
$p$ tends to $p_{\ast}$. In this case, it is optimal to target $p_{\ast}$,
thereby maximizing the arrival rate of Poisson jumps: the salesperson
is sufficiently optimistic about her chance of grabbing the buyer's
attention, so she only aims to make the buyer stay. If $p$ is rather
far away from $p_{\ast}$ (below $\pi_{0}$ such as $p_{1}$ in Figure
\ref{fig:HJB_simplified}), then the sender does not enjoy a high
arrival rate. In this case, it is optimal to maximize the sender's
payoff conditional on Poisson jumps, which she gets by targeting $p^{\ast}$:
the salesperson tries to sell her product right away and if it does
not succeed, then she just lets it go.

Next, we provide an equilibrium characterization for the case when
\eqref{cond:low_receiver_rent} is violated.
\begin{prop}
\label{prop:C2_fail} Fix $p^{\ast}\in(\hat{p},1)$ and assume that
$v\le U_{r}(p^{\ast})-U_{\ell}(p^{\ast})$. For each $c>0$ sufficiently
small, there exists a unique SMPE such that the waiting region has
upper bound $p^{\ast}$. The waiting region is $W=(p_{\ast},p^{\ast})$
for some $p_{*}<\hat{p}$, and the sender's equilibrium strategy is
as follows:\footnote{We set $W=(p_{*},p^{*})$ to be an open interval, since the sender
uses the $L$-drifting strategy for beliefs close to $p_{*}$. Including
$p_{*}$ would not lead to a well-defined stopping time and therefore
violates admissibility.}
\begin{enumerate}
\item Suppose the belief is in the waiting region with $p\in(p_{*},p^{*})$.
\begin{enumerate}
\item If $p^{\ast}\in(\hat{p},\eta)$, then there exists a cutoff $\pil_{LR}\in W$ such that for $p\in(\pil_{LR},p^{\ast})$, the sender uses the $R$-drifting
strategy with left-jumps to $0$; and for $p\in(p_{\ast},\pil_{LR})$,
she uses the $L$-drifting strategy with right-jumps to $p^{\ast}$.
\item If $p^{\ast}\in(\eta,1)$, then there exist cutoffs $p_{\ast}<\underline{\pi}_{LR}<\xi<\overline{\pi}_{LR}<p^{\ast}$
such that for $p\in[\underline{\pi}_{LR},\xi)\cup[\overline{\pi}_{LR},p^{\ast})$,
the sender plays the $R$-drifting strategy with left-jumps to $0$;
for $p=\xi$, she adopts the stationary strategy with jumps to $0$
or $p^{\ast}$; and for $p\in(p_{*},\underline{\pi}_{LR})\cup(\xi,\overline{\pi}_{LR})$,
she uses the $L$-drifting strategy with right-jumps to $p^{\ast}$.
\end{enumerate}
\item If the belief is outside the waiting region, the sender passes.
\end{enumerate}
The lower bound of the waiting region $p_{*}$ converges to zero as
$c$ tends to $0$.
\end{prop}

\begin{figure}
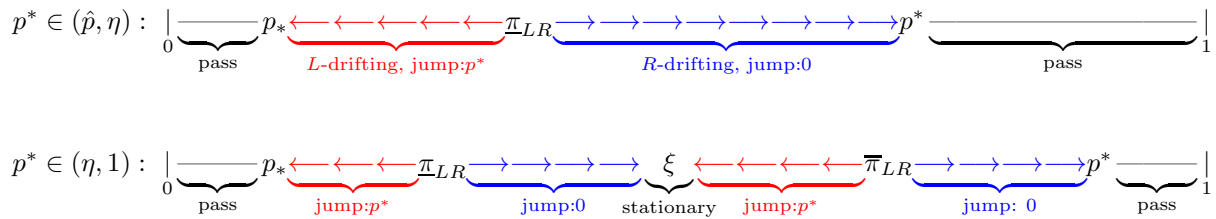

{\footnotesize{}{}{}{}{}{}{}{}{}{}{}{}
\[
p^{\ast}\in(\hat{p},\eta):~\underset{0}{|}\underbrace{\vphantom{\pi_{0}}\text{---------}}_{\hspace{-20pt}\text{pass}\hspace{-20pt}}p_{*}\textcolor{black}{ \underbrace{{\longleftarrow\vphantom{\pi_{0}}\hspace{-5pt}\longleftarrow\longleftarrow\longleftarrow\longleftarrow}}_{L\text{-drifting, jump:} p^{*}} }\pil_{LR}\textcolor{black}{\underbrace{{\longrightarrow\vphantom{\pi_{0}}\hspace{-5pt}\longrightarrow\longrightarrow\longrightarrow\longrightarrow\longrightarrow\longrightarrow\longrightarrow}}_{\text{\ensuremath{R}-drifting, jump:} 0}}p^{*}\underbrace{\vphantom{\pi_{0}}\text{------------------------------}}_{\text{pass}}\underset{1}{|}
\]
}{\footnotesize\par}

\smallskip{}
 {\footnotesize{}{}{}{}{}{}{}{}{}{}{}{}
\[
p^{\ast}\in(\eta,1):~\underset{0}{|}\underbrace{\vphantom{\pi_{0}}\text{---------}}_{\hspace{-20pt}\text{pass}\hspace{-20pt}}p_{*}\textcolor{black}{ \underbrace{{\longleftarrow\vphantom{\pi_{0}}\hspace{-5pt}\longleftarrow\longleftarrow}}_{\text{jump:} p^{*}} }\underline{\pi}_{LR}\textcolor{black}{\underbrace{{\longrightarrow\longrightarrow\longrightarrow\vphantom{\pi_{0}}\hspace{-5pt}\longrightarrow}}_{\text{jump:} 0}}\underbrace{\vphantom{\pi_{0}}\hspace{-5pt}\xi\hspace{-5pt}}_{\text{\hspace{-0.3cm}stationary\hspace{-0.3cm}}}\textcolor{black}{ \underbrace{{\vphantom{\pi_{0}}\hspace{-5pt}\longleftarrow\longleftarrow\longleftarrow\longleftarrow}}_{\text{jump:} p^{*}} }\overline{\pi}_{LR}\textcolor{black}{\underbrace{{\longrightarrow\vphantom{\pi_{0}}\hspace{-5pt}\longrightarrow\longrightarrow\longrightarrow}}_{\text{jump: } 0}}p^{*}\underbrace{\vphantom{\pi_{0}}\text{---------}}_{\text{pass}}\underset{1}{|}
\]
}{\footnotesize\par}

\caption{The sender's SMPE strategy in Proposition \ref{prop:C2_fail}, that
is, when $v\protect\leq U_{r}(p^{\ast})-U_{\ell}(p^{\ast})$.\label{fig:C2_fail} }
\end{figure}

Figure \ref{fig:C2_fail} describes the persuasion dynamics in Proposition \ref{prop:C2_fail}. There are two main differences from Proposition \ref{prop:C2_hold}.
First, if $p<p_{\ast}$ then the sender simply passes, whereas in
Proposition \ref{prop:C2_hold}, the sender uses $L$-drifting experiments
when $p\in(\pi_{\ell L},p_{\ast})$. Second, when $p$ is just above
$p_{\ast}$, the sender adopts $L$-drifting experiments, and thus
the game may stop at $p_{\ast}$. By contrast, in Proposition \ref{prop:C2_hold},
the sender always plays $R$-drifting experiments just above $p_{\ast}$,
and the game never ends with the belief reaching $p_{\ast}$. Both
of these differences are precisely due to the failure of \eqref{cond:low_receiver_rent}:
if $v\le U_{r}(p^{\ast})-U_{\ell}(p^{\ast})$ then the sender is less
willing to continue than the receiver, and thus $p_{\ast}$ is determined
by the sender's participation constraint (i.e., $V(p_{\ast})=0$).
Therefore, the sender has no incentive to experiment once $p$ falls
below $p_{\ast}$.

When $p$ is just above $p_{\ast}$, the sender goes for a big pitch
by targeting $p^{*}$ with $L$-drifting experiments. The sender does
not mind losing the buyer's confidence in the process, since the violation
of \eqref{cond:low_receiver_rent} means that, as the belief nears
$p_{*}$, she has very little motivation left for persuading the receiver
even though the latter remains willing to listen. By contrast, when
\eqref{cond:low_receiver_rent} holds (as in Proposition \ref{prop:C2_hold}),
as the belief nears $p_{*}$, the receiver loses interest in listening,
but the sender still sees a significant value in staying ``in the
game.'' Hence, the sender tries to build, instead of running down,
the receiver's confidence in that case.

\section{Concluding Discussions\label{sec:conclusion}}

We conclude by discussing how our results depend on several modelling assumptions and suggesting a few directions for future research.

\paragraph*{Binary actions and states.} We have considered the canonical Bayesian persuasion problem with two states and two actions. Some of our results {clearly} depend on specific features of the problem. However, our main economic insights hold more generally. In fact, it is often straightforward to modify our technical analysis for other persuasion problems.

To be concrete, consider an extension in which the receiver has one additional action $M$ and the players' payoffs are as depicted in Figure \ref{fig:three_actions}. Specifically, $M$ is the receiver's optimal action when the belief belongs to the intermediate range $[p^{\dagger},\hat p]$, and the sender earns $\underline{v}$ if the receiver takes action $M$. Assume $\underline{v}<p^{\dagger}v/\hat p$, so that the KG solution is still to induce two posteriors, $0$ and $\hat p$, whenever $p_{0}<\hat p$.

\begin{figure}
\raggedright{}\centering{}\beginpgfgraphicnamed{three_actions}
\begin{tikzpicture}[scale=1]

\draw[line width=0.5pt] (0, 4.3) -- (0,0) -- (6,0) -- (6,4.3);

\fill (0,4) node[left] {\footnotesize{$v$}};
\fill (0,2) node[left] {\footnotesize{$p^{\dagger}v/\hat p$}};

\fill (0,1) node[left] {\footnotesize{$\underline{v}$}};

\draw[loosely dotted] (4,0)--(4,4);
\draw[loosely dotted] (0,4)--(4,4);
\draw[loosely dotted] (0,2)--(2,2)--(2,0);
\draw[loosely dotted] (0,1)--(2,1);
\draw[loosely dotted] (1,0)--(1,1);

\draw[line width=0.7pt] (0,0) -- (2,0);
\draw[line width=0.7pt,dotted] (2,0) -- (2,1);
\draw[line width=0.7pt] (2,1) -- (4,1);
\draw[line width=0.7pt,dotted] (4,1) -- (4,4);
\draw[line width=0.7pt] (4,4) -- (6,4);

\draw[dashed,line width=0.5pt,color=black] (0,0) -- (4,4) -- (6,4);

\fill (0,0) node[below] {\footnotesize{$0$}};
\fill (2,0) node[below] {\footnotesize{$p^{\dagger}$}};
\fill (1,-0.1) node[below] {\footnotesize{$p_{0}$}};
\fill (4,0) node[below] {\footnotesize{$\hat p$}};
\fill (6,0) node[below] {\footnotesize{$1$}};

\draw[fill=black] (4,4) circle (2pt);
\draw[fill=black] (0,0) circle (2pt);
\draw[fill=black] (1,1) circle (2pt);


\draw (3,4.6) node {\textbf{Sender}};

\draw[xshift=7.5cm] (0, 4.3) -- (0,0) -- (6,0) -- (6,4.3);

\fill[xshift=7.5cm] (0,3.5) node[left] {\footnotesize{$u_{\ell}^{L}$}};
\fill[xshift=7.5cm] (0,0) node[left] {\footnotesize{$u_{r}^{L}$}};
\fill[xshift=7.5cm] (0,7/3) node[left] {\footnotesize{$u_{M}^{L}$}};

\fill[xshift=7.5cm] (6,0.5) node[right] {\footnotesize{$u_{\ell}^{R}$}};
\fill[xshift=7.5cm] (6,4) node[right] {\footnotesize{$u_{r}^{R}$}};
\fill[xshift=7.5cm] (6,17/6) node[right] {\footnotesize{$u_{M}^{R}$}};

\draw[xshift=7.5cm,dashed,color=gray] (0,0)--(6,4);
\draw[xshift=7.5cm,dashed,color=gray] (0,3.5)--(6,0.5);
\draw[xshift=7.5cm,dashed,color=gray] (0,7/3)--(6,17/6);

\draw[xshift=7.5cm,line width=0.7pt] (0,3.5) --(2,5/2) -- (4,8/3)-- (6,4);

\draw[xshift=7.5cm] (3,4.6) node {\textbf{Receiver}};

\fill[xshift=7.5cm] (0,0) node[below] {\footnotesize{$0$}};
\fill[xshift=7.5cm] (2,0) node[below] {\footnotesize{$p^{\dagger}$}};
\fill[xshift=7.5cm] (1,-0.1) node[below] {\footnotesize{$p_{0}$}};
\fill[xshift=7.5cm] (4,0) node[below] {\footnotesize{$\hat p$}};
\fill[xshift=7.5cm] (6,0) node[below] {\footnotesize{$1$}};

\draw[xshift=7.5cm,loosely dotted] (2,0)--(2,5/2);
\draw[xshift=7.5cm,loosely dotted] (4,0)--(4,8/3);
\draw[xshift=7.5cm,dotted] (1,0)--(1,79/24);

\draw[xshift=7.5cm,fill=black] (0,7/2) circle (2pt);
\draw[xshift=7.5cm,fill=black] (2,5/2) circle (2pt);
\draw[xshift=7.5cm,fill=black] (4,8/3) circle (2pt);

\draw[xshift=7.5cm,fill=white] (1,3) circle (2pt);

\draw[xshift=7.5cm,black,dashed] (0,3.5)--(4,8/3)--(6,4);

\draw[xshift=7.5cm,fill=white] (1,79/24) circle (2pt);


\end{tikzpicture}
\endpgfgraphicnamed\caption{\label{fig:three_actions} Payoffs from static persuasion when there is a middle action, $M$. Solid curves:
payoffs without persuasion (information). Dashed curve: the sender's
expected payoff in the KG solution.}
\end{figure}

An important change from our baseline environment is that the receiver does enjoy rents from the KG solution; observe that in the right panel of Figure \ref{fig:three_actions}, the dashed curve strictly exceeds the solid curve whenever $p\in(0,\hat p)$. In this case, it is possible to construct an SMPE {\it exactly} implementing the KG solution with $p^{\ast}=\hat p$, for $c$ sufficiently small. More importantly, Theorem \ref{thm:persuasion-failure} no longer holds: if $c$ is sufficiently small then the receiver prefers to wait when the sender plays an $L$-drifting experiment targeting $\hat p$ at $p_{0}(<\hat p)$.\footnote{A tempting conjecture may be that  {the no-persuasion} equilibrium is unsustainable if {the} KG solution offers strictly positive rents to the receiver. This need not be the case. If $p_0$ is {slightly below} $p^{\dagger}$, then it becomes credible (in the sense of satisfying our refinement) that the receiver stops immediately and the sender uses the $L$-drifting experiment with jump target $p^{\dagger}$ (and not $p^*$).} Meanwhile, Theorem \ref{thm:folk} remains valid: for $c$ sufficiently small and $p^{\ast}$ exceeding $\hat p$, the SMPE we construct for our baseline environment in Section \ref{sec:persuasion_dynamics} continues to be an SMPE in this extended problem. Therefore, our arguments in Section \ref{sec:folk} apply unchanged.

Relaxing the binary-state assumption raises a few significant challenges---such as defining the set of feasible experiments and analyzing a system of partial differential equations---which prevents us from providing a tight and comprehensive equilibrium characterization. Nevertheless, at least conceptually, it is not hard to see how our main insights would extend to the environment with more than two states. The no-persuasion equilibrium in Theorem \ref{thm:persuasion-failure} would exist if and only if the receiver never earns strictly positive (instantaneous) rents from the sender's optimal flow experiment. By contrast, if there is a (lump-sum) Blackwell experiment that strictly benefits both players, then the resulting outcome would be approximated by equilibria of the dynamic Persuasion model.

\paragraph*{Other features of the model.}\label{subsec:other_features}

We have restricted attention to Markov perfect equilibria, which by definition do not rely on incentives provided by off-path punishments. Certainly, other (non-Markov) equilibria could be used so as to enlarge the set of sustainable payoffs.\footnote{As is well known, it is technically challenging to define a game in continuous time without Markov restrictions \citep[see, e.g.][]{Simon1989}. Our subsequent discussion should be understood as referring to the limit of discrete-time equilibria.} Then, it seems plausible that as the players' persuasion costs vanish, one could implement all individually rational payoffs, including the dotted region in Figure \ref{fig:folk2}. For prior beliefs $p_{0}<\hat{p}$, this is indeed the case, because the no-persuasion equilibrium in Theorem \ref{thm:persuasion-failure} can be used to most effectively control the sender's incentives. For prior beliefs $p_{0}>\hat{p}$, however, no clear punishment equilibrium is available; note that for $p_{0}>\hat{p}$, the no-persuasion equilibrium \emph{maximizes} the sender's payoff. This suggests that our construction
of MPEs cannot be replaced by arguably simpler constructions that
rely on off-path punishments. Indeed, we conjecture that for $p_{0}>\hat{p}$, Theorem \ref{thm:folk} and Corollary \ref{prop:payoff_vectors} characterize the full set of equilibrium payoffs.

Our model assumes flow persuasion costs rather than discounting. This
assumption simplifies the analysis, mainly by additively separating
persuasion benefits from persuasion costs. Still, it has no qualitative
impact on our main results. Specifically, if we include both flow
costs and discounting in the analysis, then the resulting SMPEs would
converge to those of our current model as discounting becomes negligible.
If we consider only discounting (without flow costs), then the persuasion
dynamics needs some modification. Among other things, the sender has
no reason to voluntarily stop experimentation, and thus the persuasion
dynamics will be similar to that of Proposition \ref{prop:C2_hold}
(as opposed to that of Proposition \ref{prop:C2_fail}).\footnote{Specifically, the lower bound $p_{\ast}$ of the waiting region will
be determined by the receiver's incentives. In addition, at the lower
bound $p_{\ast}$, so as to stay within the waiting region, the sender
will play either $R$-drifting experiments or the stationary strategy.
This latter fact implies that if the game starts from $p_{0}\in[p_{\ast},p^{\ast})$,
then it will end only when the belief reaches either $0$ or $p^{\ast}$,
and thus the persuasion probability will always be equal to $p_{0}/p^{\ast}$.} Still, our main economic lessons will continue to apply: all three
theorems in Section \ref{sec:folk} would continue to hold.\footnote{The proofs of Theorems \ref{thm:persuasion-failure} and \ref{thm:infeasible_MPE_payoffs}
can be readily modified. For Theorem \ref{thm:folk}, it is easy to
show that the main economic logic behind it (namely, ``the power
of beliefs'' explained at the end of Section \ref{sec:illustate})
holds unchanged with discounting.} Furthermore, the relative advantages of the three main modes of persuasion
remain unchanged, so the persuasion dynamics are in many cases similar
to those described in Section \ref{sec:persuasion_dynamics}.

Our continuous-time game has a straightforward discrete-time analogue and can be interpreted as its limit. In a discrete-time model, however, it becomes important whether the receiver's per-period listening cost is independent of the amount of information the sender generates or proportional to it. In the former (independent) case, our ``power-of-beliefs'' logic no longer holds: if the current belief $p$ is just below the persuasion target $p^{\ast}$ then the receiver's gains from waiting one more period are close to $0$, in which case she would prefer to stop at $p<p^{\ast}$. Thus, any equilibrium with persuasion target  $p^{\ast}>\hat p$ would unravel, leaving $p^{\ast}=\hat p$ as the only feasible persuasion target. In our baseline model, this renders the no-persuasion equilibrium the unique SMPE. However, if the KG solution provides positive rents for the receiver, as exemplified in the case with three actions depicted in Figure \ref{fig:three_actions}, any persuasion target $p^{\ast}>\hat p$ can still be supported as SMPE as $c$ tends to $0$. Meanwhile, if the receiver's listening cost is proportional to the amount of new information, he would still be willing to wait, no matter how close $p$ is to $p^{\ast}$. Then, all our analysis and results continue to hold in the discrete-time analogue even in the baseline model.\footnote{The same logic applies when there is discounting in terms of the period length $\Delta$. If $\Delta$ is independent of the amount of new information, then all persuasive SMPEs with $p^{\ast}>\hat p$ unravel. However, if $\Delta$ is proportional to the amount of information---a sensible assumption if $\Delta$ describes information processing time---then such unraveling does not occur, and our analysis goes through unchanged.}

Our model focuses on generalized Poisson experiments to accommodate rich and flexible information choice. By contrast, an alternative
such as the drift-diffusion model does not allow for such richness.
For example, in \citet{Henry2017}, the sender samples from a fixed
exogenous process, without choosing the type of experiment. Nevertheless,
the logic that gives rise to our Theorem \ref{thm:folk}---namely, the incentivizing power of equilibrium beliefs---applies equally well to such models
(see Footnote \ref{fn:Brownian-Motion}).

\paragraph*{{Directions for future research.}}

The key features of our model are that real information takes time to generate, and that neither the sender nor the receiver has commitment
power over future actions. There are several avenues along which one
could vary these features. For example, one may consider a model in
which the sender faces the same flow information constraint as in
our model but has full commitment power over her dynamic strategy:
given our discussion in Section \ref{sec:illustate}, it is straightforward
that the sender can approximately implement the KG outcome. However,
it is non-trivial to characterize the sender's optimal dynamic strategy. Alternatively, one could further relax the commitment power by allowing
the receiver to observe only the outcome of the flow experiment, but
not the experiment itself.

More broadly, the rich persuasion dynamics found in our model owe
a great deal to the general class of Poisson experiments we allow
for. At first glance, allowing for the information to be chosen from
such a rich class of experiments \emph{at each point in time} might
appear extremely complex to analyze, and a clear analysis might seem
unlikely. Yet, the model produced a remarkably precise characterization
of the sender's optimal choice of information---namely, not just
\emph{when} to stop  {providing} information but more importantly \emph{what type} of information to {generate}. This modeling innovation may fruitfully apply to other dynamic settings.

\appendix

\section{Further Characterization on Feasible Experiments\label{appendix:flow_bound}}

This appendix formally proves Lemma \ref{lem:info} and also provides an alternative belief-based characterization for the set $\mathcal{P}^{\ast}$ of feasible experiments.

\begin{proof}[Proof of Lemma \ref{lem:info}]
Fix $\langle p_t\rangle \in \mathcal{P}^{\ast}$ and any $t \in\mathbb{R}_{+}$. For each $q\neq p$, let $\gamma(q,p)$ denote the unconditional arrival rate of posterior belief $q$ given $p_{t-}=p$. For these values to be well-defined, it is necessary and sufficient that the associated conditional likelihoods $(\lambda^{L}(q,p),\lambda^{R}(q,p))$ satisfy
\begin{equation*}
    q=\frac{p\lambda^{R}(q,p)}{(1-p)\lambda^{L}(q,p)+p\lambda^{R}(q,p)}\text{ and }\gamma (q,p)=p\lambda^R(q,p)+(1- p)\lambda^L(q,p).
\end{equation*}
Solving this system of equations, we obtain
\begin{equation*}
    \lambda^{R}(q,p)=\gamma(q,p)\frac{q}{p}\text{ and }\lambda^{L}(q,p)=\gamma(q,p)\frac{1-q}{1-p}.
\end{equation*}
Then, our information constraint can be written as
\begin{equation}\label{eq:info_constraint_lem1}
\sum_{q\neq p}|\lambda^{R}(q,p)-\lambda^{L}(q,p)|=\sum_{q\neq p}\gamma(q,p)\Big | \frac{q}{p}-\frac{1-q}{1-p}\Big |=\sum_{q\neq p}\gamma(q,p)\frac{|q-p|}{p(1-p)}\leq \lambda.
\end{equation}
For each $q$, define $\alpha(q):=\gamma(q,p)\frac{|q-p|}{\lambda p(1-p)}$. Then, the above constraint can be equivalently written as $\sum_{q\neq p}\alpha(q)\leq 1$, and the arrival rate of posterior $q$ given $p$ is given by $\gamma(q,p)=\alpha(q) \frac{ \lambda p(1-p)}{|q-p|}$.

Let $\dot{p}$ denote the instantaneous change of $\langle p_t\rangle$ conditional on no jump. Since $\langle p_t\rangle$ is a martingale, $\sum_{q\neq p}\gamma(q,p)(q-p)+\dot{p}=0$, so $\dot{p}$ satisfies
\begin{eqnarray*}
    \dot{p}&=&-\sum_{q\neq p}\gamma(q,p)(q-p)=-\sum_{q>p}\gamma(q,p)(q-p)-\sum_{q<p}\gamma(q,p)(q-p)\\
    &=&-\sum_{q>p}\gamma(q,p)|q-p|+\sum_{q<p}\gamma(q,p)|q-p|\\
    &=&-\sum_{q>p}\alpha(q)\lambda p(1-p)+\sum_{q<p}\alpha(q)\lambda p(1-p)\\
    &=&-\left(\sum_{q>p}\alpha(q)-\sum_{q<p}\alpha(q)\right)\lambda p(1-p).
\end{eqnarray*}
\end{proof}

Next, we provide an additional characterization of $\mathcal{P}^{\ast}$ based on a measure of information. Let $\langle p_{t} \rangle$ denote a regular martingale process in $\mathcal{P}$. For each $t$, $p_{t-}:=\lim_{t^{\prime}\uparrow t}p_{t^{\prime}}$, and $q\neq p_{t-}$, let $\gamma(q,p_{t-})$ denote the rate at which the belief jumps from $p_{t-}$ to $q$; formally,
\begin{equation*}
    \gamma(q,p_{t-}):=\lim_{dt\to 0}\frac{\mathbb{P}[p_{t}=q|p_{t-dt}]}{dt}.
\end{equation*}
We measure the amount of flow information of $\langle p_{t} \rangle$ at each point in history by
\begin{equation*}
    \mathcal{I}(p_{t-}):=\sum_{q\neq p_{t-}}\gamma(q,p_{t-})|p_{t}-p_{t-}|.
\end{equation*}
In other words, our information measure $\mathcal{I}(p_{t-})$ quantifies the total absolute change of the belief process at each point in time.

By \eqref{eq:info_constraint_lem1} in the proof of Lemma \ref{lem:info}, our information constraint can be written as
\begin{equation*}
    \sum_{q\neq p}|\lambda^{R}(q,p)-\lambda^{L}(q,p)|=\sum_{q\neq p}\gamma(q,p)\frac{|q-p|}{p(1-p)}\leq \lambda\Leftrightarrow \mathcal{I}(p)\leq \lambda p(1-p).
\end{equation*}
This implies that the set $\mathcal{P}^{\ast}$ of feasible experiments can be equivalently defined as
\begin{equation*}
    \mathcal{P}^{\ast}:=\{\langle p_{t} \rangle\in\mathcal{P}:\mathcal{I}(p_{t-})\leq \lambda p_{t-}(1-p_{t-})\text{ for all $t$ and $p_{t-}$}\}.
\end{equation*}
In other words, we consider belief processes whose aggregate change at each point in history is bounded by $\lambda p_{t-}(1-p_{t-})$; note that $p(1-p)$ is equal to the variance of the the Bernoulli random variable $p$. The bound's dependence on $p_{t-}$ is natural given that $p_{t-}\in[0,1]$ and $p_{t-}=0,1$ represents perfect information from which no belief change should be feasible; more generally, it captures an intuitive idea that the sender can move the receiver's belief more, the more uncertain the state is.

\section{Admissible Strategies\label{appendix:admissible}}

This appendix completes the definition of our continuous-time game by defining admissible strategies for the sender. We note that this appendix is similar to Appendix B.1 of \citet{klein/rady:11}: the two models have the same underlying technical issues and natural resolutions to them.\footnote{The difference is that the technical problems arise in their model because the evolution of beliefs is jointly controlled by two players, while in our model, it is because the sender can choose from a large set of Poisson experiments.}

Recall that the sender's strategy is a measurable function $\sigma^{S}$ that assigns a flow experiment $\sigma^{S}(p)=\left(\alpha(q;p)\right)_{q\in[0,1]}$ to each belief $p\in[0,1]$. As noted, the strategy induces a belief process satisfying:
\begin{equation}\label{eq:pt_integral}
\dot{p}_{t}=-\beta(p_{t})\lambda p_{t}(1-p_{t}),
\end{equation}
where
\begin{equation*}
    \beta(p):=\sum_{q>p}\alpha(q;p)-\sum_{q<p}\alpha(q;p).
\end{equation*}
Note that $p_{t}$ moves leftward if $\beta(p_{t})>0$ and rightward if $\beta(p_{t})<0$.

\begin{defn}\label{def:sender_admissible}
A measurable function $\sigma^{S}$ is an admissible strategy for the sender if for all $p_{0}\in[0,1]$, there exists a solution to \eqref{eq:pt_integral}.
\end{defn}

To see the role of  Definition \ref{def:sender_admissible}, first observe that for $\sigma^{S}$ with a relatively simple structure, we can find an explicit solution to \eqref{eq:pt_integral}. For example, if the sender plays only the $R$-drifting experiment, then $\beta(p)=-1$ for all $p$, in which case $p_{t}=\frac{p_{0}e^{\lambda t}}{p_{0}e^{\lambda t}+1-p_{0}}$. If the sender plays only the stationary experiment, then $\beta(p)=0$ for all $p$, in which case $p_{t}=p_{0}$. Of course, the differential equation \eqref{eq:pt_integral} cannot be solved explicitly in general. One may utilize a sufficient condition on $\beta(\cdot)$: for example, it suffices that $\beta(\cdot)$ is continuous or satisfies Carath\'{e}odory conditions \citep[see][]{goodman:70}. For our purpose, however, imposing such a sufficient condition is unnecessarily restrictive. Therefore, we require only that there is a solution to \eqref{eq:pt_integral}.  More precisely, we shall require a couple of conditions, one of which is necessary and the other is of no material consequence. This approach is valid, since the equilibrium with these weaker conditions will ensure that  \eqref{eq:pt_integral} is well defined for all $p\in [0,1]$.

To explain the necessary condition that is relevant for our context, consider, for example, a strategy such that the sender plays the $R$-drifting experiment targeting $0$ (so $\beta(p)=-1$) whenever $p\leq p_{0}$ and the $L$-drifting experiment targeting $1$ (so $\beta(p)=1$) whenever $p>p_{0}$. As depicted in the left panel of Figure \ref{fig:admissible}, the belief moves toward $p$ whether it is below or above $p$, so  $\beta(p_0)=-1$ results in  \eqref{eq:pt_integral} being ill-defined at $p_0$.  In fact, $p_{t}$ should stay constant if starting from $p_{0}$. Hence, admissibility requires  $\sigma^{S}(p_{0})$ to satisfy $\beta(p_{0})=0$.

\begin{figure}
\raggedright{}\centering{}\beginpgfgraphicnamed{admissible}
\begin{tikzpicture}[scale=1]

\draw[line width=0.5pt] (0,0)--(6,0);
\draw[fill=black] (0,0) circle (2pt);
\draw[fill=black] (6,0) circle (2pt);
\fill (0,0) node[left] {\footnotesize{$0$}};
\fill (6,0) node[below] {\footnotesize{$1$}};
\fill (6,0) node[right] {\footnotesize{$p$}};

\draw[fill=white] (3,0) circle (2pt);

\draw[line width=0.5pt,dotted] (3,-1.2)--(0,-1.2)--(0,1.2);
\fill (0,-1.2) node[left] {\footnotesize{$\beta(p)$}};

\draw[line width=0.5pt] (6,1.2)--(3,1.2);
\draw[line width=0.5pt,dotted] (3,-1.2)--(3,1.2);
\draw[line width=0.5pt] (3,-1.2)--(0,-1.2);
\fill (3.2,0) node[below] {\footnotesize{$p_{0}$}};

\draw[fill=white] (3,1.2) circle (2pt);
\draw[fill=white] (3,-1.2) circle (2pt);

\draw[>={Stealth[length=8pt,width=4pt]},->>,line width=0.7pt] (0,0)--(2.9,0);

\draw[>={Stealth[length=8pt,width=4pt]},->>,line width=0.7pt] (6,0)--(3.1,0);

\draw[xshift=7.5cm,line width=0.5pt] (0,0)--(6,0);
\draw[xshift=7.5cm,fill=black] (0,0) circle (2pt);
\draw[xshift=7.5cm,fill=black] (6,0) circle (2pt);
\fill[xshift=7.5cm] (0,0) node[left] {\footnotesize{$0$}};
\fill[xshift=7.5cm] (6,0) node[below] {\footnotesize{$1$}};
\fill[xshift=7.5cm] (6,0) node[right] {\footnotesize{$p$}};

\draw[xshift=7.5cm,line width=0.5pt,dotted] (3,-1.2)--(0,-1.2)--(0,1.2);
\fill[xshift=7.5cm] (0,1.2) node[left] {\footnotesize{$\beta(p)$}};

\draw[xshift=7.5cm,line width=0.5pt] (0,1.2)--(3,1.2);
\draw[xshift=7.5cm,line width=0.5pt,dotted] (3,-1.2)--(3,1.2);
\draw[xshift=7.5cm,line width=0.5pt] (3,-1.2)--(6,-1.2);

\draw[xshift=7.5cm,fill=white] (3,1.2) circle (2pt);
\draw[xshift=7.5cm,fill=white] (3,-1.2) circle (2pt);

\draw[xshift=7.5cm,>={Stealth[length=8pt,width=4pt]},->>,line width=0.7pt] (3,0.1)--(0,0.1);

\draw[xshift=7.5cm,fill=white] (3,0.1) circle (2pt);
\fill[xshift=7.5cm] (3.2,-0.1) node[below] {\footnotesize{$p_{0}$}};

\draw[xshift=7.5cm,>={Stealth[length=8pt,width=4pt]},->>,line width=0.7pt] (3,-0.1)--(6,-0.1);

\draw[xshift=7.5cm,fill=white] (3,-0.1) circle (2pt);

\end{tikzpicture}
\endpgfgraphicnamed\caption{\label{fig:admissible}The left panel depicts the case where the integral equation \eqref{eq:pt_integral} does not have a solution, while the right panel shows the case where \eqref{eq:pt_integral} has multiple solutions.}
\end{figure}

We next consider a condition that is not necessary for  \eqref{eq:pt_integral} to be well defined, but is sensible as a selection rule when \eqref{eq:pt_integral} admits  multiple solutions. Consider, for example, a strategy such that the sender plays the $L$-drifting experiment targeting $1$ (so $\beta(p)=1$) whenever $p\leq p_{0}$ and the $R$-drifting experiment targeting $0$ (so $\beta(p)=-1$) whenever $p>p_{0}$ (see the right panel of Figure \ref{fig:admissible}). Since the former case includes $p_{0}$, it is natural that starting from $p_{0}$, the belief moves leftward according to
\begin{equation*}
    p_{t}=\frac{pe^{-\lambda t}}{pe^{-\lambda t}+(1-p)}.
\end{equation*}
However, since $p_{t}=p_{0}$ only when $t=0$, the following is also a solution to \eqref{eq:pt_integral}:
\begin{equation*}
    p_{t}=\frac{pe^{\lambda t}}{pe^{\lambda t}+(1-p)}.
\end{equation*}
Whenever this multiplicity arises, we select the most natural one which would be obtained from the discrete-time approximation. This selection, however, is inconsequential for our equilibrium characterization, because at a point where this selection issue arises (such as $\underline{\pi}_{LR}$ or $\overline{\pi}_{LR}$ in Propositions \ref{prop:C2_hold} and \ref{prop:C2_fail}), we can arbitrarily specify the sender's strategy; the selection forces us to adopt a particular belief path, but does not restrict the sender's strategy in any way.

\section{Proofs of Propositions \ref{prop:C2_hold} and \ref{prop:C2_fail}\label{sec:Proofs-of-Propositions}}

The proofs are presented in several sections. Throughout, we take
$p^{\ast}\in(\hat{p},1)$ as given and construct the corresponding
equilibria. Section \ref{sec:Eq-payoffs} constructs the value functions
that correspond to the equilibrium strategies in Propositions \ref{prop:C2_hold}
and \ref{prop:C2_fail}. Sections \ref{subsec:sender_verify} and
\ref{sec:Verifying-the-Receiver's} respectively verify the sender's
and the receiver's incentives. Uniqueness of SMPE is proven in Online
Appendix \ref{sec:Uniqueness-of-SMPE}. A brief sketch is provided
in Section \ref{subsec:uniqueness}.

\subsection{Constructing Equilibrium Value Functions\label{sec:Eq-payoffs}}

We first compute the players' value functions under alternative persuasion
strategies; they will be used to compute the players' equilibrium
payoffs. In what follows, we take it for granted that the receiver
takes an action immediately if the belief reaches either $0$ or $p^{\ast}$.
We also assume that the receiver waits while the sender plays each
persuasion strategy in this subsection.

\paragraph*{ODEs for $R$-drifting and $L$-drifting.}

For any $p\in(0,p^{*})$, let $N_{\varepsilon}(p)$ denote a small
open neighborhood of $p$. Suppose that for any belief $p$ in $N_{\varepsilon}(p)$,
the sender plays the $R$-drifting experiment with jump target $0$.
Then, the sender's value function $V_{+}(p)$ and the receiver's value
function $U_{+}(p)$ satisfy the following ODEs:\footnote{The ODEs can be obtained heuristically in the same way as the Hamilton-Jacobi-Bellman
equation. The subscripts, ``$+$'' and ``$-$'', represent the
direction of belief drifting in the absence of Poisson jumps.}
\begin{equation}
c=\lambda p(1-p)\left(\frac{-V_{+}(p)}{p}+V_{+}^{\prime}(p)\right)\quad\text{ and }\quad c=\lambda p(1-p)\left(\frac{u_{\ell}^{L}-U_{+}(p)}{p}+U_{+}^{\prime}(p)\right).\label{eq:ODEs_right_drifting}
\end{equation}
Similarly, suppose for any belief $p$ in $N_{\varepsilon}(p)$ the
sender plays the $L$-drifting experiment with jump target $p^{\ast}$.
Then, the players' value functions, $V_{-}(p)$ and $U_{-}(p)$, satisfy
\begin{equation}
c=\lambda p(1-p)\left(\frac{v-V_{-}(p)}{p^{\ast}-p}-V_{-}^{\prime}(p)\right)\text{ and }\;\,c=\lambda p(1-p)\left(\frac{U_{r}(p^{\ast})-U_{-}(p)}{p^{\ast}-p}-U_{-}^{\prime}(p)\right).\label{eq:ODEs_left_drifting}
\end{equation}

\paragraph*{$R$-drifting strategy:}

Suppose the sender plays $R$-drifting experiments until the belief
reaches $p^{\ast}$. In this case, the players' payoffs are obtained
as the solutions to \eqref{eq:ODEs_right_drifting} with boundary
conditions $V_{+}(p^{\ast})=v$ and $U_{+}(p^{\ast})=U_{r}(p^{\ast})$,
respectively. We obtain
\[
V_{R}(p)=\frac{p}{p^{*}}v-C_{+}(p;p^{*})\quad\text{ and }\quad U_{R}(p)=\frac{p^{*}-p}{p^{*}}u_{\ell}^{L}+\frac{p}{p^{*}}U_{r}(p^{*})-C_{+}(p;p^{\ast}),
\]
where $C_{+}(p;q):=\left(p\log\left(\frac{q}{1-q}\frac{1-p}{p}\right)+1-\frac{p}{q}\right)\frac{c}{\lambda}$
represents the expected cost of using $R$-drifting experiments until
the belief moves from $p$ to either $0$ or $q$.

\paragraph*{Stationary strategy:}

Suppose the sender uses the stationary experiment with jump targets
$0$ and $p^{*}$ at $p$. Then, the players' value functions, $V_{S}(p)$
and $U_{S}(p)$, are respectively given by
\begin{equation}
V_{S}(p)=\frac{p}{p^{\ast}}v-C_{S}(p)\quad\text{ and }\quad U_{S}(p)=\frac{p^{*}-p}{p^{*}}u_{\ell}^{L}+\frac{p}{p^{*}}U_{r}(p^{*})-C_{S}(p),\label{eq:stationary_sol}
\end{equation}
where $C_{S}(p):=\frac{2c(p^{\ast}-p)}{\lambda p^{\ast}(1-p)}$ represents
the expected cost of playing the stationary strategy.\footnote{Under the stationary strategy, the total arrival rate of Poisson jumps
is equal to $\lambda_{S}(p)=\frac{\lambda}{2}(1-p)+\frac{\lambda}{2}\frac{p(1-p)}{p^{*}-p}=\frac{\lambda}{2}\frac{p^{\ast}(1-p)}{p^{\ast}-p}$.
$C_{S}(p)$ is equal to $c$ times the expected arrival time $1/\lambda_{S}(p)$.}

\paragraph*{$RS$ strategy ($R$-drifting followed by stationary):}

Suppose the sender plays the $R$-drifting strategy until $q(>p)$
and then switches to the stationary strategy. Then, the players' value
functions solve \eqref{eq:ODEs_right_drifting} with boundary conditions
$V_{+}(q)=V_{S}(q)$ and $U_{+}(q)=U_{S}(q)$, yielding
\[
V_{RS}(p;q)=\frac{p}{p^{*}}v-C_{+}(p;q)-\frac{p}{q}C_{S}(q)\text{ and }U_{RS}(p;q)=\frac{p^{*}-p}{p^{*}}u_{\ell}^{L}+\frac{p}{p^{*}}U_{r}(p^{*})-C_{+}(p;q)-\frac{p}{q}C_{S}(q).
\]
Note that $p/q$ is the probability that the belief moves from $p$
to $q$ (whereupon the sender switches to the stationary strategy).

\paragraph*{$LS$ strategy ($L$-drifting followed by stationary):}

Suppose the sender plays the $L$-drifting strategy until $q(<p)$
and then switches to the stationary strategy. Then, the players' value
functions solve \eqref{eq:ODEs_left_drifting} with boundary conditions
$V_{-}(q)=V_{S}(q)$ and $U_{-}(q)=U_{S}(q)$, resulting in
\begin{align*}
V_{LS}(p;q) & =\frac{p}{p^{\ast}}v-C_{-}(p;q)-\frac{p^{*}-p}{p^{*}-q}C_{S}(q)\text{ and}\\
U_{LS}(p;q) & =\frac{p^{*}-p}{p^{*}}u_{\ell}^{L}+\frac{p}{p^{*}}U_{r}(p^{*})-C_{-}(p;q)-\frac{p^{*}-p}{p^{*}-q}C_{S}(q),
\end{align*}
where $C_{-}(p;q):=-\frac{p^{\ast}-p}{p^{\ast}(1-p^{\ast})}\left(p^{\ast}\log\frac{1-q}{1-p}+(1-p^{\ast})\log\frac{q}{p}-\log\frac{p^{\ast}-q}{p^{\ast}-p}\right)\frac{c}{\lambda}$
denotes the expected cost of playing $L$-drifting experiments until
the belief drifts down from $p$ to $q(<p)$.

\paragraph*{Crossing lemma.}

The following lemma provides potential crossing patterns among the
value functions and plays a crucial role in the subsequent analysis.
\begin{lem}[Crossing Lemma]
\label{lem:Crossing} Let $V_{+}(p)$ and $V_{-}(p)$ be solutions
to \eqref{eq:ODEs_right_drifting} and \eqref{eq:ODEs_left_drifting},
respectively.
\begin{enumerate}
\item Let $p^{*}<8/9$. For all $p<p^{*}$, if $V_{+}(p)=V_{S}(p)$, then
$V_{+}^{\prime}(p)<V_{S}^{\prime}(p)$. Similarly, if $V_{-}(p)=V_{S}(p)$
then $V_{-}^{\prime}(p)<V_{S}^{\prime}(p)$.
\item Let $p^{*}\ge8/9$, and define $\xi_{1}:=\frac{3p^{\ast}}{4}-\sqrt{\left(\frac{3p^{\ast}}{4}\right)^{2}-\frac{p^{\ast}}{2}}$,
and $\xi_{2}:=\frac{3p^{\ast}}{4}+\sqrt{\left(\frac{3p^{\ast}}{4}\right)^{2}-\frac{p^{\ast}}{2}}$.
\begin{enumerate}
\item For all $p<p^{*}$, if $V_{+}(p)=V_{S}(p)$, then $V_{+}^{\prime}(p)=V_{S}^{\prime}(p)$
if and only if $p\in\{\xi_{1},\xi_{2}\}$, and $V_{+}^{\prime}(p)>V_{S}^{\prime}(p)$
if and only if $p\in(\xi_{1},\xi_{2})$;
\item For all $p<p^{*}$, if $V_{-}(p)=V_{S}(p)$, then $V_{-}^{\prime}(p)=V_{S}^{\prime}(p)$
if and only if $p\in\{\xi_{1},\xi_{2}\}$, and $V_{-}^{\prime}(p)>V_{S}^{\prime}(p)$
if and only if $p\in(\xi_{1},\xi_{2})$.
\end{enumerate}
\item For all $p<p^{*}$, if $V_{+}(p)=V_{-}(p)$, then $\text{sign}\left(V_{+}^{\prime}(p)-V_{-}^{\prime}(p)\right)=\text{sign}\left(V_{-}(p)-V_{S}(p)\right)$.
\end{enumerate}
All parts also hold for the receiver's value functions $U_{+}(\cdot)$,
$U_{-}(\cdot)$, and $U_{S}(\cdot)$.
\end{lem}
\begin{proof}
We focus on the sender's value functions, as the same proofs apply
to the receiver. From \eqref{eq:ODEs_right_drifting}, \eqref{eq:ODEs_left_drifting},
and \eqref{eq:stationary_sol}, we can obtain expressions for $V_{+}^{\prime}(p)$,
$V_{-}^{\prime}(p)$, and $V_{S}^{\prime}(p)$. Combining these with
$V_{+}(p)=V_{S}(p)$ and $V_{-}(p)=V_{S}(p)$, we obtain
\[
V_{+}^{\prime}(p)-V_{S}^{\prime}(p)=V_{-}^{\prime}(p)-V_{S}^{\prime}(p)=-\frac{c(2p^{2}-3p^{\ast}p+p^{\ast})}{\lambda p^{\ast}p(1-p)^{2}}\gtreqqless0\Leftrightarrow-2p^{2}+3p^{\ast}p-p^{\ast}\gtreqqless0.
\]
For $p^{*}<8/9,$ the quadratic expression in the last inequality
is always negative, which proves part $(a)$. For $p^{*}\ge8/9$,
the quadratic expression has two real roots, $\xi_{1}$ and $\xi_{2}$,
and is positive if and only if $p\in(\xi_{1},\xi_{2})$. This proves
(b).

Similarly, using $V_{+}(p)=V_{-}(p)$, we have
\[
V_{+}^{\prime}(p)-V_{-}^{\prime}(p)=\frac{p^{\ast}}{p(p^{\ast}-p)}\left(V_{-}(p)-V_{S}(p)\right),
\]
which leads to (c).
\end{proof}

\paragraph*{Construction of $\xi$.}

While $\xi$ is part of the equilibrium only for $p^{\ast}>\eta$,
we define it generally. For $p^{*}\ge8/9$ we set $\xi:=\xi_{1}$
and for $p^{*}<8/9$ we set $\xi:=p^{*}$. We define it in this way
to ensure that $V_{RS}(p;\xi)$ meets $V_{S}(p)$ from above at $p=\xi$
(as $p$ rises toward $\xi$). In particular, together with the Crossing
Lemma \ref{lem:Crossing}.(b), this means that for any $p<\xi$, $V_{RS}(p;\xi)$
is above $V_{S}(p)$, and for $p^{*}\ge8/9$, these two functions
have the same slope at $p=\xi$. This will play a crucial role later.

\paragraph*{Construction of $\eta$.\label{par:Construction-of-eta}}

The parameter $\eta$ is the value of $p^{\ast}\ge8/9$ such that
$V_{R}(\xi(p^{*}))=V_{S}(\xi(p^{*}))$.\footnote{\label{fn:eta_well-defined} To show that $\eta$ is well-defined,
we can define a function $g:(8/9,1)\to\mathbb{R}$ by $g(p^{*}):=V_{R}(\xi(p^{*}))-V_{S}(\xi(p^{*}))$
so that $g(\eta)=0$. It can be verified that $g^{\prime}(p^{\ast})>0$
for all $p^{\ast}\in(8/9,1)$.} (We make the dependence of $\xi$ on $p^{*}$ explicit here, and
also note that the functions $V_{R}(\cdot)$ and $V_{S}(\cdot)$ depend
on $p^{*}$ directly.) Solving this equation yields $p^{\ast}=\eta\approx0.943$.
We make the following observations for a later purpose:
\begin{lem}
\label{lem:significance_eta}~
\begin{enumerate}
\item If $p^{\ast}<\eta$ then $V_{R}(p)>V_{S}(p)$ for all $p\in(0,p^{\ast})$.
\item If $p^{\ast}=\eta$ then $V_{R}(p)\geq V_{S}(p)$ for all $p\in(0,p^{\ast})$,
with equality only when $p=\xi$.
\item If $p^{\ast}>\eta$ then $V_{R}(\xi)<V_{S}(\xi)$.
\end{enumerate}
The same results hold for $U_{R}(\cdot)$ and $U_{S}(\cdot)$.
\end{lem}
\begin{proof}
We focus on the sender's value functions, as the same proofs apply
to the receiver. Using the explicit solutions of $V_{R}(p)$ and $V_{S}(p)$,
we can see that $V_{S}(0)<V_{R}(0)$, $V_{S}(p^{\ast})=V_{R}(p^{\ast})$,
and $V_{S}^{\prime}(p^{\ast})>V_{R}^{\prime}(p^{\ast})$. Therefore,
either $V_{S}(p)$ stays weakly below $V_{R}(p)$ for all $p<p^{\ast}$,
or $V_{S}(p)$ crosses $V_{R}(p)$ at least twice (from below and
then from above). By Lemma \ref{lem:Crossing}.(b), the latter occurs
only if $V_{S}(p)$ crosses $V_{R}(p)$ from below at some $p<\xi$,
and then second time from above at some $p^{\prime}\in(\xi,\xi_{2})$,
which is equivalent to $V_{R}(\xi)<V_{S}(\xi)$. The desired result follows since $V_{R}(\xi(p^{\ast}))-V_{S}(\xi(p^{\ast}))$
changes the sign only once at $p^{*}=\eta$ (see Footnote \ref{fn:eta_well-defined}).
\end{proof}

\paragraph*{Pasted strategies.}

Given $\xi$, we combine alternative strategies as follows. For any
$p\leq p^{\ast}$, we define
\[
\widehat{V}(p):=\left\{ \begin{array}{ll}
V_{RS}(p;\xi) & \mbox{ if }p<\xi\\
V_{S}(\xi) & \mbox{ if }p=\xi\\
V_{LS}(p;\xi) & \mbox{ if }p\in[\xi,p^{*}],
\end{array}\right.\text{ and }\widehat{U}(p):=\left\{ \begin{array}{ll}
U_{RS}(p;\xi) & \mbox{ if }p<\xi\\
U_{S}(\xi) & \mbox{ if }p=\xi\\
U_{LS}(p;\xi) & \mbox{ if }p\in[\xi,p^{*}].
\end{array}\right.
\]
We next define $\widetilde{V}(p):=\max\{V_{R}(p),\widehat{V}(p)\}\text{ and }\widetilde{U}(p):=\max\{U_{R}(p),\widehat{U}(p)\}$.
We make several useful observations in the following lemma.

\begin{lem}
\label{lem:pasted_value_properties}~
\begin{enumerate}
\item Both $\widetilde{V}(p)$ and $\widetilde{U}(p)$ are strictly convex
in $p$ over $[0,p^{\ast}]$.
\item If $p^{\ast}\leq\eta$ then $\widetilde{V}(p)=V_{R}(p)$ and $\widetilde{U}(p)=U_{R}(p)$
for all $p\in[0,p^{\ast}]$.
\item If $p^{\ast}>\eta$ then there exists $\overline{\pi}_{LR}\in(\xi,p^{\ast})$
such that $\widetilde{V}(p)=\widehat{V}(p)$ and $\widetilde{U}(p)=\widehat{U}(p)$
for $p\leq\overline{\pi}_{LR}$ and $\widetilde{V}(p)=V_{R}(p)$ and
$\widetilde{U}(p)=U_{R}(p)$ for $p\in[\overline{\pi}_{LR},p^{\ast}]$.
\item $\widetilde{V}(p)\geq V_{S}(p)$ for all $p<p^{\ast}$, and the inequality
is strict for $p\neq\xi$.
\end{enumerate}
\end{lem}
\begin{proof}
The same proof applies to both players, so we focus on the sender's
value functions. Recall that for $p^{*}<8/9$ we have $\xi=p^{*}$
so that $\widehat{V}(p)=V_{RS}(p;p^{*})=V_{R}(p)$, which implies
(b). Since $V_{R}(p)$ is strictly convex, (a) holds as well. In what
follows, we consider $p^{*}\ge8/9$, in which case $\widehat{V}(p)\neq V_{R}(p)$.

(a) Since $\widetilde{V}(p)$ is the upper envelope of two functions
and $V_{R}(p)$ is strictly convex over $[0,p^{\ast}]$, it suffices
to prove that $\widehat{V}(p)$ is also strictly convex over $[0,p^{\ast}]$.
Both $V_{RS}(p;\xi)$ and $V_{LS}(p;\xi)$ are strictly convex over
their respective supports, and $\widehat{V}(p)$ is continuously differentiable
at the pasting point $\xi$. The latter holds because $V_{RS}(\xi;\xi)=V_{LS}(\xi;\xi)=V_{S}(\xi)$
implies $V'_{RS}(\xi;\xi)=V'_{LS}(\xi;\xi)$ by Lemma \ref{lem:Crossing}.(c).

(b) If $p^{*}<\eta$, $V_{RS}(\xi;\xi)=V_{S}(\xi)<V_{R}(\xi)$ by
Lemma \ref{lem:significance_eta}.(a). Together with the fact that
both $V_{RS}(p;\xi)$ and $V_{R}(p)$ satsify the ODE \eqref{eq:ODEs_right_drifting},
this implies that $\widehat{V}(p)=V_{RS}(p;\xi)<V_{R}(p)$ for all
$p\leq\xi$.\footnote{It is easy to see that \eqref{eq:ODEs_right_drifting} satisfies the
Lipschitz condition for uniqueness on $(0,p^{*})$.} For $p\in(\xi,p^{\ast}]$, observe that $V_{LS}(\xi;\xi)=V_{S}(\xi)<V_{R}(\xi)$
(Lemma \ref{lem:significance_eta}.(a)), $V_{LS}(p^{\ast};\xi)=V_{R}(p^{\ast})$,
and $V_{LS}^{\prime}(p^{\ast};\xi)>V_{R}^{\prime}(p^{\ast})$. Therefore,
either $\widehat{V}(p)=V_{LS}(p;\xi)<V_{R}(p)$ for all $p\in(\xi,p^{\ast})$,
or $V_{LS}(\cdot;\xi)$ crosses $V_{R}(\cdot)$ from below at least
once at some $p\in(\xi,p^{\ast})$. In the latter case, we must have
$V_{R}^{\prime}(p)=V_{+}^{\prime}(p)<V_{-}^{\prime}(p)=V_{LS}^{\prime}(p;\xi)$.
Then, by Lemma \ref{lem:Crossing}.(c), $V_{R}(p)=V_{LS}(p;\xi)=V_{-}(p)<V_{S}(p)$,
contradicting Lemma \ref{lem:significance_eta}.(a).

The result for $p^{\ast}=\eta$ follows from a continuity argument:
both $\widehat{V}(p)$ and $V_{R}(\cdot)$ change continuously in
$p^{\ast}$. Since $\widehat{V}(p)<V_{R}(p)$ for all $p<p^{\ast}$
whenever $p^{\ast}<\eta$, it must be that $\widehat{V}(p)\leq V_{R}(p)$
for all $p<p^{\ast}$ when $p^{\ast}=\eta$. This concludes the proof
for part (b).

For parts (c) and (d), the following claim is useful:
\begin{claim}
\label{claim:Vhat}Suppose $p^{*}\ge8/9$.
\begin{enumerate}
\item[(i)] $\widehat{V}(p)\ge V_{S}(p)$ for all $p\in(0,\xi_{2}]$, with strict
inequality for $p\neq\xi$.
\item[(ii)] $V_{R}(p)>V_{S}(p)$ for all $p\in[\xi_{2},p^{*})$.
\end{enumerate}
\end{claim}
\begin{proof}
(i) Consider first $p<\xi(=\xi_{1})$. We have to show that $\widehat{V}(p)=V_{RS}(p;\xi)>V_{S}(p)$.
To see this, pick $q<\xi$. Then by Lemma \ref{lem:Crossing}.(b).(i),
$V_{RS}(p;q)$ stays above $V_{S}(p)$ for $p<q$ and $V_{RS}(p;\xi)>V_{RS}(p;q)$
for all $q<\xi$. The same logic applies to $V_{LS}(p;\xi)$ for $p\in(\xi,\xi_{2}]$.
For part (ii), we check that $V_{R}(p^{*})=V_{S}(p^{*})$ and $V'_{R}(p^{*})<V'_{S}(p^{*})$.
Lemma \ref{lem:Crossing}.(b).(i) then implies that $V_{R}(p)$ and
$V_{S}(p)$ cannot intersect at $p\ge\xi_{2}$.
\end{proof}
For part (c), we first show that $V_{RS}(p;\xi)>V_{R}(p)$ for $p\le\xi$.
If $p^{\ast}>\eta$ then $V_{RS}(\xi;\xi)=V_{S}(\xi)>V_{R}(\xi)$
 (Lemma \ref{lem:significance_eta}.(c)), which immediately implies
that $\widehat{V}(p)=V_{RS}(p;\xi)>V_{R}(p)$ for all $p\leq\xi$.
Next, for $p\in(\xi,p^{\ast}]$, observe that $V_{LS}(\xi;\xi)=V_{S}(\xi)>V_{R}(\xi)$;
and $V_{LS}(p;\xi)<V_{R}(p)$ for $p=p^{*}-\varepsilon$, since $V_{LS}(p^{\ast};\xi)=V_{R}(p^{\ast})$,
and $V_{LS}^{\prime}(p^{\ast};\xi)>V_{R}^{\prime}(p^{\ast})$. This
means that $V_{LS}(\cdot;\xi)$ crosses $V_{R}(\cdot)$ at least once
in $(\xi,p^{\ast})$. To show that there is a unique crossing point
$\overline{\pi}_{LR}$, note that Claim \ref{claim:Vhat} implies
that at any crossing point $p\in(\xi,p^{*})$, $V_{LS}(p;\xi)=\widehat{V}(p)=V_{R}(p)>V_{S}(p)$,
and hence by Lemma \ref{lem:Crossing}.(c), $V_{LS}(p;\xi)$ can cross
$V_{R}(p)$ only from above. Therefore, there is a unique crossing
point.

(d) If $p^{\ast}\leq\eta$, then the result is immediate from Lemmas
\ref{lem:significance_eta}.(a) and \ref{lem:pasted_value_properties}.(b).
If $p^{\ast}>\eta$ the result is immediate from Lemma \ref{lem:pasted_value_properties}.(c)
and Claim \ref{claim:Vhat}.
\end{proof}

\subsubsection{Equilibrium payoffs and construction of $p_{\ast}$ in Proposition
\ref{prop:C2_hold}}

\label{subsubsec:P2_equil}

When (\ref{cond:low_receiver_rent}) holds, we define $p_{\ast}$
as the belief $\phi_{\ell R}$ at which the receiver is indifferent
between waiting and stopping with action $\ell$; that is, we set
$p_{*}:=\phi_{\ell R}$ where $\phi_{\ell R}$ is defined by\footnote{\label{fn:p_ast_well_def_C2hold}To see that $\phi_{\ell R}$ is well
defined, observe that, whether $p^{\ast}\leq\eta$ or $p^{\ast}>\eta$,
$\lim_{p\to0}\tilde{U}(p)=u_{\ell}^{L}-\frac{c}{\lambda}<u_{\ell}^{L}=U_{\ell}(0)$,
while $\widetilde{U}(p^{\ast})=U_{r}(p^{*})>U_{\ell}(p^{*})$ (because
$p^{\ast}>\hat{p}$). In addition, $\tilde{U}(p)$ is strictly convex
over $[0,p^{\ast}]$ (Lemma \ref{lem:pasted_value_properties}.(a)),
while $U_{\ell}(p)$ is linear. Therefore, $\widetilde{U}(p)$ crosses
$U_{\ell}(p)$ from below only once.}
\begin{equation}
U_{\ell}(\phi_{\ell R})=\widetilde{U}(\phi_{\ell R}).\label{eq:receiver_indifference}
\end{equation}
We focus on the case in which $c$ is sufficiently small. In the limit
as $c\rightarrow0$, $\widetilde{U}(p)=\frac{p^{*}-p}{p^{*}}u_{\ell}^{L}+\frac{p}{p^{*}}U_{r}(p^{*})>U_{r}(p)$
for all $p$. Therefore, there exists $c_{1}>0$ such that $p_{*}=\phi_{\ell R}<\hat{p}$
for all $c\le c_{1}$. We assume that $c\le c_{1}$ in the sequel.
The following Lemma shows that the sender's payoff is positive at
$p_{*}$ if Condition \eqref{cond:low_receiver_rent} holds.
\begin{lem}
\label{lem:Vtilde(phi_ellL)} $\widetilde{V}(\phi_{\ell R})>0$ if
and only if Condition \eqref{cond:low_receiver_rent} holds.
\end{lem}
\begin{proof}
By \eqref{eq:receiver_indifference}, we have
\[
\widetilde{V}(\phi_{\ell R})=\frac{\phi_{\ell R}}{p^{\ast}}v+\widetilde{U}(\phi_{\ell R})-\left(\frac{p^{\ast}-\phi_{\ell R}}{p^{\ast}}u_{\ell}^{L}+\frac{\phi_{\ell R}}{p^{\ast}}U_{r}(p^{\ast})\right)\overset{\eqref{eq:receiver_indifference}}{=}\frac{\phi_{\ell R}}{p^{\ast}}\left(v-(U_{r}(p^{\ast})-U_{\ell}(p^{\ast})\right),
\]
where the first equality holds because both players incur the same
costs, so that $\widetilde{V}(p)-\frac{p}{p^{\ast}}v=\widetilde{U}(p)-\left(\frac{p^{\ast}-p}{p^{\ast}}u_{\ell}^{L}+\frac{p}{p^{\ast}}U_{r}(p^{\ast})\right)$
whenever $p\in(0,p^{\ast}]$. The last expression is positive if and
only if \eqref{cond:low_receiver_rent} holds.
\end{proof}
We set the players' value functions as follows:
\[
V(p):=\left\{ \begin{array}{ll}
0 & \mbox{ if }p\in[0,p_{*})\\
\widetilde{V}(p) & \mbox{ if }p\in[p_{*},p^{*})\\
v & \mbox{ if }p\ge p^{*},
\end{array}\right.\quad\text{ and }\quad U(p):=\left\{ \begin{array}{ll}
U_{\ell}(p) & \mbox{ if }p\in[0,p_{*})\\
\widetilde{U}(p) & \mbox{ if }p\in[p_{*},p^{*})\\
U_{r}(p) & \mbox{ if }p\ge p^{*}.
\end{array}\right.
\]

\begin{lem}
\label{lem:prop_V_C2_hold}When \eqref{cond:low_receiver_rent} holds,
$V(p)$ is nonnegative and nondecreasing for all $p\in[0,1]$.
\end{lem}
\begin{proof}
Since $\widetilde{V}(\cdot)$ is convex on $[0,p^{*}]$, $\widetilde{V}(0)=-c/\lambda$,
and $\widetilde{V}(p_{*})\ge0$ by Lemma \ref{lem:Vtilde(phi_ellL)},
$\widetilde{V}(\cdot)$ is increasing on $[p_{*},p^{*}]$. Hence $V(\cdot)$
is nondecreasing on $[0,1]$, and nonnegative since $V(0)=0$.
\end{proof}

\subsubsection{Equilibrium payoffs and construction of $p_{\ast}$ in Proposition
\ref{prop:C2_fail}}

\label{subsubsec:P3_equil}

When \eqref{cond:low_receiver_rent} fails, the same construction
as above does not work; for example, $\widetilde{V}(p_{\ast})<0$
by Lemma \ref{lem:Vtilde(phi_ellL)}. The right construction requires
us to consider another $L$-drifting strategy.

\paragraph*{$L0$ strategy ($L$-drifting followed by passing):}

Suppose the sender continues to play the $L$-drifting experiment
until the belief reaches $q(<p)$ and then she stops experimenting
altogether (passes). The resulting value functions are the solutions
to \eqref{eq:ODEs_left_drifting} with boundary conditions $V_{-}(q)=0$
and $U_{-}(q)=U_{\ell}(q)$, which yields
\[
V_{L0}(p;q):=\frac{p-q}{p^{*}-q}v-C_{-}(p;q)\text{ and }U_{L0}(p;q):=\frac{p^{*}-p}{p^{*}-q}U_{\ell}(q)+\frac{p-q}{p^{*}-q}U_{r}(p^{*})-C_{-}(p;q).
\]
Note that this strategy leads to $q$ with probability $\frac{p^{*}-p}{p^{*}-q}$
and $p^{\ast}$ with probability $\frac{p-q}{p^{*}-q}$.

\paragraph*{Construction of $p_{\ast}$.}

Let $\pi_{\ell L}$ denote the lowest value of $q\in(0,\hat{p})$
such that
\begin{equation}
V_{L0}'(q;q)\geq0\Leftrightarrow\frac{\lambda q(1-q)}{p^{\ast}-q}v\geq c\Leftrightarrow q\geq\pi_{\ell L}:=\frac{1}{2}+\frac{c}{2\lambda v}-\sqrt{\left(\frac{1}{2}+\frac{c}{2\lambda v}\right)^{2}-\frac{cp^{\ast}}{\lambda v}}.\label{eq:def_pi_ellL}
\end{equation}
In words, $\pi_{\ell L}$ is the lowest belief at which the sender
is willing to play the $L0$ strategy even for an instance. When \eqref{cond:low_receiver_rent}
fails, we set $p_{\ast}:=\pi_{\ell L}$. Clearly, $\lim_{c\to0}p_{\ast}=0$.
We set $c_{2}>0$ such that $p_{\ast}=\pi_{\ell L}<\hat{p}$ for all
$c\leq c_{2}$ and assume $c\leq c_{2}$ hereafter.
\begin{lem}
\label{lem:p_ast_properties}Suppose \eqref{cond:low_receiver_rent}
fails, and $p_{\ast}=\pi_{\ell L}$. There exists $c_{3}>0$ such
that for all $c\leq c_{3}$:
\begin{enumerate}
\item $\widetilde{V}(p_{\ast})<0$;
\item There exists $\underline{\pi}_{LR}\in(p_{\ast},\min\{\hat{p},\xi\})$
such that $V_{L0}(p;p_{\ast})\geq\widetilde{V}(p)$ if and only if
$p\leq\underline{\pi}_{LR}$.
\end{enumerate}
\end{lem}
\begin{proof}
For each $p^{\ast}$, there exists $c_{3}^{1}>0$ such that $p_{\ast}<\xi$
and $V_{S}(\xi)>0$ for all $c\leq c_{3}^{1}$. In the sequel, we
 assume that $c<c_{3}:=\min\{c_{3}^{1},c_{3}^{2}\}$, where $c_{3}^{2}$
is defined in the proof for (b).

(a) Suppose $p^{\ast}\leq\eta$ so that $\widetilde{V}(p)=V_{R}(p)$
for all $p\leq p^{\ast}$. Since \eqref{eq:def_pi_ellL} holds with
equality at $q=\pi_{\ell L}=p_{*}$, we can substitute $\lambda v/c$
in the explicit solution for $V_{R}(p_{*})$ and get
\[
\widetilde{V}(p_{\ast})=V_{R}(p_{*})<0\Leftrightarrow\log\left(\frac{p^{\ast}}{1-p^{\ast}}\frac{1-p_{\ast}}{p_{\ast}}\right)>\frac{p^{\ast}-p_{\ast}}{p^{\ast}(1-p_{\ast})}.
\]
Define $f_{1}(p):=\log\left(\frac{p^{\ast}}{1-p^{\ast}}\frac{1-p}{p}\right)-\frac{p^{\ast}-p}{p^{\ast}(1-p)}$.
The above inequality holds since $f_{1}(p^{\ast})=0$ and $f_{1}^{\prime}(p)<0$
for all $p<p^{\ast}$. If $p^{\ast}>\eta$, then $\widetilde{V}(p)=V_{RS}(p;\xi)$
for all $p\leq\xi$. In this case,
\[
\widetilde{V}(p_{\ast})<0\Leftrightarrow\frac{2p_{\ast}(p^{\ast}-\xi)}{p^{\ast}\xi(1-\xi)}+p_{\ast}\log\left(\frac{\xi}{1-\xi}\frac{1-p_{\ast}}{p_{\ast}}\right)+1-\frac{p_{\ast}}{\xi}>\frac{p^{\ast}-p_{\ast}}{p^{\ast}(1-p_{\ast})}.
\]
Define $f_{2}(p):=\frac{2p(p^{\ast}-\xi)}{p^{\ast}\xi(1-\xi)}+p\log\left(\frac{\xi}{1-\xi}\frac{1-p}{p}\right)+1-\frac{p}{\xi}-\frac{p^{\ast}-p}{p^{\ast}(1-p)}$.
The desired result ($f_{2}(p_{\ast})>0$) holds, because $f_{2}(0)=0$,
$f_{2}(\xi)>0$, and $f_{2}$ is concave over $p\in(0,\xi]$.

(b) We begin by showing that there exists $c_{3}^{2}>0$ such that
for $c<c_{3}^{2}$, $V_{L0}(x;p_{*})<\widetilde{V}(x)$, where $x\in\{\hat{p},\xi\}$.
Since $\widetilde{V}(p)\ge V_{S}(p)$ (Lemma \ref{lem:pasted_value_properties}.(d)),
it suffices to show $V_{L0}(x;p_{*})<V_{S}(x)$. Indeed, we have $V_{L0}(x;p_{*})-V_{S}(x)=\left(\frac{x-p_{*}}{p^{*}-p_{*}}-\frac{x}{p^{*}}\right)v+C_{S}(x)-C_{-}(x;p_{*})<C_{S}(x)-C_{-}(x;p_{*})$,
since $C_{S}(x)/c$ is independent of $c$ and $C_{-}(x;p_{*})/c\rightarrow\infty$
as $c\rightarrow0$.\footnote{This is because $p_{*}\rightarrow0$ as $c\rightarrow0$ so that for
the $L0$ strategy the expected waiting time from any starting point
$x$ becomes infinite if the state is $L$.}

By Lemma \ref{lem:p_ast_properties}.(a) we have $V_{L0}(p_{*};p_{*})=0>\widetilde{V}(p_{*})$.
Since for $c<c_{3}^{2}$, $V_{L0}(\min\{\hat{p},\xi\};p_{*})<\widetilde{V}(\min\{\hat{p},\xi\})$,
there exists an intersection of $V_{L0}(p;p_{*})$ and $\widetilde{V}(p)$
at some $p\in(p_{*},\min\{\hat{p},\xi\})$. In the remainder of the
proof we show that $V_{L0}(\cdot;p_{*})$ can cross $\widetilde{V}(\cdot)$
only from above, which establishes uniqueness of the intersection
on the whole interval $(p_{*},p^{*})$.

We first consider $p^{*}<\eta$. In this case $\widetilde{V}(p)=V_{R}(p)$
and Lemma \ref{lem:significance_eta} implies that $V_{R}(p)>V_{S}(p)$.
Then, by Lemma \ref{lem:Crossing}.(c), $V_{L0}(p;p_{*})$ can cross
$\widetilde{V}(p)$ only from above.

Second, consider $p^{*}\ge\eta$. Since $\widetilde{V}(p)=V_{LS}(p;\xi)$
for $p\in[\xi,\overline{\pi}_{LR}]$ and both $V_{L0}$ and $V_{LS}$
satisfy \eqref{eq:ODEs_left_drifting}, no intersection can occur
in the interval $[\xi,\overline{\pi}_{LR}]$. Outside this interval
$\widetilde{V}(p)$ satisfies \eqref{eq:ODEs_right_drifting} and
$\widetilde{V}(p)>V_{S}(p)$ by Lemma \ref{lem:pasted_value_properties}.(d).
Therefore, again Lemma \ref{lem:Crossing}.(c) implies that $V_{L0}(p;p_{*})$
can cross $\widetilde{V}(p)$ only from above.
\end{proof}

\paragraph{Equilibrium payoffs.}

The equilibrium value functions are given as follows:
\[
V(p):=\left\{ \begin{array}{ll}
0 & \mbox{ if }p\in[0,p_{*}),\\
V_{L0}(p;p_{\ast}) & \mbox{ if }p\in[p_{*},\underline{\pi}_{LR})\\
\widetilde{V}(p) & \mbox{ if }p\in[\underline{\pi}_{LR},p^{\ast})\\
v & \mbox{ if }p\ge p^{*},
\end{array}\right.\text{ and }U(p):=\left\{ \begin{array}{ll}
U_{\ell}(p) & \mbox{ if }p\in[0,p_{*})\\
U_{L0}(p;p_{*}) & \mbox{ if }p\in[p_{*},\underline{\pi}_{LR})\\
\widetilde{U}(p) & \mbox{ if }p\in[\underline{\pi}_{LR},p^{*})\\
U_{r}(p) & \mbox{ if }p\ge p^{*}.
\end{array}\right.
\]

\begin{lem}
\label{lem:V_positive_convex} When \eqref{cond:low_receiver_rent}
fails, $V(\cdot)$ is nonnegative and nondecreasing on $[0,p^{\ast}]$,
and strictly convex on $[p_{*},p^{*}]$
\end{lem}
\begin{proof}
Lemma \ref{lem:p_ast_properties}.(b) implies that $V(p)=\max\{V_{L0}(p;p_{\ast}),\widetilde{V}(p)\}$
over $[p_{\ast},p^{\ast}]$. This is strictly convex since it is the
maximum of two strictly convex functions. Strict convexity of $V_{L0}(\cdot)$
on $[p_{\ast},p^{\ast}]$ is routine to verify; we had already shown
convexity of $\widetilde{V}(p)$ in Lemma \ref{lem:pasted_value_properties}.(a).
Finally, by \eqref{eq:def_pi_ellL}, $V(p)$ is continuously differentiable
at $p_{*}=\pi_{\ell L}$ and therefore convex on $[0,p^{*}]$. This
also implies that $V(p)$ is nondecreasing.
\end{proof}

\subsection{Verifying the Sender's Incentives\label{subsec:sender_verify}}

We show that for each $p^{\ast}$, the sender's strategy is a best
response if the buyer waits if and only if $p\in W$.\footnote{Recall that $W=[p_{\ast},p^{\ast})$ in Proposition \ref{prop:C2_hold},
and $W=(p_{\ast},p^{\ast})$ in Proposition \ref{prop:C2_fail}.} To this end, we must show that in the waiting region the sender's
equilibrium value function solves the Hamilton-Jacobi-Bellmann (HJB)
equation:\footnote{\label{fn:verification}More formally, since $V(p)$ has kinks, we
show that it is a \emph{viscosity solution} of \eqref{eq:HJB_general}.
Together with $V(p)>0$, this is necessary and sufficient for optimality
of the sender's strategy. For necessity see Theorem 10.8 in \citet{Oksendal2009}.
While we are not aware of a statement of sufficiency that covers precisely
our model, the arguments in \citet{Soner1986} can be easily extended
to show sufficiency.}
\begin{equation}
\tag{HJB}\max_{\alpha(\cdot;p)}\sum_{q\neq p}\alpha(q;p)v(p;q)=c,\label{eq:HJB_general}
\end{equation}
where $v(p;q)$ is as defined in Section \ref{subsec:modes_persuasion}. Outside the
waiting region, the sender's value is independent of her strategy.
Still, our refinement requires that her strategy maximize her instantaneous payoff normalized by $dt$; that is, her choice of experiment should solve
\begin{equation}
\tag{\text{R}ef}\max_{\alpha(\cdot;p)}\sum_{q\neq p}\alpha(q;p)v(p;q)-\boldsymbol{1}_{\left\{\sum\alpha(q;p_{t})>0\right\} }c.\label{eq:refinement}
\end{equation}

Proposition \ref{prop:persuasion_modes}.(b) implies that if $V(p)$
meets certain conditions, then we can restrict attention to Poisson
experiments with jump targets, $0$, $p_{\ast}$, and $p^{\ast}$,
which greatly simplifies both \eqref{eq:HJB_general} and \eqref{eq:refinement}.
Here, we show that our equilibrium value function $V(\cdot)$ satisfies
all properties required by Proposition \ref{prop:persuasion_modes}.(b),
namely that it is nonnegative, increasing, and strictly convex on
$(p_{\ast},p^{\ast}]$, and $V(p_{\ast})/p_{\ast}\leq V^{\prime}(p_{\ast})$.
If \eqref{cond:low_receiver_rent} holds, the first two properties
hold by Lemma \ref{lem:prop_V_C2_hold}, while strict convexity of $V(p)$ follows from Lemma \ref{lem:pasted_value_properties}.(a) and $V(p)=\widetilde{V}(p)$ for $p\in[p_{\ast},p^{\ast}]$. The last property also holds because $\widetilde{V}(p)$ is convex and $\lim_{p\to0}\widetilde{V}(p)=-c/\lambda<0$. If \eqref{cond:low_receiver_rent} fails, the first three properties follow from Lemma \ref{lem:V_positive_convex} and $V(p_{\ast})/p_{\ast}\leq V^{\prime}(p_{\ast})$ also holds, because $p_{\ast}=\pi_{\ell L}>0$ and $V(\pi_{\ell L})=V^{\prime}(\pi_{\ell L})=0$.

\paragraph*{Stopping region.}

We first apply Proposition \ref{prop:persuasion_modes}.(b) to the
stopping region and verify \eqref{eq:refinement}. For $p\geq p^{\ast}$,
the result is immediate from Proposition \ref{prop:persuasion_modes}.(b).(iii).
Now consider $p$ below $p_{\ast}$. Proposition \ref{prop:persuasion_modes}.(b).(ii) implies that the sender has three choices: two $L$-drifting experiments with jump
target $p_{\ast}$ or $p^{\ast}$, and simply passing. This reduces
\eqref{eq:refinement} to

\[
\max_{\alpha_{*},\alpha^{*}\ge0}\lambda p(1-p)\left[\alpha_{*}\frac{V(p_{*})}{p_{*}-p}+\alpha^{*}\frac{v}{p^{*}-p}\right]-c(\alpha_{*}+\alpha^{*})\text{ subject to }\alpha_{*}+\alpha^{*}\le1.
\]

(i) Proposition \ref{prop:C2_fail}: If \eqref{cond:low_receiver_rent}
fails, then $V(p_{\ast})=0$ so that $\alpha_{*}=0$ is optimal. The
coefficient of $\alpha^{*}$ is $\lambda vp(1-p)/(p^{*}-p)-c$. By
\eqref{eq:def_pi_ellL}, this is negative for all $p<p_{\ast}=\pi_{\ell L}$,
so $\alpha^{*}=0$ is optimal. Therefore, for all $p\in[0,p_{\ast}]$,
passing---the sender's strategy as specified in Proposition \ref{prop:C2_fail}---satisfies
\eqref{eq:refinement}.

(ii) Proposition \ref{prop:C2_hold}: If \eqref{cond:low_receiver_rent}
holds, then as discussed in Section \ref{subsec:modes_persuasion}
and depicted in Figure \ref{fig:HJB_simplified} there exists a cutoff
$\pi_{0}<p_{\ast}$ such the coefficient of $\alpha_{*}$ is greater
than the coefficient of $\alpha^{*}$ if and only if $p>\pi_{0}$.\footnote{Specifically $\pi_{0}$ satisfies
\[
\frac{V(p_{\ast})-V(\pi_{0})}{p_{\ast}-\pi_{0}}=\frac{V(p_{\ast})-V(\pi_{0})}{p_{\ast}-\pi_{0}}\Leftrightarrow\frac{V(p^{\ast})}{p^{\ast}-\pi_{0}}=\frac{v}{p^{\ast}-\pi_{0}}\Leftrightarrow\pi_{0}=\frac{p_{\ast}v-p^{\ast}V(p_{\ast})}{v-V(p_{\ast})}.
\]
} The following Lemma shows that $\pi_{\ell L}<\pi_{0}$.
\begin{lem}
\label{lem:pi_ellL_below_pi_0}If \eqref{cond:low_receiver_rent}
holds, then $\pi_{\ell L}<\pi_{0}$.
\end{lem}
\begin{proof}
Let $\pi_{\ell R}$ be the value of $p$ such that $\widetilde{V}(p)=0$.
We show that $\pi_{\ell L}<\pi_{\ell R}<\pi_{0}$. The latter inequality
is immediate from the strict convexity of $\widehat{V}(\cdot)$ on
$[0,p^{\ast}]$ (Lemma \ref{lem:pasted_value_properties}.(a)) and
the definition of $\pi_{0}$. For the former inequality, it suffices
to show that $\widetilde{V}(\pi_{\ell L})<0$, which is shown as in
the proof of Lemma \ref{lem:p_ast_properties}.(a).
\end{proof}
As in the case of Proposition \ref{prop:C2_fail}, passing satisfies
\eqref{eq:refinement} for $p\le\pi_{\ell L}$. Moreover, we have
shown that $\alpha^{*}=1$ satisfies \eqref{eq:refinement} for $p\in(\pi_{\ell L},\pi_{0})$
and $\alpha_{*}=1$ satisfies it for $p\in[\pi_{0},p_{*})$. Therefore,
the sender's strategy in Proposition \ref{prop:C2_hold} satisfied
\eqref{eq:refinement} for all $p<p_{*}$.

\paragraph*{Waiting region.}

Applying Proposition \ref{prop:persuasion_modes}.(b).(i) to $p\in W$,
\eqref{eq:HJB_general} simplifies to
\begin{equation}
\tag{HJB-S}c=\lambda p(1-p)\max_{\alpha\in[0,1]}\left[\alpha\frac{v-V(p)}{p^{*}-p}-(1-\alpha)\frac{V(p)}{p}-\left(2\alpha-1\right)V'(p)\right].\label{eq:HJB_S}
\end{equation}
Our goal is to show that the value function $V(p)$ satisfies this
equation at every $p\in W$. The key argument is the following unimprovability
lemma:
\begin{lem}[Unimprovability]
\label{lem:unimprovability}~
\begin{enumerate}
\item If $V_{+}(p)$ satisfies \eqref{eq:ODEs_right_drifting} and $V_{+}(p)\ge V_{S}(p)$
at $p\in[0,p^{*})$, then $V_{+}(p)$ satisfies \eqref{eq:HJB_S}
at $p$. If $V_{+}(p)>V_{S}(p)$, then $\alpha=0$ is the unique maximizer
in \eqref{eq:HJB_S}.
\item If $V_{-}(p)$ satisfies \eqref{eq:ODEs_left_drifting} and $V_{-}(p)\ge V_{S}(p)$
at $p\in[0,p^{*})$, then $V_{-}(p)$ satisfies \eqref{eq:HJB_S}
at $p$. If $V_{-}(p)>V_{S}(p)$, then $\alpha=1$ is the unique maximizer
in \eqref{eq:HJB_S}.
\end{enumerate}
\end{lem}
\begin{proof}
(a) Substituting $V'(p)=V'_{+}(p)$ from \eqref{eq:ODEs_right_drifting},
\eqref{eq:HJB_S} simplifies to

\[
\max_{\alpha\in[0,1]}\left[-\frac{p^{\ast}}{(p^{\ast}-p)p}(V(p)-V_{S}(p))\right]\alpha=0.
\]
If $V(p)-V_{S}(p)\ge0$, $\alpha=0$ is a maximizer, so the above
condition holds. Further, if $V(p)>V_{S}(p)$, then $\alpha=0$ is
the unique maximizer. The proof for (b) is similar.
\end{proof}
By Lemmas \ref{lem:pasted_value_properties}.(d) and \ref{lem:p_ast_properties}.(b),
$V(p)\geq V_{S}(p)$ holds for all $p\in(p_{\ast},p^{\ast})$. Therefore,
the Unimprovability Lemma \ref{lem:unimprovability} implies that
$V(p)$ satisfies \eqref{eq:HJB_general} for all points where it
is differentiable. At the remaining points $\underline{\pi}_{LR}$,
and $\overline{\pi}_{LR}$, the value function satisfies \eqref{eq:ODEs_right_drifting}
and \eqref{eq:ODEs_left_drifting}, respectively, if we replace $V_{+}'$
by the right derivative and $V_{-}'$ by the left derivative. As in
the proof of the Unimprovability Lemma \ref{lem:unimprovability}
this implies that \eqref{eq:HJB_general} continues to hold if we
insert directional derivatives. Using this observation, together with
the fact $V(p)$ is convex at the points $\underline{\pi}_{LR}$ and
$\overline{\pi}_{LR}$ where it has kinks, it can be shown that $V(p)$
is a viscosity solution of \eqref{eq:HJB_general}, which is sufficient
for optimality of the sender's strategy in the waiting region (see
Footnote \ref{fn:verification} above).

\subsection{Verifying the Receiver's Incentives\label{sec:Verifying-the-Receiver's}}

We now prove the optimality of the receiver's strategy for each belief
$p$, taking as given the sender's strategy. If the sender passes,
which occurs when $p\leq\pi_{\ell L}$ or $p\geq p^{\ast}$, then
the receiver gains nothing from waiting. Since $\pi_{\ell L}\leq p_{\ast}<\hat{p}$
(assuming $c\leq\min\{c_{1},c_{2}\}$) and $p^{\ast}>\hat{p}$, the
receiver chooses $\ell$ if $p\leq\pi_{\ell L}$ and $r$ if $p\ge p^{\ast}$.

Consider next the region $(\pi_{\ell L},p^{\ast})$ on which the sender
does not pass. For this region, we prove that given the sender's strategy,
the receiver's strategy solves her optimal stopping problem in the
dynamic programming sense. By standard verification theorems, it is
sufficient for optimality that the receiver's equilibrium payoff $U(p)$
satisfies the following HJB conditions for all $p$:\footnote{\label{fn:rec_verification}The receiver's value function $U(p)$
is not continuously differentiable at $p_{*}$ (in case \eqref{cond:low_receiver_rent}
holds), $\pil_{LR}$, and $\pih_{LR}$. At these non-smooth points,
we replace $U'(p)$ in \eqref{eq:HJB_2_experiments-rec1} by the right
derivative $U'(p_{+})$, which is the directional derivative \emph{in
the direction of the belief dynamics given by the sender's strategy}.
With this modification, \eqref{eq:HJB_2_experiments-rec1} is well
defined for all $p$.

By standard verification theorems, the conditions \eqref{eq:HJB_2_experiments-rec1}
and \eqref{eq:HJB_2_experiments-rec2} are sufficient for optimality
if $U(p)$ is continuously differentiable. To see that sufficiency
also holds for the receiver's problem, note that we can verify the
receiver's strategy separately for intervals which are closed under
the belief dynamics given by the sender's strategy. For example if
\eqref{cond:low_receiver_rent} holds and $p^{*}\ge\eta$, we can
partition $(\pi_{\ell L},p^{\ast})$ into $P=\{(\pi_{\ell L},p_{*})$,
$[p_{*},\pih_{LR})$, $[\pih_{LR},p^{*})\}$. If the prior belief
is in one of these intervals, the posterior will never leave it unless
a Poisson jump occurs, and the continuation value after a jump can
be taken as fixed. This means that we can verify the optimality of
the receiver's strategy separately for each interval; since $U(p)$
is continuously differentiable on each of the intervals, the standard
verification theorems apply.}
\begin{equation}
\tag{R1}c\ge\lambda p(1-p)\left[\alpha(p)\frac{U(q(p))-U(p)}{q(p)-p}+(1-\alpha(p))\frac{u_{\ell}^{L}-U(p)}{p}-\left(2\alpha(p)-1\right)U'(p)\right],\label{eq:HJB_2_experiments-rec1}
\end{equation}
and
\begin{equation}
\tag{R2}U(p)\ge\max\{U_{\ell}(p),U_{r}(p)\},\label{eq:HJB_2_experiments-rec2}
\end{equation}
and at least one condition holds with equality. Here, $(\alpha(p),q(p))$
represents the sender's strategy as specified in Propositions \ref{prop:C2_hold}
and \ref{prop:C2_fail}, respectively.\footnote{Specifically, $\alpha(p)=0$ if the sender plays the $R$-drifting
experiment; $(\alpha(p),q(p))=(1,q)$ if she plays the $L$-drifting
experiment with jump target $q$; and $(\alpha(p),q(p))=(1/2,p^{\ast})$
if she plays the stationary strategy.}

\paragraph*{Waiting region.}

Suppose $p\in W$. For all points where the receiver's equilibrium
payoff function $U(p)$ is differentiable, by construction, it satisfies
\eqref{eq:HJB_2_experiments-rec1} with equality.\footnote{At kinks, $U(p)$ satisfies \eqref{eq:HJB_2_experiments-rec1} if
$U'(p)$ is replaced by $U'(p_{+})$ (see footnote \ref{fn:rec_verification}).} Hence, it suffices to prove (\ref{eq:HJB_2_experiments-rec2}). We
first show that at $p^{*}$ the slope of $U(p)$ is less than or equal
to the slope of $U_{r}(p)$. To this end, observe
\[
U^{\prime}(p^{\ast})=U_{R}^{\prime}(p^{\ast})=\frac{U_{r}(p^{\ast})-u_{\ell}^{L}}{p^{\ast}}+\frac{c}{\lambda p^{\ast}(1-p^{\ast})}=U_{r}^{\prime}(p^{\ast})-\frac{u_{\ell}^{L}-u_{r}^{L}}{p^{\ast}}+\frac{c}{\lambda p^{\ast}(1-p^{\ast})}.
\]
Since $u_{\ell}^{L}>u_{r}^{L}$, we have $U^{\prime}(p^{\ast})\leq U_{r}^{\prime}(p^{\ast})$
whenever $c\leq c_{4}:=(1-p^{\ast})(u_{\ell}^{L}-u_{r}^{L})$.

(i) Proposition \ref{prop:C2_hold}: When \eqref{cond:low_receiver_rent}
holds, $U(\cdot)$ is convex on $[p_{\ast},p^{\ast}]$ since $U(p)=\widetilde{U}(p)$
for $p\in[p_{\ast},p^{\ast})$ and $\widetilde{U}(\cdot)$ is convex
on $[0,p^{\ast}]$ (Lemma \ref{lem:pasted_value_properties}.(a)).
Together with $U^{\prime}(p^{\ast})\leq U_{r}^{\prime}(p^{\ast})$,
this implies that $U(p)\ge U_{r}(p)$ for all $p\in[p_{\ast},p^{\ast}]$,
provided that $c\leq c_{4}$. We have argued in Footnote \ref{fn:p_ast_well_def_C2hold}
that $U(p)\ge U_{\ell}(p)$ for all $p\in[p_{\ast},p^{\ast}]$. Therefore
\eqref{eq:HJB_2_experiments-rec2} holds for all $p\in[p_{*},p^{*})$.

(ii) Proposition \ref{prop:C2_fail}: We begin by showing that $U_{L0}(p,p_{*})>U_{\ell}(p)$
for all $p\in(p_{*},p^{*})$. Since $U_{L0}(p_{*},p_{*})=U_{\ell}(p_{*})$,
we have
\begin{eqnarray*}
U_{L0}^{\prime}(p_{\ast};p_{\ast}) & = & \frac{U_{r}(p^{\ast})-U_{\ell}(p_{\ast})}{p^{\ast}-p_{\ast}}-\frac{c}{\lambda p_{\ast}(1-p_{\ast})}\geq\frac{U_{\ell}(p^{\ast})-U_{\ell}(p_{\ast})}{p^{\ast}-p_{\ast}}+\frac{v}{p^{\ast}-p_{\ast}}-\frac{c}{\lambda p_{\ast}(1-p_{\ast})}\\
 & = & \frac{U_{\ell}(p^{\ast})-U_{\ell}(p_{\ast})}{p^{\ast}-p_{\ast}}=u_{\ell}^{R}-u_{\ell}^{L}=U_{\ell}^{\prime}(p_{\ast}),
\end{eqnarray*}
where the inequality holds since \eqref{cond:low_receiver_rent} fails,
and the second equality follows from \eqref{eq:def_pi_ellL} and $p_{*}=\pi_{\ell L}$.
Together with the fact that $U_{L0}(\cdot;p_{\ast})$ is convex on
$[p_{\ast},p^{\ast}]$, this implies that $U_{L0}(p,p_{*})>U_{\ell}(p)$
for all $p\in(p_{*},p^{*})$.

For $p\in(p_{*},\pil_{LR})$, $U(p)=U_{L0}(p,p_{*})$. By Lemma \ref{lem:p_ast_properties}.(b),
$\underline{\pi}_{LR}\leq\hat{p}$, provided that $c\leq c_{3}$.
Hence $U_{r}(p)<U_{\ell}(p)$ for $p<\pil_{LR}$, and \eqref{eq:HJB_2_experiments-rec2}
holds since $U(p)=U_{L0}(p,p_{*})>U_{\ell}(p)$ for $p\in(p_{*},\pil_{LR})$.

Next suppose $p\in[\pil_{LR},p^{*})$. Here $U(p)=\widetilde{U}(p)$
and by the same arguments as in (i) we have $\widetilde{U}(p)>U_{r}(p)$.
To show that $\widetilde{U}(p)>U_{\ell}(p)$ it suffices to show that
$\widetilde{U}(p)-U_{L0}(p;p_{*})>0$. Since the sender and the receiver
incur the same cost for each strategy, we can rewrite this difference
as
\begin{align*}
\widetilde{U}(p)-U_{L0}(p;p_{\ast}) & =\widetilde{V}(p)-V_{L0}(p;p_{\ast})+\frac{p_{*}(p^{*}-p)}{p^{*}(p^{*}-p_{*})}\left(U_{r}(p^{*})-U_{\ell}(p^{*})-v\right)>0
\end{align*}
The inequality holds since by Lemma \ref{lem:p_ast_properties}.(b),
$\widetilde{V}(p)-V_{L0}(p;p_{*})\ge0$ for $p\ge\pil_{LR}$; and
$U_{r}(p^{*})-U_{\ell}(p^{*})-v>0$ if \eqref{cond:low_receiver_rent}
is violated.

\paragraph*{The stopping region with $p\in(\pi_{\ell L},p_{\ast})$.}

If \eqref{cond:low_receiver_rent} fails, then $p_{\ast}=\pi_{\ell L}$,
so this case does not arise. The proof of Proposition 3 is thus complete.

Now suppose that \eqref{cond:low_receiver_rent} holds and $p\in(\pi_{\ell L},p_{\ast})$.
In this case, $U(p)$ satisfies \eqref{eq:HJB_2_experiments-rec2}
with equality, so it suffices to show \eqref{eq:HJB_2_experiments-rec1}.
Consider first $p\in[\pi_{0},p_{*})$. For these beliefs, the sender
adopts the $L$-drifting experiment with jump target $p_{\ast}$,
that is, $(\alpha(p),q(p))=(1,p_{\ast})$. Plugging this into \eqref{eq:HJB_2_experiments-rec1}
and using the fact that $U(p)=U_{\ell}(p)$ for all $p\leq p_{\ast}$,
the right-hand side of \eqref{eq:HJB_2_experiments-rec1} is equal
to zero so that \eqref{eq:HJB_2_experiments-rec1} is satisfied.

Finally, consider $p\in[\pi_{\ell L},\pi_{0})$, at which the sender
plays the $L$-drifting experiment with jump target $p^{\ast}$, so
$(\alpha(p),q(p))=(1,p^{\ast})$. Since $U(p)=U_{\ell}(p)$ for all
$p\leq p_{\ast}$, \eqref{eq:HJB_2_experiments-rec1} reduces to
\[
\lambda p(1-p)\left[\frac{U_{r}(p^{\ast})-U_{\ell}(p)}{p^{\ast}-p}-U_{\ell}^{\prime}(p)\right]=\frac{\lambda p(1-p)}{p^{\ast}-p}(U_{r}(p^{\ast})-U_{\ell}(p^{\ast}))\leq c,
\]
which is equivalent to $p\le\phi_{\ell L}$, where $\phi_{\ell L}$
is the unique value of $p$ such that
\[
\frac{\lambda p(1-p)}{p^{\ast}-p}(U_{r}(p^{\ast})-U_{\ell}(p^{\ast}))=c.
\]
The following lemma shows that that $\phi_{\ell L}\ge\pi_{0}$ if
$c\le c_{5}$ for some $c_{5}>0$. It then follows that if $c\leq\min\{c_{1},\ldots,c_{5}\}$,
the receiver has no incentive to deviate from his prescribed strategy
in Proposition \ref{prop:C2_hold}, completing the proof.
\begin{lem}
\label{lem:receiver_incentive_below_pi0}Suppose \eqref{cond:low_receiver_rent}
holds. There exists $c_{5}>0$ such that if $c\leq c_{5}$ then $\phi_{\ell L}\ge\pi_{0}$.
\end{lem}
\begin{proof}
Let $\Delta U:=U_{r}(p^{*})-U_{\ell}(p^{*})$. Since $\phi_{\ell L}$
is the lowest $p$ such that $\frac{p(1-p)\lambda\Delta U}{p^{\ast}-p}\geq c$,
\begin{equation}
\pi_{0}\leq\phi_{\ell L}\quad\iff\quad\frac{\pi_{0}(1-\pi_{0})}{p^{\ast}-\pi_{0}}\frac{\lambda}{c}\Delta U<1.\label{eq:phi_below_pi0_inequal}
\end{equation}
It suffices to show that this inequality holds in the limit as $c\to0$.
Recall that
\[
\frac{V(p^{\ast})}{p^{\ast}-\pi_{0}}=\frac{v}{p^{\ast}-\pi_{0}}\quad\iff\quad\pi_{0}=\frac{p^{\ast}}{v-V(p_{\ast})}\left(\frac{p_{\ast}}{p^{\ast}}v-V(p_{\ast})\right)=\frac{p^{\ast}}{v-V(p_{\ast})}\tilde{C}(p_{\ast}),
\]
where $\widetilde{C}(p_{\ast})=\frac{p_{\ast}}{p^{\ast}}v-\widetilde{V}(p_{\ast})$
denotes the total persuasion costs incurred when $p=p_{\ast}=\phi_{\ell R}$
and \eqref{cond:low_receiver_rent} holds. By the definition of $\widetilde{V}(p_{\ast})$,
$\widetilde{C}(p_{\ast})$ can be written as
\[
\widetilde{C}(p_{\ast})=C_{+}(p_{\ast};q_{R})+\frac{p_{\ast}}{q_{R}}C_{S}(q_{R})=\left(p_{\ast}\log\left(\frac{q_{R}}{1-q_{R}}\frac{1-p_{\ast}}{p_{\ast}}\right)+1-\frac{p_{\ast}}{q_{R}}+\frac{p_{\ast}}{q_{R}}\frac{2(p^{\ast}-q_{R})}{p^{\ast}(1-q_{R})}\right)\frac{c}{\lambda},
\]
where $q_{R}:=p^{\ast}$ if $p^{\ast}\leq\eta$ and $q_{R}:=\xi$
if $p^{\ast}>\eta$. Importantly, as $c\to0$, we have $p_{\ast}\to0$,
$\widetilde{C}(p_{\ast})\to0$ and $\widetilde{C}(p_{\ast})\frac{\lambda}{c}\to1$.
It follows that $\pi_{0}\to0$ and $\pi_{0}\frac{\lambda}{c}\to\frac{p^{\ast}}{v}$,
so
\[
\frac{\pi_{0}(1-\pi_{0})}{p^{\ast}-\pi_{0}}\frac{\lambda}{c}\Delta U\to\frac{\Delta U}{v}<1,
\]
where the inequality is due to \eqref{cond:low_receiver_rent}. This
completes the proof.
\end{proof}

\subsection{SMPE Uniqueness given $p^{\ast}$}

\label{subsec:uniqueness}

Fix any $p^{*}$. To show that for $c$ sufficiently small, the strategy
profiles in Propositions \ref{prop:C2_hold} and \ref{prop:C2_fail}
are the unique SMPEs, we prove that any other choice of $p_{\ast}$
than specified in \ref{subsubsec:P2_equil} and \ref{subsubsec:P3_equil}
(i.e., $p_{\ast}\neq\phi_{\ell R}$ if \eqref{cond:low_receiver_rent}
holds and $p_{\ast}\neq\pi_{\ell L}$ if \eqref{cond:low_receiver_rent}
fails) cannot yield an SMPE. This requires a full characterization
of the sender's optimal dynamic strategy given any lower bound $p_{\ast}$
and upper bound $p^{\ast}$, and a thorough examination of the receiver's
incentives in the stopping region as well as in the waiting region.
The former closely follows our construction and analysis of the equilibrium
value functions in \ref{sec:Eq-payoffs} and \ref{subsec:sender_verify},
and the latter follows closely \ref{sec:Verifying-the-Receiver's}.
We relegate the full proof to the Online Appendix \ref{sec:Uniqueness-of-SMPE}.

\begin{singlespacing} \setlength{\bibsep}{1pt}  \bibliographystyle{economet}
\bibliography{bibckm}
 \end{singlespacing}

\newpage{}

\setcounter{page}{1}
\begin{center}
\textbf{\huge{}{}{}Supplemental Material}\textbf{\large{}{}{}}{\large{}{}}\\
 {\large{}{} }\textbf{\large{}{}{}(for online publication)}{\large\par}
\par\end{center}

\begin{center}

\par\end{center}

\smallskip{}

\section{Uniqueness of SMPE for given $p^{*}$ \label{sec:Uniqueness-of-SMPE}}

We prove that for each fixed $p^{\ast}\in(\hat{p},1)$, the equilibrium
in Propositions \ref{prop:C2_hold} and \ref{prop:C2_fail} is the
unique SMPE in each case, provided that $c$ is sufficiently small
(i.e., $c\leq\min\{c_{1},\ldots,c_{5}\}$). We first characterize
the sender's best response given any lower bound $p_{\ast}$ of the
waiting region and then show that it can be part of an equilibrium
if and only if it is as specified in Propositions \ref{prop:C2_hold}
and \ref{prop:C2_fail}, respectively.

In the following we use $p_{*}$ to denote an exogenously given lower
bound of $W$, which may be different from the lower bound $\phi_{\ell R}$
in Proposition \ref{prop:C2_hold} or $\pi_{\ell L}$ in Proposition
\ref{prop:C2_fail}. We will also use $V(p)$ and $U(p)$ to denote
generic value functions for the sender and receiver that are obtained
from the sender's best response to a given waiting region with bounds
$p_{*}$ and $p^{*}$. Hence in the following $V(p)$ and $U(p)$
may be different from the functions defined in Sections \ref{subsubsec:P2_equil}
and \ref{subsubsec:P3_equil}. Throughout we assume that $c\le\min\{c_{1},\ldots,c_{5}\}$.

\paragraph*{A necessary condition.}

One crucial observation is that at $p_{\ast}$, either $V(p_{\ast})=0$
or $U(p_{\ast})=U_{\ell}(p_{\ast})$---that is, at least one player
should not expect a strictly positive net expected payoff from continuing.
Toward a contradiction, suppose that $V(p_{\ast})>0$ and $U(p_{\ast})>U_{\ell}(p_{\ast})$.\footnote{\label{fn:right-continuity}Note that in equilibrium, $V(p)$ must
be right-continuous at $p_{*}$, i.e., $V(p_{*})=V(p_{*+})=\lim_{p\downarrow p_{*}}V(p)$.
If $V(p_{*+})=0$ this is obvious. Next suppose $V(p_{*+})>0$. For
$p\in W$, the value of any strategy of the sender is continuous in
$p$ and hence $V(p)$ must be continuous in the waiting region. Therefore,
a lack of right continuity at $p_{*}$ can only arise if $p_{*}\notin W$,
so that $V(p_{*})=0$. But if $V(p_{*+})>0$ and $V(p_{*})=0$, then
for $p<p_{*}$ close to $p_{*}$, there is no strategy for the sender
that satisfies \eqref{eq:refinement}, hence such a discontinuity
cannot arise in equilibrium. We thus conclude that either $V(p_{*})=V(p_{*+})=0$,
or $V(p_{*})=V(p_{*+})>0$ and in the latter case $p_{*}\in W$. Noting
right-continuity at $p_{*}$ makes the necessary condition much stronger
and is key to the arguments below.} In this case, if $p$ is just below $p_{\ast}$ then, as when $p\in(\pi_{0},\phi_{\ell R})$
in Proposition \ref{prop:C2_hold}, the sender's flow payoff is maximized
by her playing the $L$-drifting experiment with jump target $p_{\ast}$.
But then, since $U(p_{\ast})>U_{\ell}(p_{\ast})$, the sender has
no incentive to stop at $p$, contradicting that $p<p_{*}$ is not
in the waiting region.

Now we proceed by characterizing the sender's value $V(p)$ and the
receiver's value $U(p)$, if the sender plays a best response. In
particular we characterize $V(p_{*})$ and $U(p_{*})$ which will
enable us to use the necessary condition to narrow down possible equilibrium
values of $p_{*}$.

\paragraph*{The sender's best response in the waiting given (any) $p_{\ast}$.}

Let $\overline{V}(\cdot)$ and $\overline{U}(\cdot)$ denote the value
functions given in Section \ref{subsubsec:P3_equil}; that is, $\overline{V}(\cdot)$
and $\overline{U}(\cdot)$ represent the equilibrium value functions
for Proposition \ref{prop:C2_fail}. They play an important role in
the subsequent analysis, because they coincide with the players' payoffs
in a hypothetical situation where given $p^{\ast}$, the sender chooses
both her dynamic strategy and the boundary of the waiting region $p_{*}$, ignoring the receiver's incentives. This implies that in any SMPE with fixed upper bound $p^{*}$, the
sender's payoff can never exceed $\overline{V}(p)$. The following
result is then immediate.
\begin{lem}
\label{lem:p_ast_lower_bound}Fix $p^{*}$. Then in any SMPE $p_{*}\ge\pi_{\ell L}$.
\end{lem}
\begin{proof}
Suppose instead that $p_{*}<\pi_{\ell L}$. Then the sender's value
is $V(p)=\overline{V}(p)=0$ for $p\in(p_{*},\pi_{\ell L})$ which
is achieved by passing. Any other strategy leads to a negative value,
so passing is the unique best response. Given this, the receiver's
best response is to take action $\ell$ for $p\in(p_{*},\pi_{\ell L})$,
contradicting the hypothesis that the infimum of the waiting region
is $p_{*}<\pi_{\ell L}$.
\end{proof}
Next we characterize the sender's best response if $p_{*}>\pi_{\ell L}$.
We begin with cases where it coincides with the strategy prescribed
in Proposition \ref{prop:C2_fail}
\begin{lem}
\label{lem:p_ast_R_region}Fix $p^{*}$ and suppose that $p_{\ast}$
lies in the region where the sender either plays the stationary strategy
($\xi$ for $p^{\ast}>\eta$), or plays the $R$-drifting experiment
in Proposition \ref{prop:C2_fail}. Then, the sender's value over
$W$ of her best response to $p_{*}$ is given by $\Vh(p)$ and the
receiver's value is given by $\Uh(p)$.
\end{lem}
\begin{proof}
If the prior is $p_{0}\in(p_{\ast},p^{\ast}]$ and the sender mimics
her equilibrium strategy from Proposition \ref{prop:C2_fail}, then
the receiver will stop at the same time as if $p_{*}=\pi_{\ell L}$.
Therefore the sender's payoff is equal to $\overline{V}(p_{0})$ which
is an upper bound for her optimal payoff. Hence the strategy remains
a best response. To show that the receiver's value is given by $\Uh(p_{*})$
we must also characterize the sender's best response (and not just
her value). For $p\notin\{\xi,\pil_{LR},\pih_{L}\}$ the Unimprovability
Lemma \ref{lem:unimprovability} implies that the sender has a unique
best response in Proposition \ref{prop:C2_fail}, since by Lemmas
\ref{lem:significance_eta} and \ref{lem:pasted_value_properties}.(c),
$\Vh(p)>V_{S}(p)$ for $p\neq\xi$, hence we get uniqueness also if
$p_{*}\notin\{\xi,\pil_{LR},\pih_{L}\}$. For $p_{*}=\xi$ uniqueness
of the sender's best response follows since choosing $\alpha>1/2$
at $p_{*}$ yields a value of zero for the sender, and $\alpha<1/2$
violates admissibility since the sender uses the $L$-drifting experiment
for $p>\xi$. If $p_{*}\in\{\pil_{LR},\pih_{LR}\}$, non-uniqueness
in Proposition \ref{prop:C2_fail} arises because the sender is indifferent
between the $L0$ and $RS$ strategies (or $LS$ and $R$ at $\pih_{LR}$).
This is no longer the case if $p_{*}\in\{\pil_{LR},\pih_{LR}\}$ since
the $L0$ or $LS$ strategies are no longer feasible, so that uniqueness
obtains. Since the sender has a unique best response given by the
strategy from Proposition \ref{prop:C2_fail}, the receiver's value
from the sender's best response is given by $\Uh(p_{*})$.
\end{proof}
This lemma immediately allows us to apply the necessary condition.
Recall that $\Vh(p)>0$ and $\Uh(p)>\mathcal{U}(p)$ for all $p\in(\pi_{\ell L},p^{*}]$.
Hence if $p_{*}$ is in the region where the sender either plays the
stationary strategy ($\xi$ for $p^{\ast}>\eta$), or plays the $R$-drifting
experiment in Proposition \ref{prop:C2_fail},  $\Vh(p_{*})>0$ and
$\Uh(p_{*})>U_{\ell}(p_{*})$.\footnote{As argued in Footnote \ref{fn:right-continuity}, we must have $p_{*}\in W$
since $\Vh(p_{*+})>0$, and hence $\Vh(p_{*})>0$ and $\Uh(p_{*})>U_{\ell}(p_{*})$.} Therefore $p_{*}$ in this region leads to a violation of the necessary
condition.

Next we consider the case where $p_{\ast}$ lies in the region where
the sender plays an $L$-drifting experiment in Proposition \ref{prop:C2_fail}.
This is the case if $p_{*}\in(\pi_{\ell L},\pil_{LR})$; and when
$p^{*}\ge\eta$ also if $p_{*}\in(\xi,\pih_{LR})$. In this case,
the sender cannot simply replicate her strategy in Proposition \ref{prop:C2_fail}.
For example, suppose $p_{*}\in(\pi_{\ell L},\pil_{LR})$ and $p_{0}\in(p_{*},\pil_{LR})$.
If the sender used the $L0$ strategy, the receiver would stop when
the belief drifts to $p_{*}$ while in Proposition \ref{prop:C2_fail}
he would wait until the belief reaches $\pi_{\ell L}$. Therefore,
the sender's best response may be different from the strategy in Proposition
\ref{prop:C2_fail}.

The following lemma reports a set of observations about the best response
that will allow us to use the necessary condition for an SMPE. To
state this precisely, let $\pi_{\ell R}$ be the unique value such
that $\widetilde{V}(\pi_{\ell R})=0$.\footnote{From our previous results, we know (i) $\pi_{\ell R}\in(\pi_{\ell L},\underline{\pi}_{LR})$
(Lemma \ref{lem:p_ast_properties}.(b)); (ii) if \eqref{cond:low_receiver_rent}
holds, then $\pi_{\ell R}<\phi_{\ell R}$ (Lemma \ref{lem:Vtilde(phi_ellL)})
and $\pi_{\ell R}<\pi_{0}$ (Lemma \ref{lem:pasted_value_properties}.(a));
and (iii) if \eqref{cond:low_receiver_rent} fails, then $\pi_{\ell R}\geq\phi_{\ell R}$
(Lemma \ref{lem:Vtilde(phi_ellL)}).}
\begin{lem}
\label{lem:p_ast_L_region}Let $V(p)$ and $U(p)$ denote the players'
expected payoffs from the sender's best response to a fixed lower
bound $p_{*}$ of the waiting region. Then there exists $c_{6}>0$
such that for all $c\le c_{6}$:
\begin{enumerate}
\item If $p^{\ast}\ge\eta$ and $p_{\ast}\in(\xi,\hat{p})$, then $V(p_{\ast})\geq V_{S}(p_{\ast})>0$
and $U(p_{\ast})\ge U_{S}(p_{\ast})>U_{\ell}(p_{\ast})$.
\item If $p_{\ast}\in(\pi_{\ell L},\pi_{\ell R})$ then $V(p)=V_{L0}(p;p_{\ast})$
and $U(p)=V_{L0}(p;p_{\ast})$ for all $p\in[p_{\ast},\pi_{\ell R})$.
\item If $p_{\ast}\in[\pi_{\ell R},\underline{\pi}_{LR})$ then $V(p_{*})=\widetilde{V}(p_{*})$
and $U(p_{*})=\widetilde{U}(p_{*})$.
\end{enumerate}
\end{lem}
\begin{proof}
If $p_{*}>\xi$, let $c_{6}>0$ be chosen such that for all $c\le c_{6}$,
$V_{S}(p)>0$ for all $p\in[p_{*},p^{*})$, and $U_{S}(p_{\ast})>\max\{U_{\ell}(p_{\ast}),U_{r}(p_{\ast})\}$
.

The sender's best response can be constructed in a similar way as
in Section \ref{sec:Eq-payoffs}, taking the lower bound $p_{\ast}$
as a constraint. If $p_{*}\in(\pi_{\ell L},\xi)$, the $L0$ strategy
now uses $p_{*}\neq\pi_{\ell L}$ as a stopping bound and the value
is given by $V_{L0}(p;p_{*})$ as before. No further modifications
are needed in this case. If $p_{*}>\xi$, we replace the $LS$ strategy
by a modified version which we denote $LS_{*}$. According to this
strategy, the sender uses the $L$-drifting experiment for $p>p_{*}$
and switches to the stationary strategy when the belief drifts to
$p_{*}$. The value of this strategy is given by $V_{LS}(p;p_{*})$.
We also set $\widehat{V}_{*}(p)=V_{LS}(p;p_{*})$ if $p_{*}>\xi$.
With this we can then characterize the value of the sender's best
response as $\widetilde{V}_{*}(p)=\max\{V_{R}(p),\widehat{V}_{*}(p)\}=\max\{V_{R}(p),V_{LS}(p;p_{*})\}$.
Verification of the sender's best response proceeds using similar
steps to those in Section \ref{subsec:sender_verify}.

(a) The value of the sender's best response in this case is given
by $V(p)=\max\{V_{LS}(p;p_{*}),V_{R}(p)\}$. If $V_{R}(p_{*})>V_{S}(p_{*})$
this implies $V(p)=V_{R}(p)$ since $V_{R}(p)>V_{S}(p)$ by Lemma
\ref{lem:Crossing}.(b) and therefore $V_{LS}(p;p_{*})$ cannot cross
$V_{S}(p)$ from below by Lemma \ref{lem:Crossing}.(c). With $V(p)=V_{R}(p)>V_{S}(p)$,
the Unimprovability Lemma \ref{lem:unimprovability} implies that
there is a unique best response, the $R$-drifting strategy, and therefore
$U(p_{*})=U_{R}(p_{*})>U_{S}(p_{\ast})>U_{\ell}(p_{\ast})$

If $V_{R}(p_{*})\le V_{S}(p_{*})$, we must have $p_{*}<\xi_{2}$
since $V_{R}(p)>V_{S}(p)$ for $p\in[\xi_{2},p^{*})$ by Lemma \ref{lem:Crossing}.(b).
Also by Lemma \ref{lem:Crossing}.(b), $V_{LS}(p;p_{*})>V_{S}(p)$
for $p\in(p_{*},\xi_{2}]$. Hence $V(p)>V_{S}(p)$ for all $p\in(p_{*},p^{*}$).
By the Unimprovability Lemma \ref{lem:unimprovability} this implies
that there is a unique best response for all $p\in[p_{*},p^{*}]$
except at the belief where $V_{LS}(p;p_{*})$ and $V_{R}(p)$ intersect
(which could occur at $p_{*}$), and in the latter case $U_{LS}(p;p_{*})=U_{R}(p)$
since the sender and the receiver incur the same cost. Hence $U(p_{*})=U_{LS}(p_{*};p_{*})=U_{S}(p_{*})>U_{\ell}(p_{\ast})$,

(b) In this case, $V(p)=V_{L0}(p;p_{*})>0>\widetilde{V}(p)>V_{S}(p)$
for $p\in(\pi_{\ell L},\pi_{\ell R})$. Therefore, the Unimprovability
Lemma \ref{lem:unimprovability} implies that there is a unique best
response for $p\in(\pi_{\ell L},\pi_{\ell R})$ and the the sender
uses the $L0$ strategy with stopping bound $p_{*}$. Therefore we
have $U(p)=V_{L0}(p;p_{\ast})$ for all $p\in[p_{\ast},\pi_{\ell R})$.

(c) In this case $V(p)=\widetilde{V}(p)>V_{S}(p)$ for $p\in[p_{*},\xi)$,\footnote{To see this note that $\widetilde{V}(p_{*})\ge V_{L0}(p_{*};p_{*})=0$
and $\widetilde{V}(p)>V_{S}(p)$ for $p\in[p_{*},\xi)$. Hence by
the Crossing Lemma \ref{lem:Crossing}.(d), $V_{L0}(p;p_{*})<\widetilde{V}(p)$
for all $p\in(p_{*},\xi)$. } and the Unimprovability Lemma \ref{lem:unimprovability} implies
that there is a unique best response: the sender uses the $RS$-strategy
if $p^{*}>\eta$ and the $R$-drifting strategy if $p^{*}<\eta$.\footnote{In the case $p^{*}=\eta$, both the sender and the receiver are indifferent
between the $RS$-strategy and the $R$-drifting strategy.} Therefore we have $U(p_{*})=\widetilde{U}(p_{*})$
\end{proof}
By Lemmas \ref{lem:p_ast_lower_bound}, \ref{lem:p_ast_R_region},
and \ref{lem:p_ast_L_region}.(a), we must have $\pi_{\ell L}\le p_{*}<\pil_{LR}$
if $c<\min\{c_{1},\ldots,c_{6}\}$. Now we further narrow down possible
equilibrium values of $p_{*}$ and show that in Proposition \ref{prop:C2_hold}
we must have $p_{*}=\phi_{\ell R}$, and in Proposition \ref{prop:C2_fail}
we must have $p_{*}=\pi_{\ell L}$. With that it only remains to show
uniqueness of the sender's equilibrium strategy which follows from
the Unimprovability Lemma.

\paragraph*{Proposition \ref{prop:C2_hold}.}

Proposition \ref{prop:C2_hold} concerns the case where \eqref{cond:low_receiver_rent}
holds. First we rule out $p_{*}\in(\pi_{\ell R},\underline{\pi}_{LR})\setminus\{\phi_{\ell R}\}$.
If $p_{*}\in(\pi_{\ell R},\underline{\pi}_{LR})$, $V(p_{*})=\widetilde{V}(p_{*})$
and $U(p_{*})=\widetilde{U}(p_{*})$ by Lemma \ref{lem:p_ast_L_region}.(c).
Since $p_{*}>\pi_{\ell R}$, $V(p_{*})=\widetilde{V}(p_{*})>0$ and
hence the necessary condition for an SMPE implies that $U(p_{*})=\widetilde{U}(p_{*})=U_{\ell}(p_{\ast})$.
But this condition only holds when $p_{*}=\phi_{\ell R}$, in which
case the equilibrium is as specified in the Proposition \ref{prop:C2_hold}.

Next, we rule out $p_{\ast}\in(\pi_{\ell L},\pi_{\ell R})$. We begin
by showing that $p_{*}<\pi_{\ell R}$ implies $p_{*}<\phi_{\ell L}$
if $c\le c_{5}$. To see this, recall that by Lemma \ref{lem:pi_ellL_below_pi_0},
$\pi_{0}<\phi_{\ell L}$ if $c\le c_{5}.$ Since $\widetilde{V}(p)$
is convex, the construction of $\pi_{0}$ therefore implies $\widetilde{V}(\phi_{\ell L})>0$.
On the other hand $\widetilde{V}(p_{*})<0$ if $p_{*}<\pi_{\ell R}$.
Therefore $p_{*}<\pi_{\ell R}$ implies $p_{*}<\phi_{\ell L}$.

Now we proceed given that $p_{\ast}<\phi_{\ell L}$: By Lemma \ref{lem:p_ast_L_region}.(b),
$p_{\ast}\in(\pi_{\ell L},\pi_{\ell R})$ implies that $U(p)=U_{L0}(p;p_{\ast})$
for all $p\in[p_{\ast},\pi_{\ell R})$. Therefore, $U(p)=U_{L0}(p;p_{\ast})<\max\{U_{\ell}(p),U_{r}(p)\}$
for $p\in(p_{\ast},\phi_{\ell L})$ which means that it is not optimal
for the receiver to wait for beliefs $p\in(p_{\ast},\phi_{\ell L})$
in the waiting region. Therefore $p_{\ast}\in(\pi_{\ell L},\phi_{\ell L})$
cannot arise in equilibrium.

\paragraph*{Proposition \ref{prop:C2_fail}.}

Suppose \eqref{cond:low_receiver_rent} fails. We first rule out $p_{*}\in(\pi_{\ell R},\underline{\pi}_{LR})$.
If $p_{*}\in(\pi_{\ell R},\underline{\pi}_{LR})$, then Lemma \ref{lem:p_ast_L_region}.(c)
implies $V(p_{*})=\widetilde{V}(p_{*})>0$ and with \eqref{cond:low_receiver_rent}
failing, this implies $U(p_{\ast})>U_{\ell}(p_{\ast})$. Hence, $p_{*}\in(\pi_{\ell R},\underline{\pi}_{LR})$
leads to a violation of the above necessary condition.

Next we rule out $p_{*}\in(\pi_{\ell L},\pi_{\ell R}]$. If $p_{*}\in(\pi_{\ell L},\pi_{\ell R}]$,
then by Lemma \ref{lem:p_ast_L_region}.(b) and (c), $V(p_{*})=0$.
Hence, for any $p\in(\pi_{\ell L},p_{*})$, by the refinement, the
sender uses the $L$-drifting experiment with jumps to $p^{*}$. Since
\eqref{cond:low_receiver_rent} fails, by the same argument as in
the analysis of the waiting region in Section \ref{sec:Verifying-the-Receiver's},
the receiver prefers to wait. This contradicts $p<p_{*}$.

We have ruled out all values for $p_{*}$ except $p_{*}=\pi_{\ell L}$.
Hence, the equilibrium specified in Proposition \ref{prop:C2_fail}
is unique (up to tie breaking at $\pil_{LR}$ and $\pih_{LR}$, and
if $p^{*}=\eta$). Uniqueness follows from the Unimprovability Lemma
\ref{lem:unimprovability} since $V(p)>V_{S}(p)$ for all $p\in[p_{*},p^{*})$.

\end{document}